\documentclass[conference]{IEEEtran}
\ifCLASSINFOpdf
\else
\fi
\usepackage{amsfonts}
\usepackage{amsmath,amsthm}
\usepackage{graphicx,epstopdf,epsfig,multirow,epic,bm}
\usepackage{color}
\usepackage{multicol}
\usepackage{makeidx,showidx}
\usepackage{ifthen}
\usepackage{mathrsfs}
\usepackage{colortbl}
\usepackage{multicol}
\usepackage{latexsym}
\usepackage{amssymb}
\usepackage{algorithm}
\usepackage{algorithmic}
\usepackage{longtable}
\usepackage{array}
\usepackage{enumerate}
\usepackage{paralist}
\usepackage{fancyhdr}

\theoremstyle{plain}
\newtheorem{thm}{Theorem}
\newtheorem{cor}[thm]{Corollary}
\newtheorem{lem}[thm]{Lemma}

\theoremstyle{definition}
\newtheorem{defn}{Definition}
\theoremstyle{plain}

\theoremstyle{definition}

\theoremstyle{definition}

\theoremstyle{definition}

\newcommand{\ud}{\mathrm{deg}}
\newcommand{\qqed}{\hfill $\Box$}

\begin{document}

%
\title{Text-based Passwords Generated From Topological Graphic Passwords}


\author{\IEEEauthorblockN{Bing  Yao$^{1,5}$,~Xiaohui Zhang$^{1}$,~Hui Sun$^{1}$,~Yarong Mu$^{1}$,~Yirong Sun$^{1}$,~Xiaomin Wang$^{2}$\\ Hongyu Wang$^{2,\ddagger}$, ~Fei Ma$^{2}$,~Jing Su$^{2}$,~Chao Yang$^{3}$,~Sihua Yang$^{4}$, ~Mingjun Zhang$^{4}$}
\IEEEauthorblockA{{1} College of Mathematics and Statistics,
 Northwest Normal University, Lanzhou,  730070,  China}
\IEEEauthorblockA{{2} School of Electronics Engineering and Computer Science, Peking University, Beijing 100871, China}
\IEEEauthorblockA{{3} School of Mathematics, Physics \& Statistics, Shanghai University of Engineering Science, Shanghai, 201620, CHINA}
\IEEEauthorblockA{{4} School of Information Engineering, Lanzhou University of Finance and Economics, Lanzhou, 730030, CHINA}
\IEEEauthorblockA{{5} School of Electronics and Information Engineering, Lanzhou Jiaotong University, Lanzhou, 730070, China\\
$^\ddagger$ The corresponding author's email: why1988jy@163.com}

\thanks{Manuscript received June 1, 2017; revised August 26, 2017.
Corresponding author: Bing Yao, email: yybb918@163.com.}}


%


\maketitle

\begin{abstract}
Topological graphic passwords (Topsnut-gpws) are one of graph-type passwords, but differ from the existing graphical passwords, since Topsnut-gpws are saved in computer by algebraic matrices.  We focus on the transformation between text-based passwords (TB-paws) and Topsnut-gpws in this article. Several methods for generating TB-paws from Topsnut-gpws are introduced; these methods are based on topological structures and graph coloring/labellings, such that authentications must have two steps: one is topological structure authentication, and another is text-based authentication. Four basic topological structure authentications are introduced and many text-based authentications follow Topsnut-gpws. Our methods are based on algebraic, number theory and graph theory, many of them can be transformed into polynomial algorithms. A new type of matrices for describing Topsnut-gpws is created here, and such matrices can produce TB-paws in complex forms and longer bytes. Estimating the space of TB-paws made by Topsnut-gpws is very important for application. We propose to encrypt dynamic networks and try to face: (1) thousands of nodes and links of dynamic networks; (2) large numbers of Topsnut-gpws generated by machines rather than human's hands. As a try, we apply spanning trees of dynamic networks and graphic groups (Topsnut-groups) to approximate the solutions of these two problems. We present some unknown problems in the end of the article for further research. \\[4pt]
\end{abstract}
\textbf{\emph{Keywords---Text-based passwords; graphical password; topological graphic password; computational security; encryption.}}


%
\IEEEpeerreviewmaketitle

\pagestyle{fancy}
\pagestyle{plain}

\section{Introduction}

Graphical passwords (GPWs) are familiar with people in nowadays, such as 1-dimension code, 2-dimension code, face authentication, finger-print authentication, speaking authentication, and so on, in which 2-dimension code is widely used in everywhere of the world. A 2-dimension code can be considered as a GPW, since it is a picture. Researchers have worked on GPWs for a long time (\cite{Suo-Zhu-Owen-2005, Biddle-Chiasson-van-Oorschot-2009, Gao-Jia-Ye-Ma-2013}). Wang \emph{et al.} propose another type of graphic passwords (Topsnut-gpws) in \cite{Wang-Xu-Yao-2016} and \cite{Wang-Xu-Yao-Key-models-Lock-models-2016}, which differ from the existing GPWs.

As an example, we have two \emph{Topsnut-gpws} shown in Fig.\ref{fig:1-example}(a) and (b), where $T$ is as a \emph{public key}, $H$ is as a \emph{private key}. The \emph{authentication} in network communication is given in Fig.\ref{fig:1-example}(c). By observing Fig.\ref{fig:1-example} carefully, we can see that the labels of nodes (also, vertices) and edges of two Topsnut-gpws $T$ and $H$ form a complementary relationship, and the labels of each edge and its two nodes in $T$ and $H$ satisfies some certain mathematical restraints. Another important character of Topsnut-gpws is the \emph{configuration}, also, the \emph{topological structure} (called \emph{graph} hereafter). Thereby, we say that Topsnut-gpws are natural-inspired from mathematics of view. In general, Topsnut-gpws are easy saved in computer by algebraic matrices, and Topsnut-gpws occupy small space rather than that of the existing GPWs such that Topsnut-gpws can be implemented quickly.

Topsnut-gpw can be  as a platform for password, cipher code and encryption of information security. As Topsnut-gpws were made by ``topological configurations plus number theory'', we will apply a particular class of matrices to describe Topsnut-gpws for the purpose of writing easily in computer and running quickly by computer. These matrices are called \emph{Topsnut-matrices}, and can yield randomly text-based passwords (TB-paws for short) for authentication and encryption in communication. For the theoretical base, we will introduce some operations on Topsnut-matrices in order to implement them for building up TB-paws flexibly.

\begin{figure}[h]
\centering
\includegraphics[height=6cm]{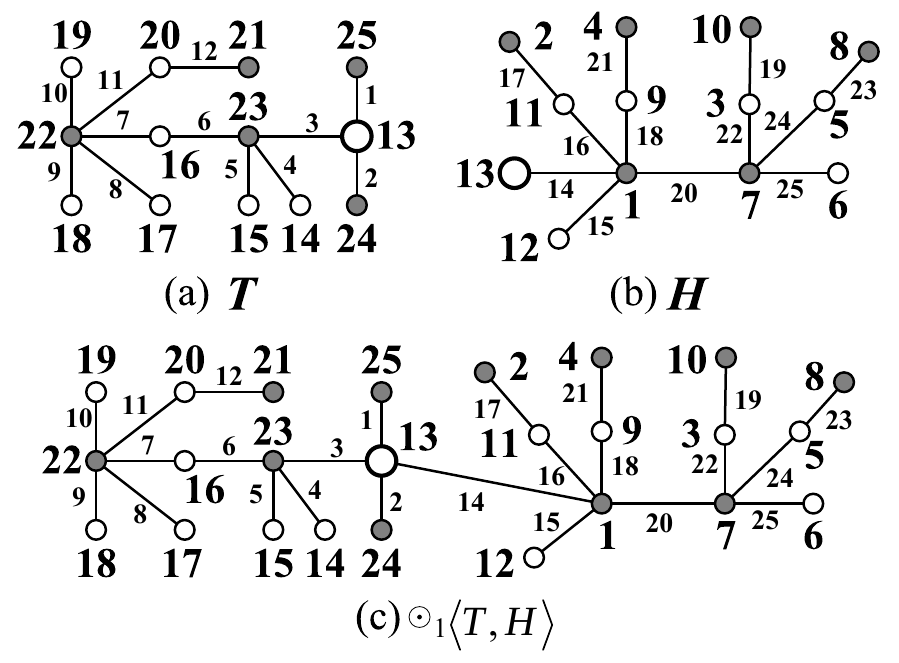}
\caption{\label{fig:1-example}{\small (a) A Topsnut-gpw as a public key; (b) a Topsnut-gpw as a private key; (c) an authentication $\odot_1\langle T,H \rangle$.}}
\end{figure}

As known, Topsnut-gpws are related with many mathematical conjectures or NP-problems, so Topsnut-gpws are \emph{computationally unbreakable} or \emph{provable security}. A Topsnut-gpw $G$ has an advantage, that is, it can generate text-based passwords with longer byte such that it is impossible to rebuild the original Topsnut-gpw $G$ from the derivative text-based passwords made by $G$. This derives us to explore the area of generating text-based passwords from Topsnut-gpws in this article. We believe this transformation from Topsnut-gpws to text-based passwords is very important for the real application of Topsnut-gpws.

\subsection{Examples and problems}

We write ``text-based passwords'' by TB-paws, and ``topological graphic passwords'' as Topsnut-gpws hereafter, for the purpose of quick statement. We will make some TB-paws from a Topsnut-gpw depicted in Fig.\ref{fig:example-1}. Along a path $P_1=y^2_{2,10}x^1_{3,1}x^1_{4,6}x^1_{1,1}x^1_{2,6}$ shown in Fig.\ref{fig:example-1}, we have a TB-paw
$$D_1=3163321891571570125125$$
obtained from the labels of vertices and edges on the path $P_1$.
\begin{figure}[h]
\centering
\includegraphics[height=5.4cm]{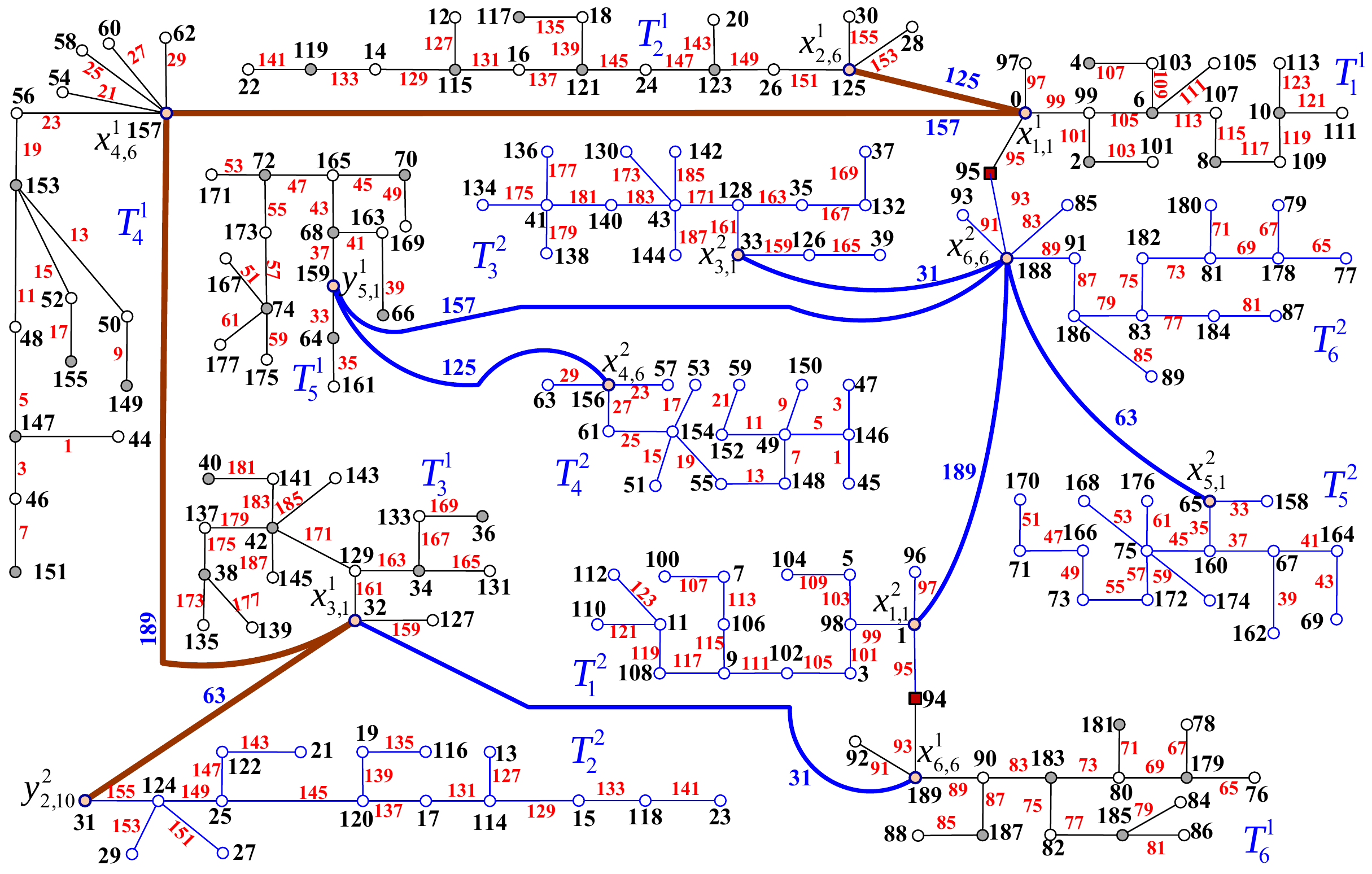}
\caption{\label{fig:example-1}{\small A Topsnut-gpw $G$ cited from \cite{Wang-Xu-Yao-2017}.}}
\end{figure}

The Topsnut-gpw $G$ depicted in Fig.\ref{fig:example-1} admits an odd-elegant labelling $f:V(G)\rightarrow [0,189]$ such that each edge $uv\in E(G)$ holds $f(uv)=f(u)+f(v)~(\bmod~190)$ to be an odd number, and $f(x)\neq f(y)$ for any pair of vertices $x,y\in V(G)$, as well as $f(uv)\neq f(st)$ for any two edges $uv$ and $st$ of $G$. By cryptography of view, the Topsnut-gpw $G$ has twelve sub-Topsnut-gpws $T^1_i$ and $T^2_i$ with $i\in [1,6]$ to form a larger \emph{authentication}, where $T^1_i$ with $i\in [1,6]$ are \emph{public keys}, and $T^2_i$ with $i\in [1,6]$ are \emph{private keys}. Moreover, a sub-Topsnut-gpw $T^1_1$ pictured in Fig.\ref{fig:example-11} distributes us a TB-paw
$${
\begin{split}
D(T^1_1)=&095950979709999101210310199\\
&105610910310746111105611310711\\
&581171091191012311310121111
\end{split}}
$$
Obviously, to reconstruct the sub-Topsnut-gpw $T^1_1$ from the TB-paw $D(T^1_1)$ is difficult, and the TB-paw $D(T^1_1)$ does not rebuild the original Topsnut-gpw $G$ at all. It means that the procedure of generating TB-paws from Topsnut-gpws is \emph{irreversible}. On the other hands, this Topsnut-gpw $G$ can distribute us $(570!)\cdot (190!)\cdot 2^{190}$ TB-paws $D(G)$ shown in the formula (\ref{eqa:wang-paper-example}), such that each TB-paw $D(G)$ has at least 380 bytes or more.

\begin{figure}[h]
\centering
\includegraphics[height=5cm]{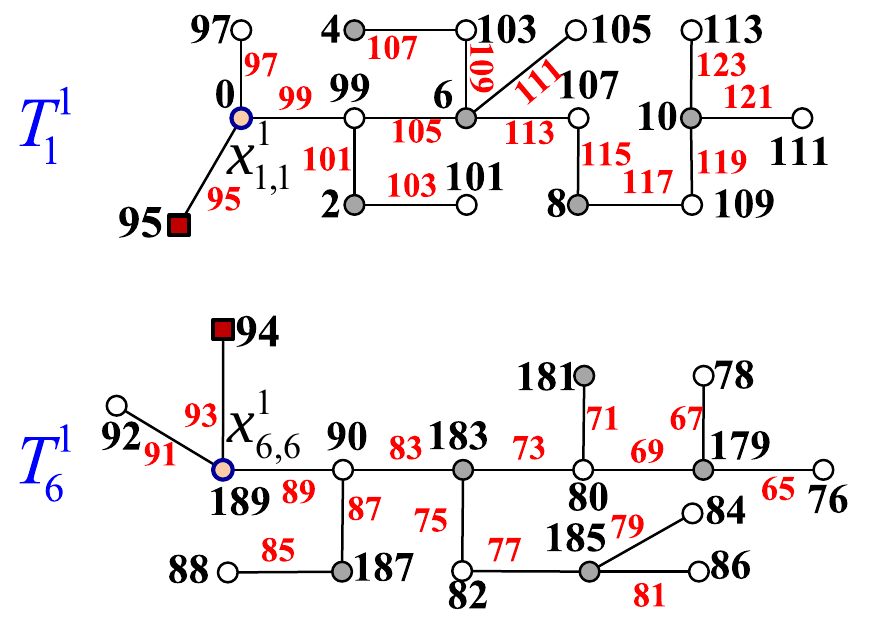}
\caption{\label{fig:example-11}{\small Two sub-Topsnut-gpws $T^1_1$ and $T^1_6$ obtained from the Topsnut-gpw $G$ shown in Fig.\ref{fig:example-1}, in which each edge $uv\in E(T^1_1)$ (or $uv\in E(T^1_6)$) holds $f(uv)=f(u)+f(v)~(\bmod~190)$.}}
\end{figure}

For the encryption of data and dynamic networks, we propose the following problems:
\begin{asparaenum}[\textbf{Problem} 1. ]
\item How to generate TB-paws from a given Topsnut-gpw?

\item How many TB-paws with the desired $k$-byte are there in a given Topsnut-gpw?
\item How to encrypt a dynamic network by Topsnut-gpws or TB-paws?
\end{asparaenum}

We will try to find some ways for answering partly the above problems in the later sections. In graph theory, Topsnut-gpws are called \emph{labelled graphs}, so both concepts of Topsnut-gpws and labelled graphs will be used indiscriminately in this article.

\subsection{Preliminary}

The following terminology, notation, labellings, particular graphs and definitions will be used in the later discussions.

\begin{asparaenum}[1) ]
\item The notation $[m,n]$ indicates a \emph{consecutive set} $\{m,m+1,\dots, n\}$ with integers $m,n$ holding $0\leq m<n$, $[a,b]^o$ denotes an \emph{odd-set} $\{a,a+2,\dots, b\}$ with odd integers $a,b$ with respect to $1\leq a< b$, and $[\alpha,\beta]^e$ is an \emph{even-set} $\{\alpha,\alpha+2,\dots, \beta\}$ with even integers $\alpha,\beta$.
\item The number of elements of a set $X$ is written as $|X|$.
\item  $N(u)$ is  the set of vertices adjacent with a vertex $u$, $\ud_G(v)=|N(v)|$ is called the \emph{degree} of the vertex $u$. If $\ud_G(u)=1$ we call $u$ a \emph{leaf}.
\item A \emph{lobster} is a tree such that the deletion of leaves of the tree results in a \emph{caterpillar}, where the deletion of leaves of a caterpillar produces just a \emph{path}.
\item A graph $G$ having $p$ vertices and $q$ edges is called a $(p,q)$-graph.
\item A \emph{spider} $S$ is a tree having paths $P_i=u_{i,1}u_{i,2}\cdots u_{i,m_i}$ with $m_i\geq 1$ and $i\in [1,n]$, its own vertex set $V(S)=\{u_0, v_k, u_{i,j}: k\in [1,m],j\in [1,m_i], i\in [1,n]\}$, such that its own edge set $E(S)=\{u_0v_k: k\in [1,m]\}\cup \{u_0u_{i,1}: i\in [1,n]\}\cup \big (\bigcup^n_{i=1}E(P_i) \big)$, and $m+n\geq 3$. Clearly, $\ud_S(u_0)\geq 3$, and $1\leq \ud_S(x)\leq 2$ for any vertex $x\in V(S)\setminus \{u_0\}$. We call $u_0$ as the \emph{body}, and each path $P_i$ is a \emph{leg} of length $m_i$ of $S$.
\item A \emph{ring-like network} $N_{ring}$ has a unique cycle $C_m$, and each vertex $u_i$ of $C_m$ is coincident with some vertex $v_i$ of a tree $T_i$ with $i\in[1,m]$.
\item The set of all subsets of a set $X$ is denoted as $X^2$, but the empty set is not allowed in $X^2$. For example, for a set $X=\{a,b,c,d\}$, then $X^2$ contains: $\{a\}$, $\{b\}$, $\{c\}$, $\{d\}$, $\{a,b\}$, $\{a,c\}$, $\{a,d\}$, $\{b,c\}$, $\{b,d\}$, $\{c,d\}$, $\{a,b,c\}$, $\{a,b,d\}$, $\{a,c,d\}$, $\{b,c,d\}$, $\{a,b,c,d\}$.
\end{asparaenum}

\begin{defn}\label{defn:define-labelling}
\cite{Yao-Sun-Zhang-Mu-Sun-Wang-Su-Zhang-Yang-Yang-2018arXiv} A \emph{labelling} $h$ of a graph $G$ is a mapping $h:S\subseteq V(G)\cup E(G)\rightarrow [a,b]$ such that $h(x)\neq h(y)$ for any pair of elements $x,y$ of $S$, and write the label set $h(S)=\{h(x): x\in S\}$. A \emph{dual labelling} $h'$ of a labelling $h$ is defined as: $h'(z)=\max h(S)+\min h(S)-h(z)$ for $z\in S$. Moreover, $h(S)$ is called the \emph{vertex label set} if $S=V(G)$, $h(S)$ the \emph{edge label set} if $S=E(G)$, and $h(S)$ a \emph{universal label set} if $S=V(G)\cup E(G)$.
\end{defn}

A combinatoric definition of set-labellings is as follows.

\begin{defn}\label{defn:set-labelling}
\cite{Yao-Sun-Zhang-Mu-Sun-Wang-Su-Zhang-Yang-Yang-2018arXiv} Let $G$ be a $(p,q)$-graph.

(i) A set mapping $F: V(G)\cup E(G)\rightarrow [0, p+q]^2$ is called a \emph{total set-labelling} of $G$ if $F(x)\neq F(y)$ for distinct elements $x,y\in V(G)\cup E(G)$.

(ii) A vertex set mapping $F: V(G) \rightarrow [0, p+q]^2$ is called a \emph{vertex set-labelling} of $G$ if $F(x)\neq F(y)$ for distinct vertices $x,y\in V(G)$.

(iii) An edge set mapping $F: E(G) \rightarrow [0, p+q]^2$ is called an \emph{edge set-labelling} of $G$ if $F(uv)\neq F(xy)$ for distinct edges $uv, xy\in E(G)$.

(iv) A vertex set mapping $F: V(G) \rightarrow [0, p+q]^2$ and a proper edge mapping $g: E(G) \rightarrow [a, b]$ are called a \emph{v-set e-proper labelling $(F,g)$} of $G$ if $F(x)\cap F(y)=\emptyset $ for distinct vertices $x,y\in V(G)$ and two edge labels $g(uv)\neq g(wz)$ for distinct edges $uv, wz\in E(G)$.

(v) An edge set mapping $F: E(G) \rightarrow [0, p+q]^2$ and a proper vertex mapping $f: V(G) \rightarrow [a,b]$ are called an \emph{e-set v-proper labelling $(F,f)$} of $G$ if $F(uv)\neq F(wz)$ for distinct edges $uv, wz\in E(G)$ and two vertex labels $f(x)\neq f(y)$ for distinct vertices $x,y\in V(G)$.\qqed
\end{defn}

\begin{defn} \label{defn:proper-bipartite-labelling-ongraphs}
(\cite{Gallian2016, Bing-Yao-Cheng-Yao-Zhao2009, Zhou-Yao-Chen-Tao2012}) Suppose that a connected $(p,q)$-graph $G$  with $1\leq p-1\leq q$ admits a mapping
$\theta:V(G)\rightarrow \{0,1,2,\dots \}$. For edges $xy\in E(G)$
the induced edge labels are defined as
$\theta(xy)=|\theta(x)-\theta(y)|$. Write
$\theta(V(G))=\{\theta(u):u\in V(G)\}$,
$\theta(E(G))=\{\theta(xy):xy\in E(G)\}$. There are the following
restrictions:
\begin{asparaenum}[(a)]
\item \label{Proper01} $|\theta(V(G))|=p$.
\item \label{Proper02} $|\theta(E(G))|=q$.
\item \label{Graceful-001} $\theta(V(G))\subseteq [0,q]$, $\min \theta(V(G))=0$.
\item \label{Odd-graceful-001} $\theta(V(G))\subset [0,2q-1]$, $\min \theta(V(G))=0$.
\item \label{Graceful-002} $\theta(E(G))=\{\theta(xy):xy\in E(G)\}=[1,q]$.
\item \label{Odd-graceful-002} $\theta(E(G))=\{\theta(xy):xy\in E(G)\}=[1,2q-1]^o$.
\item \label{Set-ordered} $G$ is a bipartite graph with the bipartition
$(X,Y)$ such that $\max\{\theta(x):x\in X\}< \min\{\theta(y):y\in
Y\}$ ($\theta(X)<\theta(Y)$ for short).
\item \label{Graceful-matching} $G$ is a tree containing a perfect matching $M$ such that
$\theta(x)+\theta(y)=q$ for each edge $xy\in M$.
\item \label{Odd-graceful-matching} $G$ is a tree having a perfect matching $M$ such that
$\theta(x)+\theta(y)=2q-1$ for each edge $xy\in M$.
\end{asparaenum}

A \emph{graceful labelling} $\theta$ holds (\ref{Proper01}),
(\ref{Graceful-001}) and (\ref{Graceful-002}) true; a \emph{set-ordered
graceful labelling} $\theta$ satisfies (\ref{Proper01}),
(\ref{Graceful-001}), (\ref{Graceful-002}) and (\ref{Set-ordered}), simultaneously;
a \emph{strongly graceful labelling} $\theta$ holds (\ref{Proper01}),
(\ref{Graceful-001}), (\ref{Graceful-002}) and
(\ref{Graceful-matching}) true; a \emph{strongly set-ordered graceful
labelling} $\theta$ complies with (\ref{Proper01}), (\ref{Graceful-001}),
(\ref{Graceful-002}), (\ref{Set-ordered}) and
(\ref{Graceful-matching}) meanwhile. An \emph{odd-graceful labelling} $\theta$ holds (\ref{Proper01}),
(\ref{Odd-graceful-001}) and (\ref{Odd-graceful-002}) true; a
\emph{set-ordered odd-graceful labelling} $\theta$ obeys
(\ref{Proper01}), (\ref{Odd-graceful-001}), (\ref{Odd-graceful-002})
and (\ref{Set-ordered}), simultaneously; a \emph{strongly odd-graceful labelling}
$\theta$ holds (\ref{Proper01}), (\ref{Odd-graceful-001}),
(\ref{Odd-graceful-002}) and (\ref{Odd-graceful-matching}) true at the same time; a
\emph{strongly set-ordered odd-graceful labelling} $\theta$ fulfils
(\ref{Proper01}), (\ref{Odd-graceful-001}),
(\ref{Odd-graceful-002}), (\ref{Set-ordered}) and
(\ref{Odd-graceful-matching}), simultaneously.\qqed
\end{defn}

Another group of definitions is about the sum of end labels of edges, we present it as follows:

\begin{defn}\label{defn:felicitous-14}
(\cite{Gallian2016, Zhou-Yao-Chen-2013}) A $(p,q)$-graph $G$ with $1\leq p-1\leq q$ admits a labelling $f:V(G) \rightarrow H$, where $H$ is an \emph{integer set}. For edges $xy\in E(G)$
the induced edge labels are defined as $f(uv)=f(u)+f(v)$ or
$f(uv)=f(u)+f(v)~(\textrm{mod}~M)$  for every edge $uv\in E(G)$. And
$f(V(G))=\{f(u):u\in V(G)\}$ is the \emph{vertex label set}, and
$f(E(G))=\{f(xy):xy\in E(G)\}$  is the \emph{edge label set}. There are the following
constraints:
\begin{asparaenum}[c-1. ]
\item \label{vertex-set-H} $f(V(G))\subseteq H$.
\item \label{vertex-set-sum-00} $f(V(G))\subseteq [0,q-1]$.
\item \label{vertex-set-sum-11} $f(V(G))\subseteq [0,q]$.
\item \label{vertex-set-sum-odd} $f(V(G))\subseteq [0,2q-1]$.
\item \label{vertex-set-sum-even} $f(V(G))\subseteq [0,2q]$.
\item \label{edge-labels} $f(uv)=f(u)+f(v)$.
\item \label{edge-labels-even-odd} $f(uv)=f(u)+f(v)$ when $f(u)+f(v)$ is even, and $f(uv)=f(u)+f(v)+1$ when $f(u)+f(v)$ is odd.
\item \label{modulo-00} $f(uv)=f(u)+f(v)~(\textrm{mod}~q)$.
\item \label{modulo-11} $f(uv)=f(u)+f(v)~(\textrm{mod}~2q)$.
\item \label{edge-set-sum}  $f(E(G))=[0,q-1]$.
\item \label{even-edge-set}  $f(E(G))=[0, 2q-2]^e$.
\item \label{even-edge-set-11}  $f(E(G))=[2, 2q]^e$.
\item \label{edge-set-odd-sum}  $f(E(G))=[1,2q-1]^o$.
\item \label{edge-set-q}  $|f(E(G))|=q$.
\item \label{sequential-edge-set}  $f(E(G))=[c,c+q-1]$.
\item \label{modulo-ordered} There exists an integer $k$ so that $\min \{f(u),f(v)\}\leq k <\max\{f(u),f(v)\}$.
\item \label{modulo-ordered-00} $G$ is bipartite with its bipartition $(X,Y)$ so that $\max f(X)<\min f(Y)$.
\end{asparaenum}

We call $f$ to be: (1) a \emph{felicitous labelling} if c-\ref{vertex-set-sum-11}, c-\ref{modulo-00} and c-\ref{edge-set-sum}  hold true; (2) an \emph{odd-elegant labelling} if c-\ref{vertex-set-sum-odd}, c-\ref{modulo-11} and c-\ref{edge-set-odd-sum}  hold true; (3) a \emph{harmonious labelling} if c-\ref{vertex-set-sum-00}, c-\ref{modulo-00} and c-\ref{edge-set-sum}  hold true, when $G$ is a tree, exactly one edge label may be used on two vertices; (4) a \emph{properly even harmonious labeling} if c-\ref{vertex-set-sum-even}, c-\ref{modulo-11} and c-\ref{even-edge-set}  hold true;
(5) a \emph{$c$-harmonious labeling} if c-\ref{vertex-set-sum-00}, c-\ref{edge-labels} and c-\ref{sequential-edge-set}  hold true; (6) an \emph{even sequential harmonious labeling} if c-\ref{vertex-set-sum-even}, c-\ref{edge-labels-even-odd} and c-\ref{even-edge-set-11}  hold true; (7) a \emph{$H$-harmonious harmonious labeling} if c-\ref{vertex-set-H}, c-\ref{edge-labels} and c-\ref{edge-set-q}  hold true; (8) a \emph{strongly harmonious labeling} if c-\ref{vertex-set-sum-11}, c-\ref{modulo-00}, \ref{modulo-ordered} and c-\ref{edge-set-sum} hold true; (9) a \emph{set-ordered harmonious labeling} if c-\ref{vertex-set-sum-11}, c-\ref{modulo-00}, c-\ref{modulo-ordered-00} and c-\ref{edge-set-sum} hold true; (10) an \emph{set-ordered odd-elegant labelling} if c-\ref{vertex-set-sum-odd}, c-\ref{modulo-11}, c-\ref{modulo-ordered-00} and c-\ref{edge-set-odd-sum}  hold true;
\end{defn}

\section{Techniques for generating TB-paws from Topsnut-gpws}

Our methods for generating TB-paws from Topsnut-gpws are mainly based on the following disciplines: Topsnut-configurations, graph-labellings, Topsnut-matrices, Topsnut-matchings and graphic groups, these are two invariable quantities of Topsnut-gpws.

\subsection{Topsnut-configurations}

By simple and clear statements, we utilize the odd-graceful/odd-graceful labellings and Topsnut-configuration to show several methods for creating TB-paws.

\subsubsection{Path-neighbor-method}

As known, each of caterpillars (see Fig.\ref{fig:caterpillar-general}) and lobsters admits an odd-graceful labelling \cite{Zhou-Yao-Chen-Tao2012}.

$H$ is a caterpillar of a $(p,q)$-graph $G$ admitting a set-ordered odd-graceful labelling $f$. So, the deletion of leaves of $H$ is just a path $P=u_1u_2\cdots u_n$ in the caterpillar $H$, such that each $u_i$ has its own leaf set $L(u_i)=\{v_{i,j}:j\in [1,m_i]\}$ with $m_i\geq 0$ and $i\in [1,n]$, and the vertex set is
$$V(H)=V(P)\cup L(u_1)\cup L(u_2)\cup \cdots \cup L(u_n)$$ See a caterpillar $T$ depicted in Fig.\ref{fig:caterpillar-general}. Thereby, we can get a vv-type TB-paw $D_{v}(P)=f(u_1)f(u_2)\cdots f(u_n)$ and a vev-type TB-paw $$D_{vev}(P)=f(u_1)f(u_1u_2)f(u_2)f(u_2u_3)\cdots f(u_{n-1}u_n)f(u_n)$$ by the \emph{path-method} for deriving two types of TB-paws from Topsnut-gpws.

\begin{figure}[h]
\centering
\includegraphics[height=2.2cm]{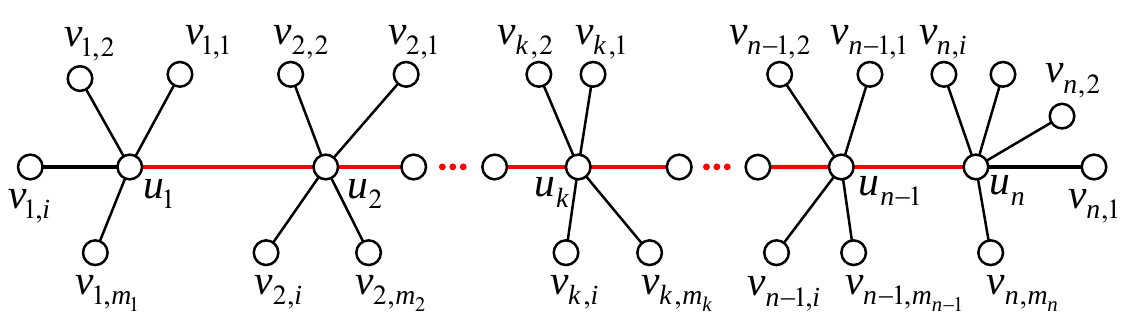}
\caption{\label{fig:caterpillar-general}{\small A general caterpillar $T$.}}
\end{figure}

From a path $Q=u_1u_2u_3u_4u_5$ revealed in Fig.\ref{fig:labelled-caterpillar}, we can get a vv-type TB-paw $D_{vv}(Q)=037102512$ and a vev-type TB-paw $D_{vev}(Q)=03737271015251312$ by the path-method.

\begin{figure}[h]
\centering
\includegraphics[height=2.6cm]{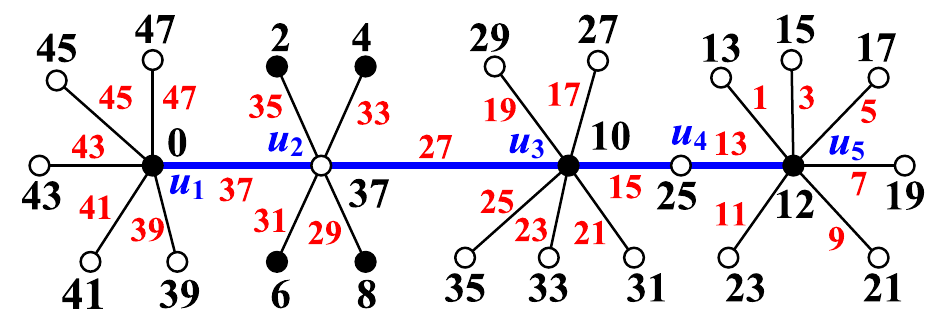}
\caption{\label{fig:labelled-caterpillar}{\small A labelled caterpillar $H$, also, a Topsnut-gpw.}}
\end{figure}

Next, we introduce the \emph{path-neighbor-method}.

Let a vertex $u$ have its neighbor set $N(u)=\{v_1,v_2,\dots ,v_{m_u}\}$ with $f(v_{i})<f(v_{i+1})$, we have a vv-type TB-paw
$$D_{vv}(u)=f(u)f(v_1)f(v_2)\cdots f(v_{m_u})f(u)$$
and
$${
\begin{split}
D_{vev}(u)=&f(u)f(uv_1)f(v_1)f(v_1v_2)f(v_2)\cdots \\
&f(v_{m_u-1}v_{m_u})f(v_{m_u})f(v_{m_u}u)f(u)
\end{split}}
$$
by the mini-principle, and moreover we get a vv-type TB-paw
$$D'_{vv}(u)=f(u)f(v_{m_u})f(v_{m_u-1})\cdots f(v_2)f(v_1)f(u)$$
and another vev-type TB-paw
$${
\begin{split}
&D'_{vev}(u)=f(u)f(uv_{m_u})f(v_{m_u})f(v_{m_u}v_{m_u-1})\\
&f(v_{m_u-1})\cdots f(v_3v_2)f(v_2)f(v_2v_1)f(v_1)f(v_1u)f(u)
\end{split}}
$$
by the maxi-principle. Let $N(v_1)=\{v_{1,1},v_{1,2},\dots ,v_{1,m_{v_1}}\}$ with $f(v_{1,j})<f(v_{1,j+1})$, where $v_1\in N(u)$. By the mini-principle, for the edge $uv_1$, we write a vv-type TB-paw
\begin{equation}\label{eqa:path-neighbor-method-1}
{
\begin{split}
D_{vv}(uv_1)=&f(u)f(v_1)f(v_2)\cdots f(v_{m_u})f(u)\\
&f(v_1)f(v_{1,1})f(v_{1,2})\cdots f(v_{1,m_{v_1}})f(v_1)
\end{split}}
\end{equation}
by the mini-principle, denoted as
\begin{equation}\label{eqa:formula-1}
D_{vv}(uv_1)=D_{vv}(u)\uplus D_{vv}(v_1),
\end{equation}
and moreover we can write a vev-type TB-paw
\begin{equation}\label{eqa:path-neighbor-method-2}
{
\begin{split}
&D_{vev}(uv_1)=f(u)f(uv_1)f(v_1)f(v_1v_2)f(v_2)\cdots \\
&f(v_{m_u-1})f(v_{m_u-1}v_{m_u})f(v_{m_u})f(v_{m_u}u)f(u)\\
&f(v_1)f(v_1v_{1,1})f(v_{1,1})f(v_{1,1}v_{1,2})f(v_{1,2})\cdots \\
&f(v_{1,m_{v_1}-1}v_{1,m_{v_1}})f(v_{1,m_{v_1}})f(v_{1,m_{v_1}}v_1)f(v_1)
\end{split}}
\end{equation}
by the mini-principle, denoted as
\begin{equation}\label{eqa:formula-2}
D_{vev}(uv_1)=D_{vev}(u)\uplus D_{vev}(v_1).
\end{equation}
Similarly with (\ref{eqa:formula-1}) and (\ref{eqa:formula-2}), we can write $D'_{vv}(uv_1)=D'_{vv}(u)\uplus D'_{vv}(v_1)$ and $D'_{vev}(uv_1)=D'_{vev}(u)\uplus D'_{vev}(v_1)$ by the maxi-principle.

For example, by means of a caterpillar $H$ exhibited in Fig.\ref{fig:labelled-caterpillar} and two formulae (\ref{eqa:path-neighbor-method-1}) and (\ref{eqa:path-neighbor-method-2}), we have two vv-type TB-paws
$${
\begin{split}
D_{vv}(H)=&0373941434547037246810371025272\\
&9313335102510122512131517192123
\end{split}}
$$
$${
\begin{split}
D'_{vv}(H)=&047454341393703710864237103531\\
&292725102512102512232119171513
\end{split}}
$$
according to the mini-principle and the maxi-principle. Similarly,
$${
\begin{split}
D_{vev}(H)=&037373939414143434545474703\\
&73523343162982710371015251727\\
&19292131233325351025151013122\\
&5121133155177199211123
\end{split}}$$
is a vev-type TB-paw by the mini-principle, and moreover,
$${
\begin{split}
D'_{vev}(H)=&04747454543434141393937370372\\
&7102983163343523710273725352333\\
&21311929172715251025131215102\\
&5121123921719517315113
\end{split}}
$$
is obtained by the maxi-principle.

It is easy to see that there are many ways to generate vv-type/vev-type TB-paws from a Topsnut-gpw made by a labelled caterpillar, except the mini-principle and the maxi-principle. In a vv-type/vev-type TB-paw $f(u_i)D(N(u_i))f(u_i)$, we say $head=f(u_i)$, $f(u_i)=tail$, and $D(N(u_i))=body$ in the vv-type/vev-type TB-paw $f(u_i)D(N(u_i))f(u_i)$. So, we have $(m_i)!$ permutations for writing $D(N(u_i))$, and a caterpillar with the path $u_1u_2\cdots u_n$ distributes us $\prod^n_{i=1}(m_i)!$ vv-type/vev-type TB-paws at least.

\vskip 0.2cm

\subsubsection{Cycle-neighbor-method} By a caterpillar $T$ depicted in Fig.\ref{fig:caterpillar-general}, we add an edge $u_1u_n$ to $T$ for joining the vertex $u_1$ with $u_n$, the resulting graph is denoted as $T'=T+u_1u_n$, in which there is a cycle $C=u_1u_2\cdots u_nu_1$. So, we have a vv-type TB-paw
\begin{equation}\label{eqa:Cycle-neighbor-method-1}
{
\begin{split}
D_{vv}(T')&=D_{vv}(u_1)\uplus D_{vv}(u_2)\uplus \cdots \\
&\quad \uplus D_{vv}(u_n)\uplus D_{vv}(u_1)\\
&=\left [\uplus ^n_{k=1}D_{vv}(u_k)\right ]\uplus D_{vv}(u_1)
\end{split}}
\end{equation}
along the cycle $C$, and a vev-type TB-paw
\begin{equation}\label{eqa:Cycle-neighbor-method-2}
{
\begin{split}
D_{vev}(T')&=D_{vev}(u_1)\uplus D_{vev}(u_2)\uplus \cdots \\
&\quad \uplus D_{vev}(u_n)\uplus D_{vev}(u_1)\\
&=\left [\uplus ^n_{k=1}D_{vev}(u_k)\right ]\uplus D_{vev}(u_1).
\end{split}}
\end{equation}

Since we have $n$ initial vertices of the cycle $C=u_1u_2\cdots u_nu_1$, so the number of vv-type/vev-type TB-paws distributed from $C$ is equal to \begin{equation}\label{eqa:TB-paws-by-cycle}
N_{tbp}(C)=n\cdot (m_1+1)!\cdot (m_n+1)!\cdot \prod^{n-1}_{i=2}(m_i)!.
\end{equation}

\vskip 0.2cm

\subsubsection{Lobster-neighbor-method} In \cite{Zhou-Yao-Chen-Tao2012} and \cite{Zhou-Yao-Chen-2013}, the authors have proven: \emph{Each lobster admits one of odd-graceful labelling and odd-elegant labelling}. Thereby, we can apply lobsters to make Topsnut-gpws, or we select sub-Topsnut-gpws being lobsters of Topsnut-gpws to derive vv-type/vev-type TB-paws. Another advantage about lobsters is helpful for us to produce random Topsnut-gpws that generate random vv-type/vev-type TB-paws.

Recall, a lobster is defined as a tree $T$ such that the deletion of leaves of $T$ results in a caterpillar, that is, the remainder $T-L(T)$ is just a caterpillar, where $L(T)$ is the set of all leaves of $T$. In other words, each lobster can be constructed by adding leaves to some caterpillar. The results in \cite{Zhou-Yao-Chen-Tao2012} and \cite{Zhou-Yao-Chen-2013} enable us to build up lobsters admitting odd-graceful/odd-elegant labellings by caterpillars admitting set-ordered odd-graceful/odd-elegant labellings through adding leaves.

We show an example for illustrating ``adding leaves to a caterpillar admitting a set-ordered odd-graceful labelling produces a lobster admitting an odd-graceful labelling''. Based on a caterpillar $H$, as revealed in Fig.\ref{fig:labelled-caterpillar}, we can see that Fig.\ref{fig:add-leaves-1} gives the procedure of ``adding randomly leaves to $H$'', and the labelling new edges is shown in Fig.\ref{fig:add-leaves-2}, and moreover the procedure of ``labelling new vertices and relabelling old vertices'' presents the desired odd-graceful lobster (see Fig.\ref{fig:add-leaves-3}).

\begin{figure}[h]
\centering
\includegraphics[height=3cm]{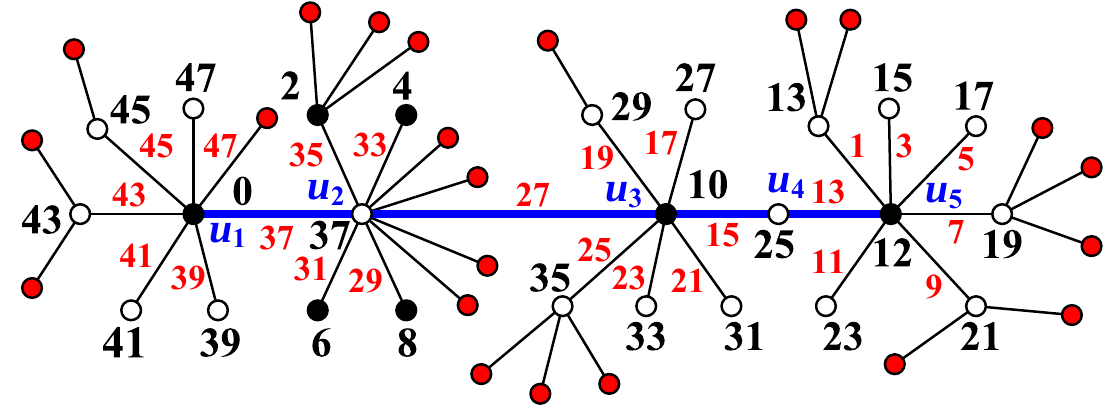}
\caption{\label{fig:add-leaves-1}{\small Adding leaves (with red vertices) randomly to a caterpillar $H$ exhibited in Fig.\ref{fig:labelled-caterpillar} for producing a lobster.}}
\end{figure}

\begin{figure}[h]
\centering
\includegraphics[height=2.8cm]{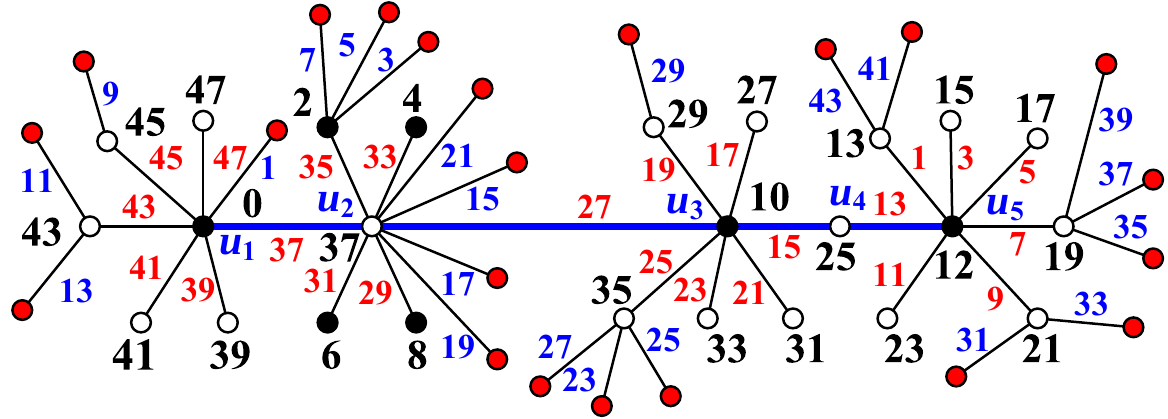}
\caption{\label{fig:add-leaves-2}{\small Labelling new edges.}}
\end{figure}

\begin{figure}[h]
\centering
\includegraphics[height=3cm]{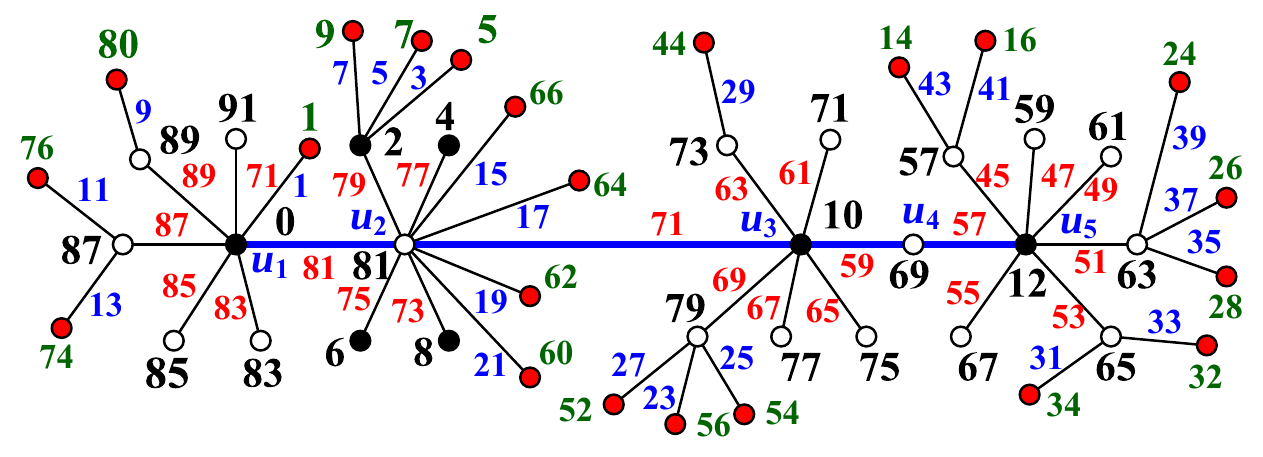}
\caption{\label{fig:add-leaves-3}{\small An odd-graceful lobster obtained by labelling new vertices and relabelling old vertices and old edges.}}
\end{figure}

We, now, come to introduce the \emph{lobster-neighbor-method} for getting vv-type/vev-type TB-paws from a Topsnut-gpw made by an odd-graceful lobster in the following algorithm.

\begin{thm}\label{thm:lobster-neighbor-method}
There exists an efficient and polynomial algorithm (LOBSTER-algorithm) for generating vv-type/vev-type TB-paws from Topsnut-gpws made by odd-graceful lobsters.
\end{thm}
\begin{proof} We, directly, use an algorithmic proof here for generating vv-type/vev-type TB-paws from Topsnut-gpws.

\textbf{Step 1.} Suppose that a lobster $T$ corresponds to a caterpillar $H$ obtained by deleting some leaves from $T$. Write $L'(T)$ as the set of deleted leaves, so $H=T-L'(T)$. Conversely, $T$ is obtained by adding the leaves of $L'(T)$ to $H$. Let $P=u_1u_2\cdots u_n$ be the path as the remainder after the deletion of leaves of the caterpillar $H$, and let $g$ be a set-ordered odd-graceful labelling of $H$. Thereby, we have
$$D_{vv}(u_i)=g(u_i)g(v_{i,1})g(v_{i,2})\cdots g(v_{i,m_i})g(u_i)$$
with $v_{i,j}\in N(u_i)=\{v_{i,j}:j\in [1,m_i]\}$
and
$${
\begin{split}
D_{vev}(u_i)=&g(u_i)g(u_iv_{i,1})g(v_{i,1})g(v_{i,1}v_{i,2})g(v_{i,2})\\
&\cdots g(v_{i,m_i-1}v_{i,m_i})g(v_{i,m_i})g(v_{i,m_i}u_i)g(u_i)
\end{split}}
$$
with $v_{i,j}\in N(u_i)=\{v_{i,j}:j\in [1,m_i]\}$. Thereby, we have a vv-type TB-paw
\begin{equation}\label{eqa:lobster-neighbor-method-1}
{
\begin{split}
D_{vv}(H)&=D_{vv}(u_1)\uplus D_{vv}(u_2)\uplus \cdots \uplus D_{vv}(u_n)\\
&=\uplus ^n_{k=1}D_{vv}(u_k)
\end{split}}
\end{equation}
and a vev-type TB-paw
\begin{equation}\label{eqa:lobster-neighbor-method-2}
{
\begin{split}
D_{vev}(H)&=D_{vev}(u_1)\uplus D_{vev}(u_2)\uplus \cdots \uplus D_{vev}(u_n)\\
&=\uplus ^n_{k=1}D_{vev}(u_k).
\end{split}}
\end{equation}

\textbf{Step 2.} Adding randomly leaves to $H$ for forming a lobster $T$. Since $g$ is a set-ordered odd-graceful labelling of the caterpillar $H$, so $V(H)=X\cup Y$ with $X\cap Y=\emptyset $, and any edge $xy$ of $H$ holds $x\in X$ and $y\in Y$ such that $\max g(X)<\min g(Y)$. By the hypothesis above, we can write $X=\{x_1,x_2,\dots x_s\}$ with $g(x_i)<g(x_{i+1})$ for $i\in [1,s-1]$, and $Y=\{y_1,y_2,\dots y_t\}$ with $g(y_j)<g(y_{j+1})$ for $j\in [1,t-1]$. Suppose that each vertex $x_i$ is added leaves from the set $L(x_i)=\{u_{i,j}:j\in [1,a_i]\}$ with $i\in [1,s]$, and each vertex $y_j$ is added leaves from the set $L(y_j)=\{w_{j,k}:k\in [1,b_j]\}$ with $k\in [1,t]$. Here, it is allowed some $a_i=0$ or $b_j=0$. The resulting tree is just $T$. Therefore, $$|E(T)|=|E(H)|+\sum^s_{i=1} a_i+\sum^t_{j=1} b_j,$$ and write $M=|E(T)|-|E(H)|$.

We define a labelling $f$ for $T$ in the following steps.

\textbf{Substep 2.1.} We label the edges $x_iu_{i,j}$ of $T$ in the increasing order: $f(x_1u_{1,j})=2j-1$ for $j\in [1,a_1]$, $f(x_2u_{2,j})=2j+f(x_1u_{1,a_1})$ for $j\in [1,a_2]$, and
$$f(x_ku_{k,j})=2j+\sum ^{k-1}_{i=1}f(x_iu_{i,a_i}),~ j\in [1,a_k]$$
with $k\in [2,s]$. Thus, $f(x_su_{2,a_s})=(2\sum ^{s}_{i=1}a_i)-1$.

\textbf{Substep 2.2.} For the edges $y_jw_{j,k}$, we set in the decreasing order: $f(y_tw_{t,k})=2k+f(x_su_{2,a_s})$ with $k\in [1,b_t]$, $f(y_{t-1}w_{t-1,k})=2k+f(y_tw_{t,k})$ with $k\in [1,b_{t-1}]$, and
$$f(y_{t-j}w_{t-j,k})=2k+\sum ^{j}_{i=1}f(y_{t-i+1}w_{t-i+1,b_{t-i+1}}),$$
with $k\in [1,b_{t-j}]$ and $j\in [1,t-1]$.

\textbf{Substep 2.3.} We come to label the vertices of $T$ in the following way: $f(x)=g(x)$ for $x\in X$; $f(u_{i,j})=f(x_iu_{i,j})-f(u_{i})$ for $u_{i,j}\in L(x_i)$ with $i\in [1,s]$; $f(y)=g(y)+2M$ for $y\in Y$; and $f(w_{k,j})=f(y_{k})-f(y_kw_{k,j})$ for $w_{k,j}\in L(y_k)$ with $k\in [1,t]$

\textbf{Step 3.} Producing a vv-type TB-paw and a vev-type TB-paw from the lobster $T$. We use the notation $L^*$ to denote the set of new leaves added to $H$ hereafter. We set $D_{vv}(L^*(u_i))=f(\alpha_{i,1})f(\alpha_{i,2})\cdots f(\alpha_{i,c_i})f(u_i)$ for $\alpha_{i,j}\in L^*(u_i)$ with $i\in [1,n]$, and get a vv-type sub-TB-paw
$${
\begin{split}
&D'_{vv}(u_i)=f(u_i)\uplus D_{vv}(L^*(u_i))\uplus f(v_{i,1})\uplus D_{vv}(L^*(v_{i,1}))\\
&\quad \uplus f(v_{i,2})\uplus D_{vv}(L^*(v_{i,2}))\uplus \cdots \uplus D_{vv}(L^*(v_{i,m_i-1}))\\
&\quad \uplus f(v_{i,m_i})\uplus D_{vv}(L^*(v_{i,m_i}))\uplus f(u_i)
\end{split}}
$$
with $u_i\in V(P)$, where $L^*(v_{i,j})$ is the set of new leaves added to $v_{i,j}$ and $D_{vv}(L^*(v_{i,j}))=f(\beta_{i,1})f(\beta_{i,2})\cdots f(\beta_{i,c_i})f(v_{i,j})$ for $\beta_{i,j}\in L^*(v_{i,j})$. Hence, we get the desired vv-type TB-paw
\begin{equation}\label{eqa:lobster-neighbor-method-11}
{
\begin{split}
D_{vv}(T)&=D'_{vv}(u_1)\uplus D'_{vv}(u_2)\uplus \cdots \uplus D'_{vv}(u_n)\\
&=\uplus ^n_{k=1}D'_{vv}(u_k)
\end{split}}
\end{equation}

Next, for getting a vev-type TB-paw from the lobster $T$, we take
$${
\begin{split}
D_{vev}(L^*(u_i))=&f(u_i\alpha_{i,1})f(\alpha_{i,1})f(\alpha_{i,1}\alpha_{i,2})f(\alpha_{i,2})\\
&\cdots f(\alpha_{i,c_i})f(\alpha_{i,c_i}u_i)f(u_i)
\end{split}}
$$ for $\alpha_{i,j}\in L^*(u_i)$, and moreover
$${
\begin{split}
&\quad D_{vev}(L^*(v_{i,j}))=f(v_{i,j}\beta_{i,1})f(\beta_{i,1})f(\beta_{i,1}\beta_{i,2})f(\beta_{i,2})\\
&\cdots f(\beta_{i,c_i-1}\beta_{i,c_i})f(\beta_{i,c_i})f(\beta_{i,c_i}v_{i,j})f(v_{i,j})
\end{split}}
$$ for $\beta_{i,j}\in L^*(v_{i,j})$. So,
$${
\begin{split}
&\quad D'_{vev}(u_i)=f(u_i)\uplus D_{vev}(L^*(u_i))\uplus f(u_iv_{i,1})\\
&f(v_{i,1})\uplus D_{vev}(L^*(v_{i,1}))\uplus f(v_{i,1}v_{i,2})f(v_{i,2})\\
&\uplus D_{vev}(L^*(v_{i,2}))\uplus f(v_{i,2}v_{i,3})f(v_{i,3})\uplus \cdots \\
&\uplus f(v_{i,m_i-1}v_{i,m_i})f(v_{i,m_i})\uplus D_{vev}(L^*(v_{i,m_i}))\\
&\uplus f(v_{i,m_i}u_i)f(u_i)
\end{split}}
$$ with $u_i\in V(P)$. Thereby, the lobster $T$ distributes a vev-type TB-paw as follows
\begin{equation}\label{eqa:lobster-neighbor-method-22}
{
\begin{split}
D_{vev}(T)&=D'_{vev}(u_1)\uplus D'_{vev}(u_2)\uplus \cdots \uplus D'_{vev}(u_n)\\
&=\uplus ^n_{k=1}D'_{vev}(u_k).
\end{split}}
\end{equation}

Since the above algorithm is constructive, so we claim that our algorithm are polynomial and efficient. The proof of the theorem is complete.\end{proof}

\vskip 0.2cm

We estimate the space of Topsnut-gpws made by labelled lobsters. Since adding $m$ leaves to a $(p,q)$-caterpillar $T$ admitting a labelling $f$ producing lobsters, we assume these $m$ leaves are added to $k$ vertices of $T$ with $1\leq k\leq m$.

We select $k$ vertices from $T$ for adding $m$ leaves to them, then we have $A^k_p=p(p-1)\cdots (p-k+1)$ selections, rather than ${p \choose k}=\frac{p!}{k!(p-k)!}$. Next, we decompose $m$ into a group of $k$ parts $m_1,m_2,\cdots ,m_k$ holding $m=m_1+m_2+\cdots +m_k$ with $m_i\neq 0$. Suppose there is $P(m,k)$ groups of such $k$ parts. For a group of $k$ parts $m_1,m_2,\cdots ,m_k$, let $m_{i_1},m_{i_2},\cdots ,m_{i_k}$ be a permutation of the group $\{m_1,m_2,\cdots ,m_k\}$, so we have the number of such permutations is a factorial $k!$. Since the $(p,q)$-caterpillar $T$ is labelled well by the odd-graceful labelling $f$, then we have
\begin{equation}\label{eqa:c3xxxxx}
A_{leaf}(T,m)=\sum^m_{k=1}A^k_p\cdot P(m,k)\cdot k!=\sum^m_{k=1} P(m,k)\cdot p!
\end{equation} to be the number of lobsters made by adding $m$ leaves to $T$, where $P(m,k)=\sum ^k_{r=1}P(m-k,r)$.  Here, computing $P(m,k)$ can be transformed into finding the number $A(m,k)$ of solutions of equation $m=\sum ^k_{i=1}ix_i$. There is a recursive formula
\begin{equation}\label{eqa:c3xxxxx}
A(m,k)=A(m,k-1)+A(m-k,k)
\end{equation}
with $0 \leq k\leq m$. It is not easy to compute the exact value of $A(m,k)$, for example,
$${
\begin{split}
&\quad A(m,6)=\biggr\lfloor \frac{1}{1036800}(12m^5 +270m^4+\\
&+1520m^3-1350m^2-19190m-9081)+\\
&\frac{(-1)^m(m^2+9m+7)}{768}+\frac{1}{81}\left[(m+5)\cos \frac{2m\pi}{3}\right ]\biggr\rfloor .
\end{split}}$$

Finally, let $N(\textrm{caterpillar},p)$ be the number of caterpillars of $p$ vertices, and let a caterpillar $T$ of $p$ vertices has $N(\textrm{odd-graceful})$ set-ordered odd-graceful labellings. Then all of caterpillars of $p$ vertices give us  at least $N(\textrm{caterpillar},p)\cdot N(\textrm{odd-graceful})\cdot A_{leaf}(T,m)$ lobsters having odd-graceful labellings.

\vskip 0.2cm

\subsubsection{Spider-neighbor-method}
Spiders are interesting graphic configurations since they are useful in networks (see Fig.\ref{fig:spider-1}). A spider $S$ has its body $u_0$ joining leaves $v_1, v_2, \dots ,v_k$ and joining legs $P_i=u_{i,1}u_{i,2}\cdots u_{i,m_i}$ with $m_i\geq 1$ and $i\in [1,n]$. Suppose that $S$ admits an odd-graceful labelling $h$ such that $h(v_i)<h(v_{i+1})$, and $h(u_{j,1})<h(u_{j+1,1})$.

We take its body $u_0$ as the beginning of vv-type TB-paws or vev-type TB-paws, and then have
\begin{equation}\label{eqa:Spider-neighbor-method-1}
{
\begin{split}
D_{vv}(S)=&h(u_0)h(v_1)\cdots h(v_k)h(u_0)\uplus D_{vv}(P_1)\\
&\uplus h(u_0)\uplus D_{vv}(P_2)\uplus h(u_0)\uplus D_{vv}(P_3)\\
&\uplus \cdots \uplus D_{vv}(P_n)
\end{split}}
\end{equation} by the mini-principle, where $$D_{vv}(P_j)=h(u_{j,1})h(u_{j,2})\cdots h(u_{j,m_j})$$ with $i\in [1,n]$. Next, we write
$${
\begin{split}
D_{vev}(P_j)=&h(u_{j,1})h(u_{j,1}u_{j,2})h(u_{j,2})h(u_{j,2}u_{j,3})\\
&h(u_{j,3})\cdots h(u_{j,m_j-1}u_{j,m_j})h(u_{j,m_j})
\end{split}}
$$ by the mini-principle, and furthermore
$${
\begin{split}
D_{vev}(u_0)=&h(u_0)h(u_0v_1)h(v_1)h(u_0v_2)\\
&h(v_2)\cdots h(u_0v_k)h(v_k).
\end{split}}
$$ Then we have a vev-type TB-paw
\begin{equation}\label{eqa:Spider-neighbor-method-2}
{
\begin{split}
D_{vev}(S)=&D_{vev}(u_0)\uplus h(u_0)h(u_0u_{1,1})\uplus D_{vev}(P_1)\\
&\uplus h(u_0)h(u_0u_{2,1})\uplus D_{vev}(P_2) \cdots\\
&\uplus h(u_0)h(u_0u_{n,1})\uplus D_{vev}(P_n)
\end{split}}
\end{equation} by the mini-principle. In fact, we can arrange $h(v_1),h(v_2),\dots h(v_k)$ and $h(u_0)h(u_0u_{i,1})D(P_i)$ with $i\in [1,n]$ into $(k+n)$ permutations to form many $D_{vev}(S)$ in the form (\ref{eqa:Spider-neighbor-method-2}). In other words, the number of vv-type/vev-type TB-paws generated by a spider is at least $(k+n)!$.

If a spider $S$ admits a set-ordered odd-graceful labelling, then we can add new leaves to $S$ such that the resulting tree $S^*$ admits an odd-graceful labelling. This tree $S^*$ is called a \emph{haired-spider} (or \emph{super spider}, see Fig.\ref{fig:spider-2}). Since adding leaves randomly, haired-spiders produce random vv-type TB-paws or random vev-type TB-paws.

If a spider $S$ is a subgraph of a $(p,q)$-graph $G$ admitting an $\varepsilon$-labelling $f$, then $S$ admits a labelling $f^*$ induced by $f$. We can use $S$ to generate vv-type/vev-type TB-paws such that this procedure is \emph{irreversible}.

If a spider $S$ admits an \emph{edge-magic proper total coloring} (\cite{Yao-Sun-Zhang-Mu-Sun-Wang-Su-Zhang-Yang-Yang-2018arXiv}) (see for an example depicted in Fig.\ref{fig:spider-1}), then we can add randomly leaves to $S$ (as a \emph{public key}) for generating \emph{super spiders} (as \emph{private keys}) admitting edge-magic proper total colorings (see Fig.\ref{fig:spider-2}). A \emph{super spider} $G$ generates $A_{super}(G)$ vv-type/vev-type TB-paws, where $A_{super}(G)$ can be computed by the numbers $(n+k)!$ and $\prod^n_{i=1}(m_i)!$, see the formula (\ref{eqa:TB-paws-by-cycle}).

\begin{figure}[h]
\centering
\includegraphics[height=3.4cm]{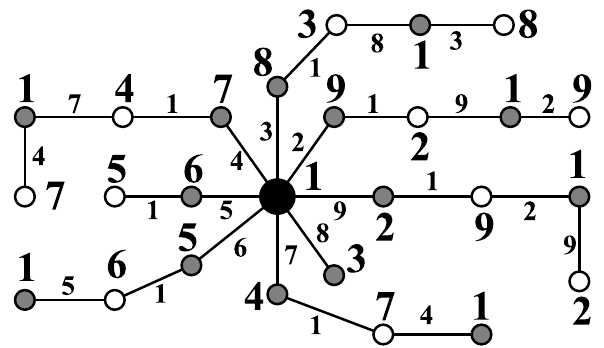}
\caption{\label{fig:spider-1}{\small A spider admits an edge-magic proper total coloring $f$ such that $f(u)+f(uv)+f(v)=12$ for each edge $uv$.}}
\end{figure}

\begin{figure}[h]
\centering
\includegraphics[height=5cm]{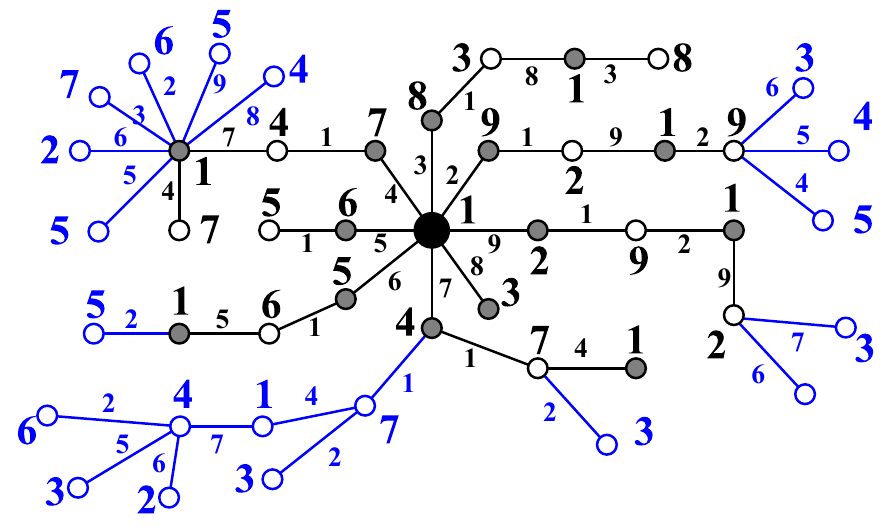}
\caption{\label{fig:spider-2}{\small A super spider $S^*$ admits an edge-magic proper total coloring, where $S^*$ is based on Fig.\ref{fig:spider-1}.}}
\end{figure}

\vskip 0.2cm

\subsubsection{Euler-Hamilton-method} Euler's graphs and Hamilton cycles are popular in graph theory. Sun \emph{et al.} \cite{Sun-Zhang-Yao-IAEAC-2017} show a connection between Euler's graphs and Hamilton cycle by an operation, called \emph{non-adjacent identifying operation and the 2-edge-connected 2-degree-vertex splitting operation}. An example is shown in Fig.\ref{fig:Euler-Hamilton}.

\begin{figure}[h]
\centering
\includegraphics[height=3.4cm]{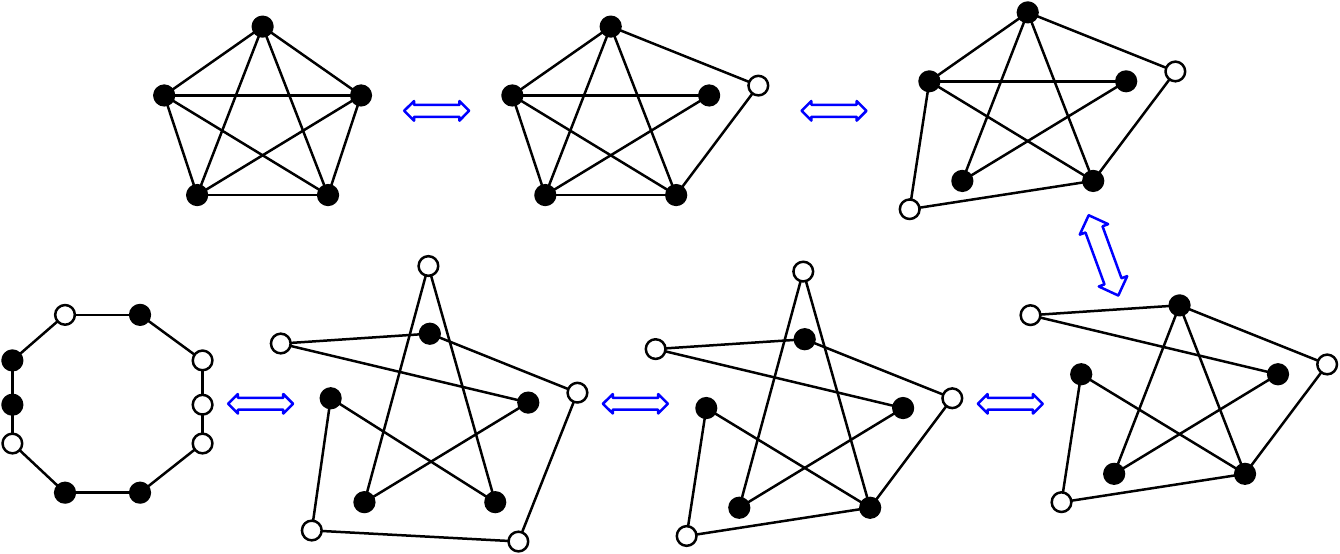}
\caption{\label{fig:Euler-Hamilton}{\small A procedure of connecting an Euler's graph and a Hamilton cycle of length 10 by the non-adjacent identifying operation and the 2-edge-connected 2-degree-vertex splitting operation introduced in \cite{Sun-Zhang-Yao-IAEAC-2017}.}}
\end{figure}

Sun \emph{et al.} \cite{Hui-Sun-Bing-Yao2018} investigate some v-set e-proper $\varepsilon$-labelling on Euler's graphs, where $\varepsilon \in \{$graceful, odd-graceful, harmonious, $k$-graceful, odd sequential, elegant, odd-elegant, felicitous, odd-harmonious, edge-magic total$\}$. (see examples displayed in Fig.\ref{fig:v-set-e-proper})

\begin{figure}[h]
\centering
\includegraphics[height=6cm]{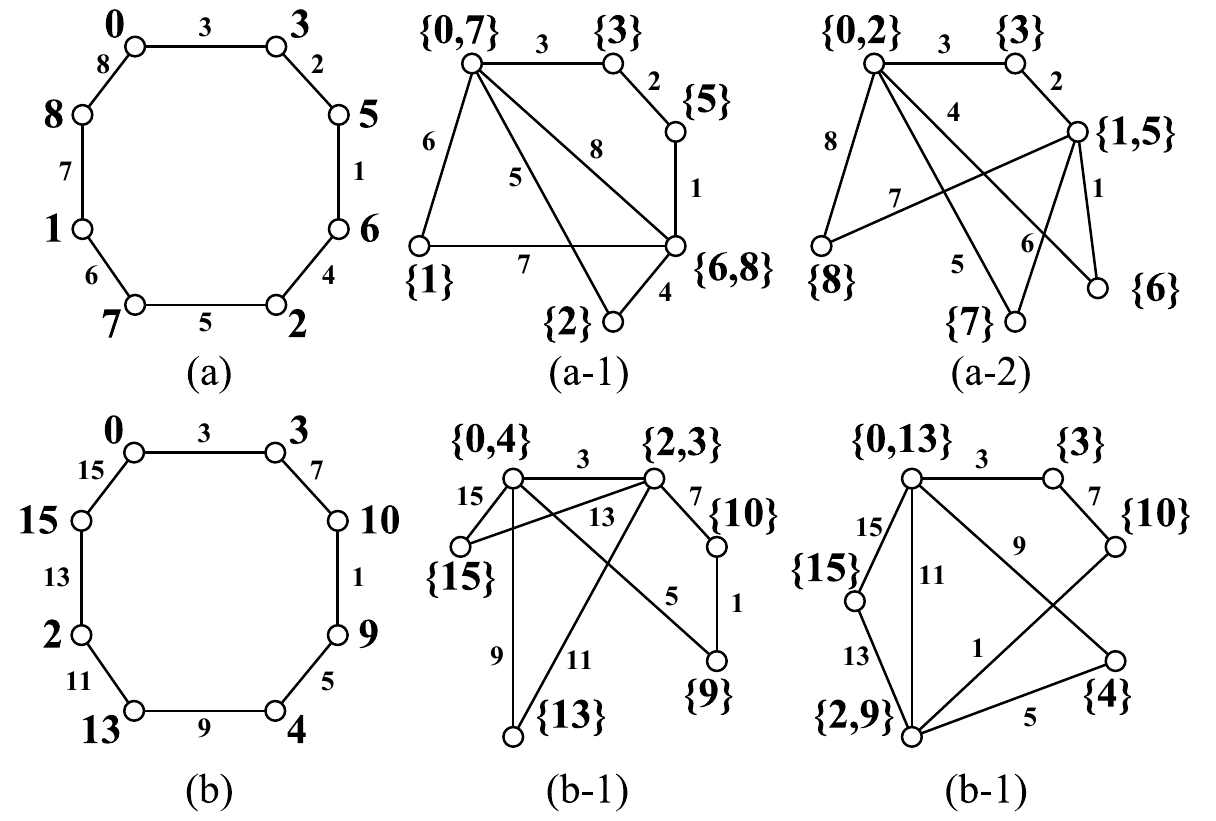}
\caption{\label{fig:v-set-e-proper}{\small (a) A cycle admitting a graceful labelling; (a-1) and (a-2) an Euler's graph admitting two v-set e-proper graceful labellings; (b) a cycle admitting an odd-graceful labelling; (b-1) and (b-2) an Euler's graph admitting two v-set e-proper odd-graceful labellings.}}
\end{figure}

It is easy to generate vv-type/vev-type TB-paws from Topsnut-gpws made by Hamilton cycles. We use an example to introduce the Euler-Hamilton-method in the following: We make a vev-type TB-paw $D_{vev}(C_8)=03325164257617880$ form Fig.\ref{fig:v-set-e-proper}(a), thus, $D_{vev}(C_8)$ can be obtained from Fig.\ref{fig:v-set-e-proper}(a-1) too, and vice versa. Obviously, it is not relaxed to pick up $D_{vev}(C_8)$ from Fig.\ref{fig:v-set-e-proper}(a-1) if Topsnut-gpws have large number vertices and edges. Thereby, a vev-type TB-paw (as \emph{a public key}) made by a labelled Hamilton cycle induces directly a vev-type TB-paw (as \emph{a private key}) generated from a labelled Euler's graph. But, such vv-type/vev-type TB-paws can be attacked since labelled Hamilton cycles are easy to be found by compute attack.

We can let the Euler-Hamilton-method to produce complex vv-type/vev-type TB-paws in the following way: As known, a graph $G$ is an Euler's graph if and only if there are $m$ edge-disjoint cycles $C_1,C_2,\dots, C_m$ such that $E(G)=\bigcup ^m_{i=1}E(C_i)$ (\cite{Bondy-2008}). So, we can get $m$ vev-type TB-paws $D_{vev}(C_i)$ with $i\in [1,m]$. Let $i_1,i_2,\dots ,i_m$ be a permutation of $1,2,\dots ,m$. Thereby, we have many vev-type TB-paws like
$$D_{vev}(G)=D_{vev}(C_{i_1})\uplus D_{vev}(C_{i_2})\uplus \cdots \uplus D_{vev}(C_{i_m}),$$
or
$$D'_{vev}(G)=D(C_{i_1})\uplus D(C_{i_2})\uplus \cdots \uplus D(C_{i_j}),~j<m.$$

Clearly, it is an \emph{irreversible} procedure of generating $D'_{vev}(G)$ from Euler's graphs. We can give a character of a non-Euler's graph: Each non-Euler's graph $G$ corresponds $m$ disjoint paths $P_1,P_2,\dots ,P_m$ such that $|E(G)|=\sum ^m_{i=1}|E(P_i)|$. In fact, any non-Euler's graph $G$ can be add a set $E^*$ of $m$ new edges such that the resulting graph $G+E^*$ is just an Euler's graph, and $G+E^*$ corresponds a Hamilton cycle $C_q$ with $q=|E(G+E^*)|$ \cite{Sun-Zhang-Yao-IAEAC-2017}. Now, we delete all edges of $E^*$ from $C_q$, so $C_q-E^*$ is just a graph consisted of $m$ disjoint paths $P_1,P_2,\dots ,P_m$. The above deduction tell us a way for producing vv-type/vev-type TB-paws from labelled disjoint paths $P_1,P_2,\dots ,P_m$ matching with the non-Euler's graph $G$ (see Fig.\ref{fig:non-Euler-graph}).

\begin{figure}[h]
\centering
\includegraphics[height=6.6cm]{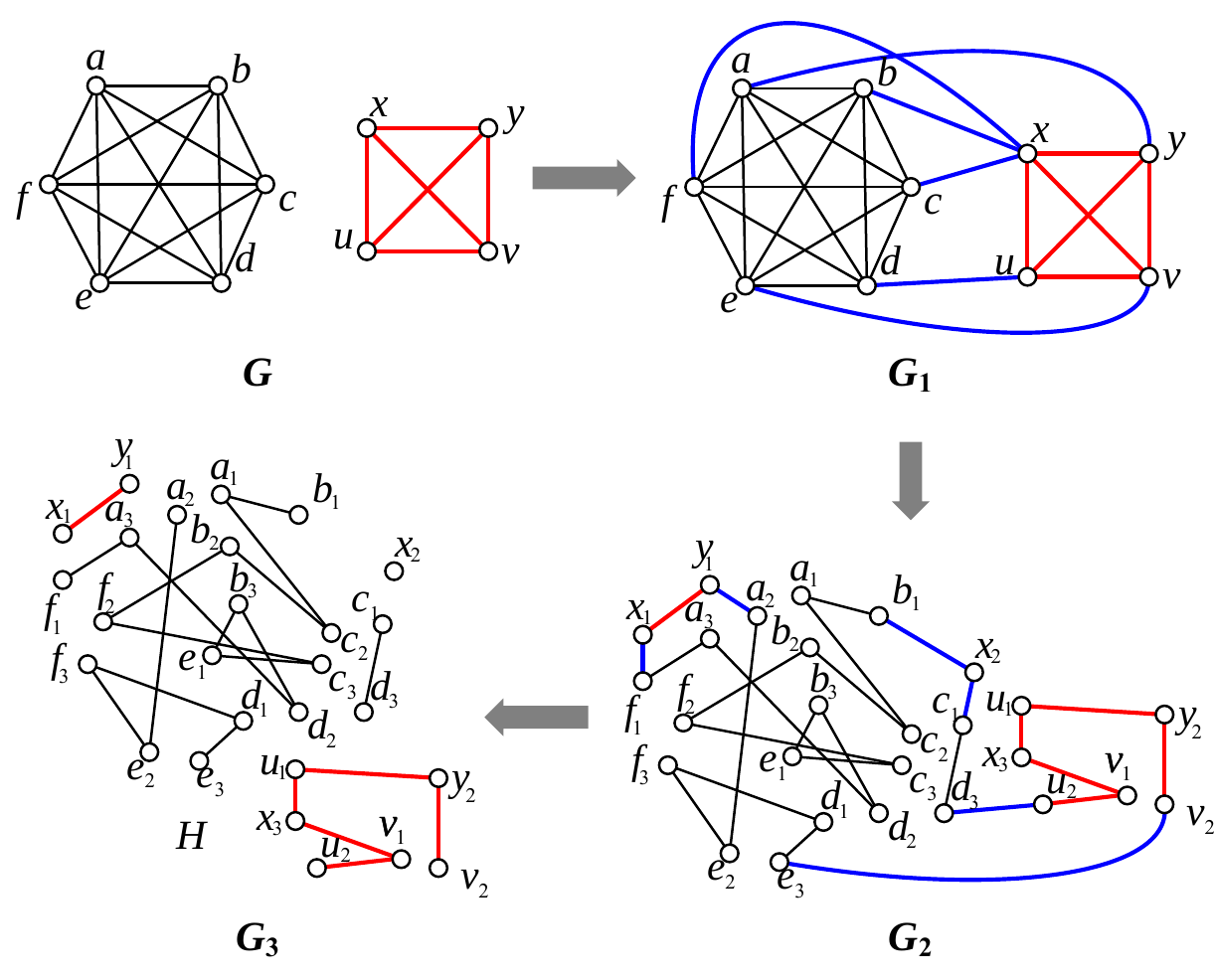}
\caption{\label{fig:non-Euler-graph}{\small $G$ is a non-Euler's graph. $G_1$ is an Euler's graph obtained by adding six new edges to $G$. $G_2$ is a Hamilton cycle obtained by implementing the non-adjacent identifying operation and the 2-edge-connected 2-degree-vertex splitting operation to $G_1$ ( \cite{Sun-Zhang-Yao-IAEAC-2017}). $G_3$ is the desired union of paths after deleting six edges in blue.}}
\end{figure}

\subsection{One labelling with many meanings}

In general situation, a $(p,q)$-graph $G$ admits a vertex labelling $f:V(G)\rightarrow [0,q-1]$ (or $[1,q]$), such that each edge $uv$ is labelled by $f'(uv)=F(f(u),f(v))$, where $F$ is a restrict condition. So, $G$ is just a Topsnut-gpw, $f'(uv)$ may be a number or a set. Here, we define£º
\begin{defn}\label{defn:edge-odd-graceful-labelling}
$^*$ A $(p,q)$-graph $G$ admits an \emph{edge-odd-graceful total labelling} $h:V(G)\rightarrow [0,q-1]$ and $h:E(G)\rightarrow [1,2q-1]^o$ such that  $\{h(u)+h(uv)+h(v):uv\in E(G)\}=[a,b]$ with $b-a+1=q$.\qqed
\end{defn}

\begin{figure}[h]
\centering
\includegraphics[height=2.2cm]{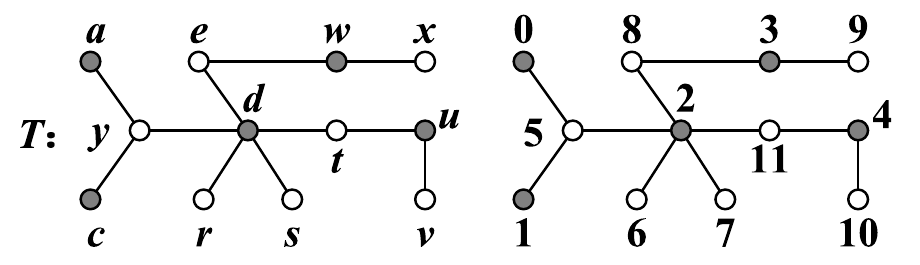}
\caption{\label{fig:a-labelling-many-meanings}{\small A lobster $T$ admitting a vertex labelling.}}
\end{figure}

We say that the lobster $T$ demonstrated in Fig.\ref{fig:a-labelling-many-meanings} admits: (1) a pan-edge-magic total labelling $f_1$ such that $f_1(u)+f_1(uv)+f_1(v)=16$; (2) a pan-edge-magic total labelling $f_2$ holding $f_2(u)+f_2(uv)+f_2(v)=27$; (3) a felicitous labelling $f_3$ satisfying $f_3(uv)=f_3(u)+f_3(v)~(\bmod~11)$; (4) an edge-magic graceful labelling $f_4$ such that $|f_4(u)+f_4(v)-f_4(uv)|=4$; (5) an edge-odd-graceful labelling $f_5$ keeping $f_5(uv)=$an odd number, $\{f_5(u)+f_5(uv)+f_5(v):uv\in E(T)\}=[16,26]$.

So, $T$ admits an e-set v-proper labelling defined as follows: Let $f:V(T)\rightarrow [0,11]$ with $f(a)=0$, $f(c)=1$, $f(d)=2$, $f(w)=3$, $f(u)=4$, $f(y)=5$, $f(r)=6$, $f(s)=7$, $f(e)=8$, $f(x)=9$, $f(v)=10$, $f(t)=11$. And each edge of $T$ has its own label set as follows:
$$
\begin{array}{ll}
f'(ay)=\{1,5,11,21,22\},&f'(cy)=\{2,6,10,19,21\},\\
f'(dy)=\{3,7,9,17,20\},&f'(de)=\{6,10,11,17\},\\
f'(dr)=\{4,8,15,19\},&f'(ds)=\{5,7,9,13,18\},\\
f'(dt)=\{2,3,5,9,14\}, &f'(ew)=\{0,5,7,9,16\},\\
f'(xw)=\{1,4,7,8,15\}, &f'(ut)=\{1,4,11,12\},\\
f'(uv)=\{2,3,10,13\}.
\end{array}
$$ Thereby, we can get five vev-type TBpaws as follows:

\vskip 0.4cm

\noindent $D^1_{vev}(T)=295110101286772685349231114210$,

\noindent $D^2_{vev}(T)=220522021121961872178163159214111241310$,

\noindent $D^3_{vev}(T)=27550612869721080319221144310$,

\noindent $D^4_{vev}(T)=235102124657268738929111141010$,

\noindent and

\noindent $D^5_{vev}(T)=2175210191215613721189379251114310$\\
staring with the vertex $d$ of $T$.

\vskip 0.4cm

In real operation, we can select any vertex of $T$ as the initial vertex for irregular reason. Furthermore, we have TB-paws as follows $$D=D^{i_1}_{vev}(T)\uplus D^{i_2}_{vev}(T)\uplus \cdots \uplus D^{i_m}_{vev}(T)$$ is a vev-type TB-paw, where $D^{i_j}_{vev}(T)\in \{D^i_{vev}(T):i\in [1,5]\}$ with $j\in [1,m]$. Summarizing the above facts, we have a new labelling with many meanings as follows:

\begin{defn}\label{defn:multiple-meanings-vertex-labelling}
$^*$ A $(p,q)$-graph $G$ admits a \emph{multiple edge-meaning vertex labelling} $f:V(G)\rightarrow [0,p-1]$ such that (1) $f(E(G))=[1,q]$ and $f(u)+f(uv)+f(v)=$a constant $k$; (2) $f(E(G))=[p,p+q-1]$ and $f(u)+f(uv)+f(v)=$a constant $k'$; (3) $f(E(G))=[0,q-1]$ and $f(uv)=f(u)+f(v)~(\bmod~q)$; (4) $f(E(G))=[1,q]$ and $|f(u)+f(v)-f(uv)|=$a constant $k''$; (5) $f(uv)=$an odd number for each edge $uv\in E(G)$ holding $f(E(G))=[1,2q-1]^o$, and $\{f(u)+f(uv)+f(v):uv\in E(T)\}=[a,b]$ with $b-a+1=q$.\qqed
\end{defn}

\subsection{A new total set-labelling}

A new total set-labelling is defined by the intersection operation on sets for making TB-paws with longer bytes.

\begin{defn}\label{defn:graceful-odd-graceful-total-set-labelling}
$^*$ A $(p,q)$-graph $G$ admits a vertex set-labelling $f:V(G)\rightarrow [1,q]^2$~(or $[1,2q-1]^2)$, and induces an edge set-labelling $f'(uv)=f(u)\cap f(v)$. If we can select a \emph{representative} $a_{uv}\in f'(uv)$ for each edge label set $f'(uv)$ with $uv \in E(G)$ such that $\{a_{uv}:~uv\in E(G)\}=[1,q]$ (or $[1,2q-1]^o$), then we call $f$ a \emph{graceful-intersection (an odd-graceful-intersection) total set-labelling} of $G$.\qqed
\end{defn}

\begin{thm}\label{thm:total-set-labellings}
Each tree $T$ admits a graceful-intersection (an odd-graceful-intersection) total set-labelling.
\end{thm}
\begin{proof} Assume that a tree $T-x$ admits a graceful intersection total set-labelling $f:V(T-x)\rightarrow [1,q]^2$, where $x$ is a leaf of a tree $T$ of $q$ edges. Let $y$ be adjacent with $x$ in $T$. We add the leaf $x$ to $T-x$, and define $h:V(T)\rightarrow [1,q]^2$ by $h(u)=f(u)$ for $u\in V(T-\{x,y\})$, $h(y)=f(y)\cup \{q\}$ and $h(x)=\{q\}$. Notice that $$h(xy)=h(x)\cap h(y)=h(x)=\{q\}$$ we select $q$ as the representative of $h(xy)$. By the hypothesis of induction, $T$ admits a graceful intersection total set-labelling $f:V(T)\rightarrow [1,q]^2$.

The proof of ``$T$ admits an odd-graceful intersection total set-labelling'' is similar with the above one with $[1,2q-1]^2$.
\end{proof}

\vskip 0.2cm

A Topsnut-gpw $H$ shown in Fig.\ref{fig:graceful-intersection} distributes us a TB-paw
$${
\begin{split}
D(H)=&1231133412334312344563445\\
&643455777898899661011124\\
&566101112121211111010.
\end{split}}$$

\begin{figure}[h]
\centering
\includegraphics[height=2.8cm]{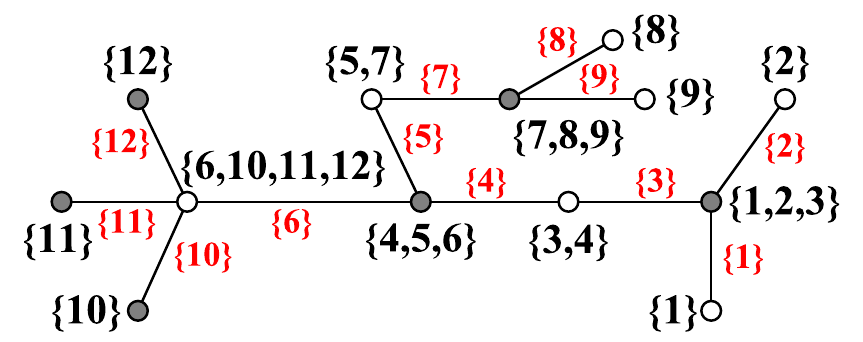}
\caption{\label{fig:graceful-intersection}{\small A tree $H$ admitting a graceful-intersection total set-labelling for illustrating Theorem \ref{thm:total-set-labellings}.}}
\end{figure}

We define a \emph{regular rainbow set-sequence} $\{R_k\}^{q}_1$ as: $R_k=[1,k]$ with $k\in [1,q]$, where $[1,1]=\{1\}$.

\begin{thm}\label{thm:rainbow-total-set-labellings}
Each tree $T$ of $q$ edges admits a regular rainbow intersection total set-labelling based on a regular rainbow set-sequence $\{R_k\}^{q}_1$.
\end{thm}
\begin{proof} Suppose $x$ is a leaf of a tree $T$ of $q$ edges, so $T-x$ is a tree of $(q-1)$ edges. Assume that $T-x$ admits a regular rainbow set-sequence $\{R_k\}^{q-1}_1$ total set-labelling $f$. Let $y$ be adjacent with $x$ in $T$. We define a labelling $g$ of $T$ in this way: $g(w)=f(w)$ for $w\in V(T)\setminus \{y,x\}$, $g(y)=R_{q+1}=[1,q+1]$ and $g(x)=R_q=[1,q]$. Therefore, we have $g(u_iv_j)=g(u_i)\cap g(v_j)=[1,i]\cap [1,j]$ for $u_iv_j\in E(T)\setminus \{xy\}$, and $g(xy)=g(x)\cap g(y)=[1,q]$, and $g(s)\neq g(t)$ for any pair of vertices $s$ and $t$. We claim that $g$ is a regular rainbow intersection total set-labelling of $T$ by the hypothesis of induction.
\end{proof}

An example is pictured in Fig.\ref{fig:rainbow-intersection} for understanding Theorem \ref{thm:rainbow-total-set-labellings}, and we can write a TB-paw from this example as follows
$${
\begin{split}
D(T^*)=&\underline{1234}1112121231231234\underline{123456789}\\
&\underline{1234}\underline{123456789}123456789\underline{12345678910}\\
&12345678910\underline{123456789101112}123456789\\
&101112\underline{1234567891011}1234567891011\\
&123456789101112\underline{12345678910111213}\\
&12345678\underline{12345678}\underline{1234567}1234567\\
&123456\underline{123456}12345\underline{12345}.
\end{split}}$$

\begin{figure}[h]
\centering
\includegraphics[height=3cm]{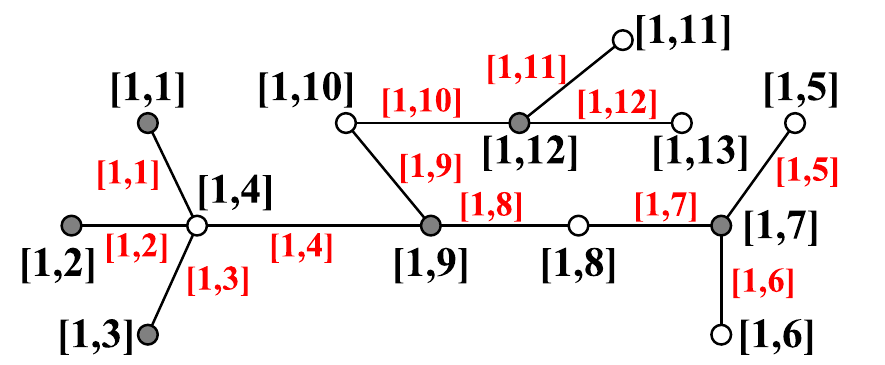}
\caption{\label{fig:rainbow-intersection}{\small A tree $T^*$ admitting a regular rainbow intersection total set-labelling for understanding Theorem \ref{thm:rainbow-total-set-labellings}.}}
\end{figure}

The proof of Theorem \ref{thm:rainbow-total-set-labellings} can be used to estimate the number of regular rainbow intersection total set-labellings of a tree $T$ of $q$ edges based on a regular rainbow set-sequence $\{R_k\}^q_1$. As known, a tree $T$ has its number $n_1(T)$ of leaves as follows
\begin{equation}\label{eqa:Yao-Zhang-Wang-2010}
n_1(T)=2+\sum_{d\geq 3}(d-2)n_d(T)
\end{equation}
where $n_d(T)$ is the number of vertices of degree $d$ in the tree $T$ \cite{Yao-Zhang-Wang-2010, Yao-Chen-Yang-Wang-Zhang-Zhang2012}. The formula (\ref{eqa:Yao-Zhang-Wang-2010}) tells us there are at least $n_1(T)$ different regular rainbow intersection total set-labellings for each tree $T$. Ie seems to be difficult to find all such total set-labellings for a given tree.

Each tree admits a regular odd-rainbow intersection total set-labelling based on a \emph{regular odd-rainbow set-sequence} $\{R_k\}^{q}_1$ defined as: $R_k=[1,2k-1]$ with $k\in [1,q]$, where $[1,1]=\{1\}$. Moreover, we can define a \emph{regular Fibonacci-rainbow set-sequence} $\{R_k\}^{q}_1$ by $R_1=[1,1]$, $R_2=[1,1]$, and $R_{k+1}=R_{k-1}\cup R_{k}$ with $k\in [2,q]$; or a $\tau$-term Fibonacci-rainbow set-sequence $\{\tau,R_i\}^{q}_1$ holds: $R_i=[1,a_i]$ with $a_i>1$ and $i\in [1,q]$, and
$$R_k=\sum ^{k-1}_{i=k-\tau}R_i$$
with $k>\tau$ \cite{Ma-Wang-Wang-Yao-Theoretical-Computer-Science-2018}. It is interesting on various rainbow set-sequences for non-tree graphs.

\section{Topsnut-matrices}

Topsnut-matrices differ from the popular matrices in algebra. No operations of addition and subtraction on numbers are suitable for Topsnut-matrices. We will show some operations on Topsnut-matrices from the insight of construction and decomposition on graphs.

\subsection{Definition of Topsnut-matrices}

\begin{defn}\label{defn:Topsnut-matrix}
(\cite{Sun-Zhang-Zhao-Yao-2017, Yao-Sun-Zhao-Li-Yan-2017}) A \emph{Topsnut-matrix} $A_{vev}(G)$ of a $(p,q)$-graph $G$ is defined as
\begin{equation}\label{eqa:a-formula}
\centering
A_{vev}(G)= \left(
\begin{array}{ccccc}
x_{1} & x_{2} & \cdots & x_{q}\\
w_{1} & w_{2} & \cdots & w_{q}\\
y_{1} & y_{2} & \cdots & y_{q}
\end{array}
\right)=(X~W~Y)^{-1}
\end{equation}\\
where
\begin{equation}\label{eqa:three-vectors}
{
\begin{split}
&X=(x_1 ~ x_2 ~ \cdots ~x_q), W=(e_1 ~ e_2 ~ \cdots ~e_q)\\
&Y=(y_1 ~ y_2 ~\cdots ~ y_q),
\end{split}}
\end{equation}
where each edge $e_i$ has its own two ends $x_i$ and $y_i$ with $i\in [1,q]$; and $G$ has another \emph{Topsnut-matrix} $A_{vv}(G)$ defined as $A_{vv}(G)=(X,Y)^{-1}$, where $X,Y$ are called \emph{vertex-vectors}, $W$  an \emph{edge-vector}.\qqed
\end{defn}

Clearly, the number of different edge-vectors $W=(e_1 ~ e_2 ~ \cdots ~e_q)$ of a $(p,q)$-graph $G$ is just $q!$, and each end of two ends of an edge can be arranged in $X$ or in $Y$, so the number of the Topsnut-matrices $A_{vev}(G)$ (resp. $A_{vv}(G)$) of $G$ is equal to $q!\cdot 2^q$.

\begin{figure}[h]
\centering
\includegraphics[height=4cm]{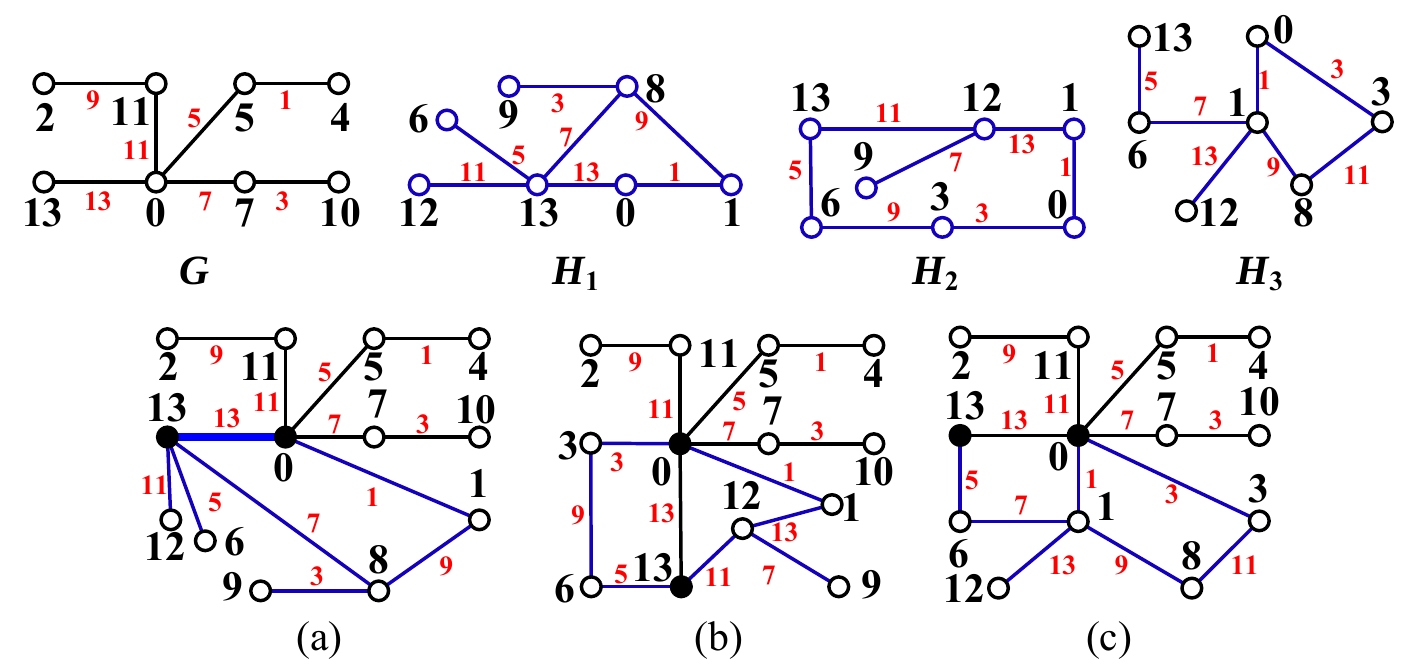}
\caption{\label{fig:twin-odd-graceful-0}{\small An odd-graceful tree $G$ matches with three odd-elegant graphs $H_1,H_2,H_3$ respectively, and there are three odd-graceful vs odd-elegant graphs (a) $\ominus (G,H_1)$, (b) $\odot \langle G,H_2\rangle $, and (c) $\odot \langle G,H_3\rangle $.}}
\end{figure}

\begin{figure}[h]
\centering
\includegraphics[height=2.8cm]{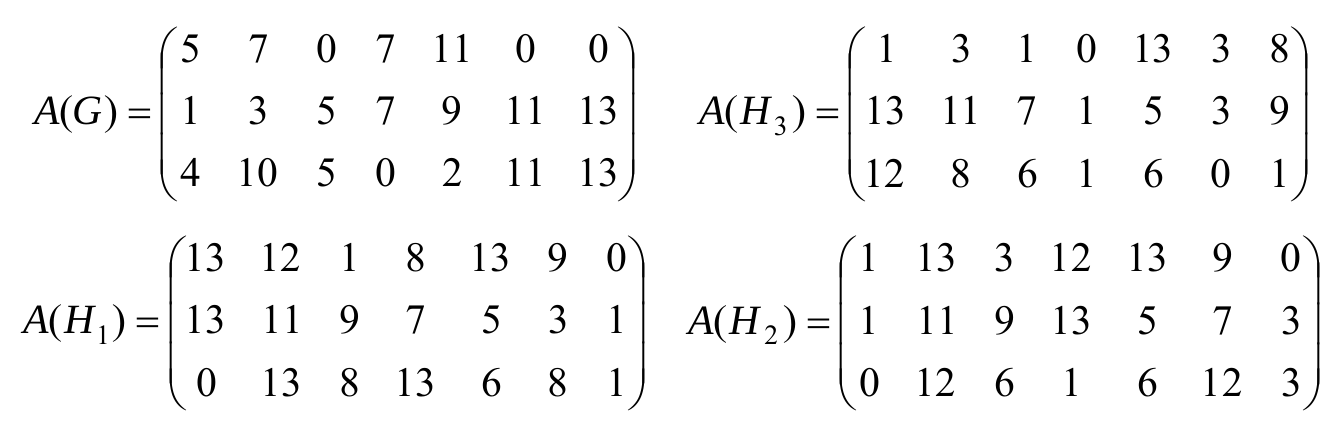}
\caption{\label{fig:4-Topsnut-matrices}{\small An odd-graceful tree $G$ and three odd-elegant graphs $H_1,H_2,H_3$ presented in Fig.\ref{fig:twin-odd-graceful-0} have their own Topsnut-matrices $A_{vev}(G)$ and $A_{vev}(H_1),A_{vev}(H_2)$ and $A_{vev}(H_3)$.}}
\end{figure}

\begin{figure}[h]
\centering
\includegraphics[height=1.4cm]{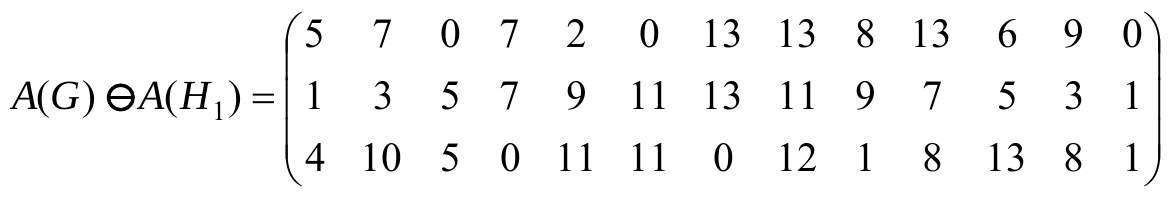}
\caption{\label{fig:matrix-G-H1}{\small The Topsnut-matrix $A_{vev}(\ominus \langle G,H_1\rangle )$ of an odd-graceful vs odd-elegant graph $\ominus \langle G,H_1\rangle $, write as $A_{vev}(\ominus \langle G,H_1\rangle )=A_{vev}(G)\ominus A_{vev}(H_1)$.}}
\end{figure}

\begin{figure}[h]
\centering
\includegraphics[height=1.3cm]{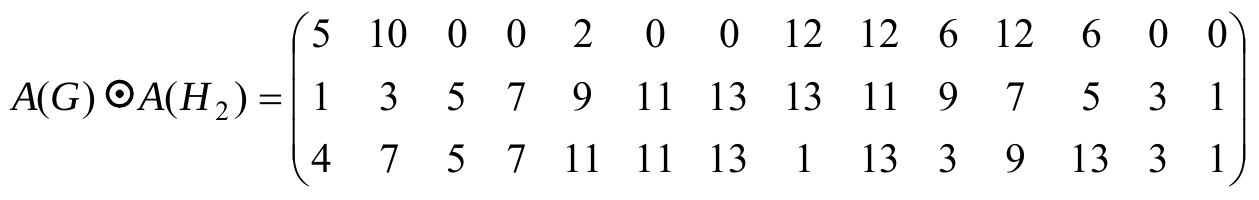}
\caption{\label{fig:matrix-G-H2}{\small The Topsnut-matrix $A_{vev}(\odot \langle G,H_2\rangle )$ of an odd-graceful vs odd-elegant graph $\odot \langle G,H_2\rangle $, denoted as $A_{vev}(\odot \langle G,H_2\rangle )=A_{vev}(G)\odot A_{vev}(H_2)$.}}
\end{figure}

We define the following operations:

\vskip 0.2cm

\textbf{Opr-1.} $A_{vev}(G)=A_{vev}(\ominus \langle G,H_1\rangle )\ominus^{-} A_{vev}(H_1)$

$A_{vev}(H_1)=A_{vev}(\ominus \langle G,H_1\rangle )\ominus^{-} A_{vev}(G)$

$A_{vev}(G)=A_{vev}(\odot \langle G,H_2\rangle )\odot^{-} A_{vev}(H_2)$

$A_{vev}(H_2)=A_{vev}(\odot \langle G,H_2\rangle )\odot^{-} A_{vev}(G)$.

\textbf{Opr-2.} Set three reciprocals

$X_{-1}=(x_q ~ x_{q-1} ~ \cdots ~x_1)$, $W_{-1}=(e_q ~ e_{q-1} ~ \cdots ~e_1)$

$Y_{-1}=(y_q ~ y_{q-1} ~\cdots ~ y_1)$, then we get the \emph{reciprocal} of $A_{vev}(G)$, denoted as $A^{-1}_{vev}(G)=(X_{-1},W_{-1},Y_{-1})^{-1}$.

\vskip 0.2cm

\begin{figure}[h]
\centering
\includegraphics[height=5cm]{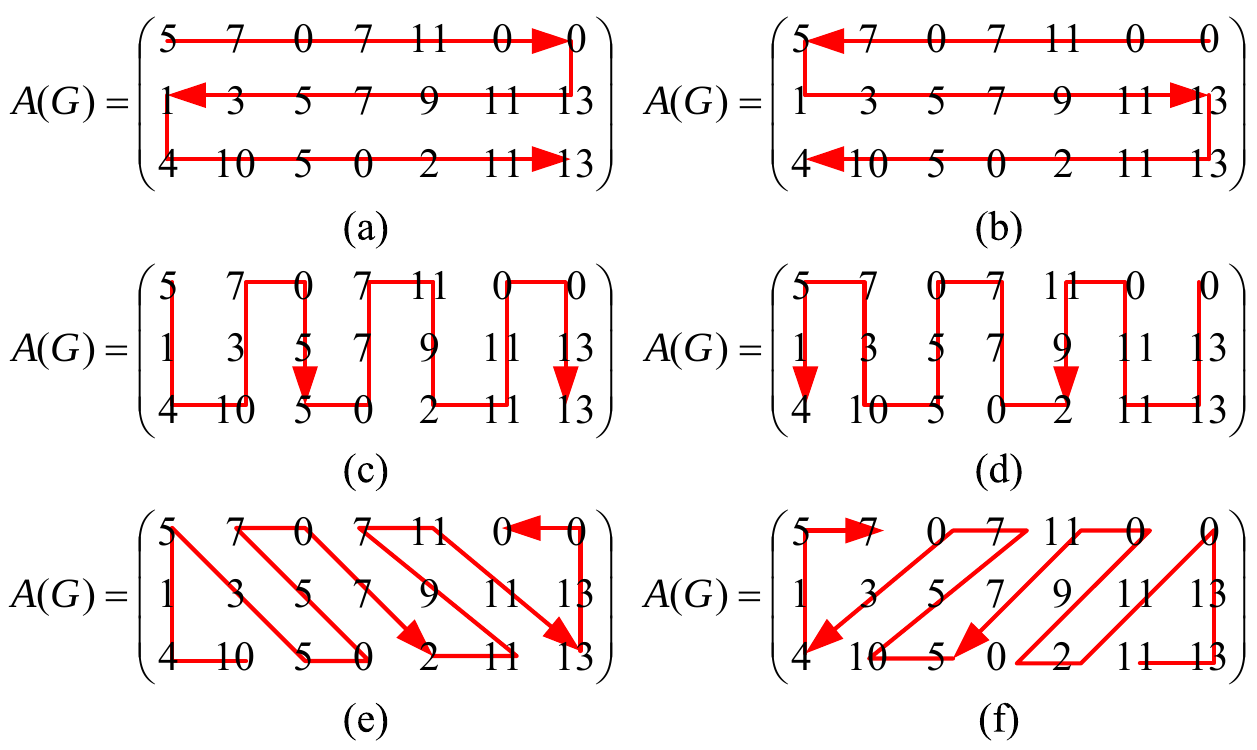}
\caption{\label{fig:6-methods}{\small Six methods for generating TB-paws from the Topsnut-matrix $A(G)$ indicated in Fig.\ref{fig:4-Topsnut-matrices}.}}
\end{figure}

A \emph{text string} $D=c_1c_2\cdots c_m$ has its own \emph{reciprocal text string} defined by $D_{-1}=c_mc_{m-1}\cdots c_2c_1$, also, we say $D$ and $D_{-1}$ match with each other. we consider that $D$ is a\emph{ public key}, and $D_{-1}$ is a \emph{private key}. For a fixed Topsnut-matrix $A_{vev}(G)$ and its reciprocal $A^{-1}_{vev}(G)$, we have the following basic methods for generating vv-type/vev-type TB-paws:

\vskip 0.2cm

\textbf{Met-1.} \textbf{I-route.} $D_1(G)=x_1x_2\cdots x_qe_1 e_2\cdots e_qy_1 y_2\cdots y_q$ with its reciprocal
$D_1(G)_{-1}$ (see Fig.\ref{fig:6-methods}(a)).

\textbf{Met-2.} $D_2(G)=x_q x_{q-1} \cdots x_1e_q e_{q-1} \cdots e_1y_q y_{q-1} \cdots y_1$ with its reciprocal $D_2(G)_{-1}$ (see Fig.\ref{fig:6-methods}(b)).

\textbf{Met-3.} \textbf{II-route.} $D_3(G)=x_1e_1y_1y_2e_2x_2x_3e_3y_3\cdots x_qe_qy_q$ with its reciprocal $D_3(G)_{-1}$ (see Fig.\ref{fig:6-methods}(c)).

\textbf{Met-4.} $D_4(G)=x_q e_q y_q y_{q-1}e_{q-1}x_{q-1} \cdots y_2e_2x_2x_1e_1y_1$ with its reciprocal $D_4(G)_{-1}$ (see Fig.\ref{fig:6-methods}(d)).

\textbf{Met-5.} \textbf{III-route.} We set $$D_5(G)=y_2y_1e_1x_1e_2y_3y_4e_3x_2\cdots x_{q-2}e_{q-1}y_qe_{q}x_{q}x_{q-1}$$ with its reciprocal $D_5(G)_{-1}$ (see Fig.\ref{fig:6-methods}(e)).

\textbf{Met-6.} We take
$${
\begin{split}
D_6(G)= &y_{q-1}y_qe_{q}x_{q}e_{q-1}y_{q-2}y_{q-3}e_{q-2}x_{q-1}\cdots \\
&x_{2}e_{2}y_1e_{1}x_{1}x_{2}
\end{split}}
$$ with its reciprocal $D_6(G)_{-1}$ (see Fig.\ref{fig:6-methods}(f)).

For example, we can get the following vv-type/vev-type TB-paws by a Topsnut-matrix $A(a-1)$ depicted in Fig.\ref{fig:set-matrix-1}:
$$D_I(a-1)=55072216807876543216833680707168,$$
$$D_{II}(a-1)=51683250733684225070761687168807,$$
and
$$D_{III}(a-1)=36815236835074070752261687168078$$
according to I-route, II-route and III-route.

\begin{figure}[h]
\centering
\includegraphics[height=1.5cm]{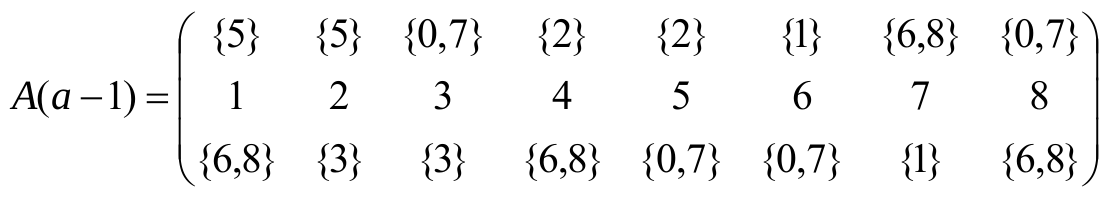}
\caption{\label{fig:set-matrix-1}{\small A Topsnut-matrix $A(a-1)$ of a graph exhibited in Fig.\ref{fig:v-set-e-proper}(a-1).}}
\end{figure}

Moreover, the Topsnut-matrix $A(a-2)$ pictured in Fig.\ref{fig:v-set-e-proper}(a-2) induces
$$D_I(a-2)=81577631561234567802815020202315,$$
$$D_{II}(a-2)=88028715761502576402023315231516,$$
and
$$D_{III}(a-2)=80188715026157502024763315231561$$
by I-route, II-route and III-route.

\begin{figure}[h]
\centering
\includegraphics[height=1.5cm]{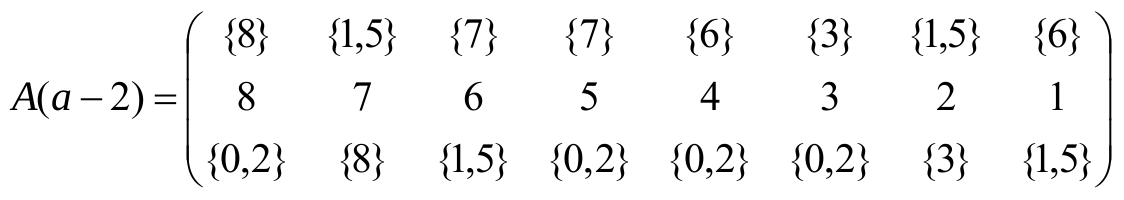}
\caption{\label{fig:set-matrix-2}{\small A Topsnut-matrix $A(a-2)$ of a graph revealed in Fig.\ref{fig:v-set-e-proper}(a-2).}}
\end{figure}

\textbf{Met-7.} In general, let $g: X\cup W\cup Y\rightarrow \{a_i:i\in [1,3q]\}$ be a bijection on the Topsnut-matrix $A_{vev}(G)$ of $G$, so it induces a vv-type/vev-type TB-paw
\begin{equation}\label{eqa:Topsnut-matrix-vs-TB-paws}
g(G)=g^{-1}(a_{i_1})g^{-1}(a_{i_2})\cdots g^{-1}(a_{i_{3q}})
\end{equation} with its reciprocal $g(G)_{-1}$, where $a_{i_1},a_{i_2},\dots,a_{i_{3q}}$ is a permutation of  $a_1,a_2,\dots g^{-1},a_{3q}$. So, there are $(3q)!$ vv-type/vev-type TB-paws by (\ref{eqa:Topsnut-matrix-vs-TB-paws}), in general. Clearly, there are many random routes for inducing vv-type/vev-type TB-paws from Topsnut-matrices. It may be interesting to look continuous routes in $A_{vev}(G)$ (see red lines presented in Fig.\ref{fig:6-methods} and Fig.\ref{fig:6-other-methods}).

Motivated from Fig.\ref{fig:6-methods} and Fig.\ref{fig:6-other-methods}, a Topsnut-matrix $A_{vev}(G)$ of a $(p,q)$-graph $G$ may has $N_{fL}(m)$ groups of continuous fold-lines $L_{j,1},L_{j,2},\dots, L_{j,m}$~($=\{L_{j,i}\}^m_1$) for $j\in [1,N_{fL}(m)]$ and $m\in [1,M]$, where each continuous fold-line $L_{j,i}$ has own initial point $(a_{j,i},b_{j,i})$ and terminal point $(c_{j,i},d_{j,i})$ in $xoy$-plan, and $L_{j,i}$ is internally disjoint, such that each element of $A_{vev}(G)$ is on one and only one of the continuous fold-lines $=\{L_{j,i}\}^m_1$ after we put the elements of $A_{vev}(G)$ into  $xoy$-plan. Notice that each  fold-line $L_{j,i}$ has its initial and terminus points, so $2m\leq 3q$, so $M=\lfloor 3q/2\rfloor $. Each group of continuous fold-lines $\{L_{j,i}\}^m_1$ can distributes us $m!$ vv-type/vev-type TB-paws, so we have at least $N_{fL}(m)\cdot m!$ vv-type/vev-type TB-paws with $m\in [1,M]$. Thereby,  $G$ gives us the number $N^*(G)$ of vv-type/vev-type TB-paws in total as follows
\begin{equation}\label{eqa:random-routes}
{
\begin{split}
N^*(G)=q!\cdot 2^q\sum^{M}_{m=1}N_{fL}(m)\cdot m!.
\end{split}}
\end{equation}

\begin{figure}[h]
\centering
\includegraphics[height=5cm]{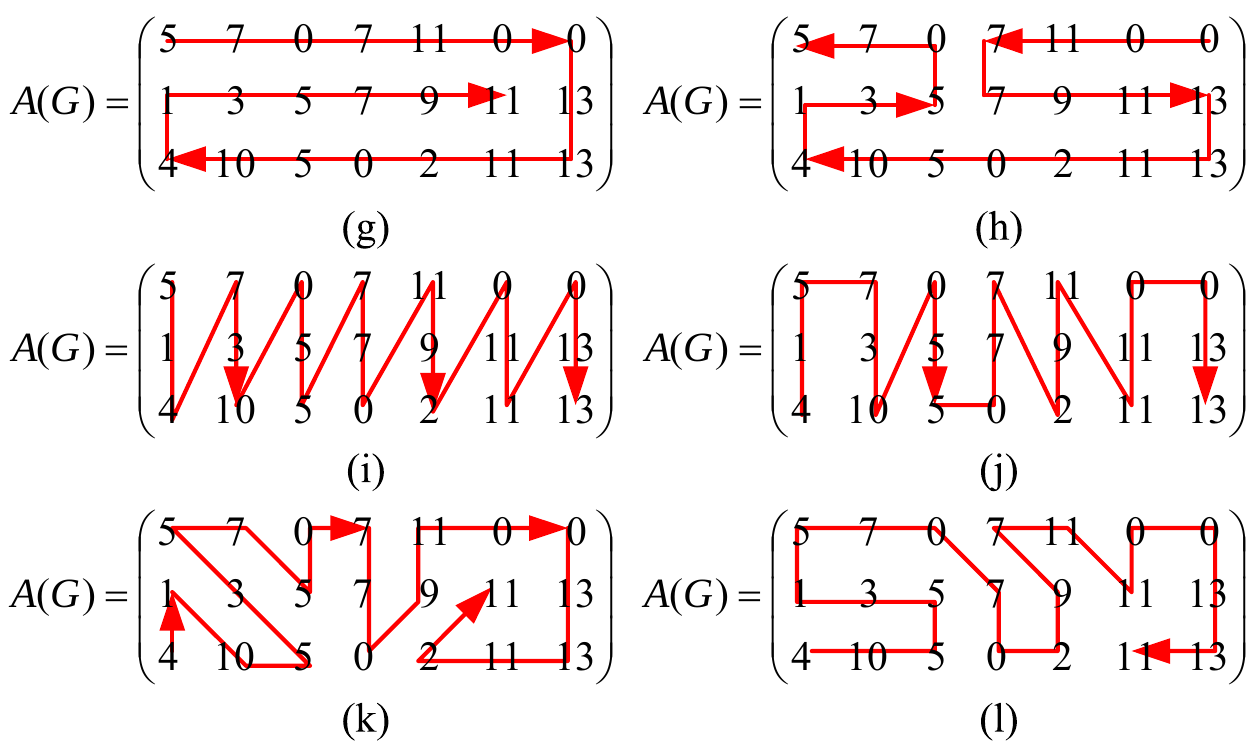}
\caption{\label{fig:6-other-methods}{\small Six random routes on the Topsnut-matrix $A(G)$ demonstrated in Fig.\ref{fig:4-Topsnut-matrices}.}}
\end{figure}

The number $N_{tbp}(G)$ of all vv-type/vev-type TB-paws generated from a $(p,q)$-graph $G$ can be computed in the formula (\ref{eqa:TB-paws-number}):

\begin{thm}\label{thm:TB-paws-numbers}
A Topsnut-gpw $(p,q)$-graph $G$ distributes us
\begin{equation}\label{eqa:TB-paws-number}
N_{tbp}(G)=(3q)!\cdot q!\cdot 2^{q-1}
\end{equation} vv-type/vev-type TB-paws in total.
\end{thm}
\begin{proof}Let $D=a_{i_1}a_{i_2}\cdots a_{i_{3q}}$ be a vv-type/vev-type TB-paw made by $A_{vev}(G)_{3\times q}$, the first number $a_{i_1}$ in $D$ has $3q$ positions to be selected for standing, then the second number $a_{i_2}$ of $D$ has $(3q-1)$ positions to be selected for standing, go on in this way, the number of vv-type/vev-type TB-paws produced from $A_{vev}(G)_{3\times q}$ is just $(3q)!$, as desired. Because the number of the Topsnut-matrices $A_{vev}(G)$ of a $(p,q)$-graph $G$ is $q!\cdot 2^q$, thus, we get (\ref{eqa:TB-paws-number}).
\end{proof}

\vskip 0.2cm

The Topsnut-gpw $G$ appeared in Fig.\ref{fig:example-1} has 190 edges, so it gives us
\begin{equation}\label{eqa:wang-paper-example}
N_{tbp}(G)=(570!)\cdot (190!)\cdot 2^{190}
\end{equation}
vv-type/vev-type TB-paws in total, in which each vev-type TB-paw $D(G)$ has at least 380 bytes or more.

\subsection{Operations of vertex-split (edge-split) and vertex-coincidence (edge-coincidence)}

The authors \cite{Yao-Sun-Zhang-Mu-Sun-Wang-Su-Zhang-Yang-Yang-2018arXiv} have defined the following operations for Topsnut-matrices: In Fig.\ref{fig:two-operations}, a \emph{vertex-split operation} from (a) to (b); a \emph{vertex-coincident operation} from (b) to (a); an \emph{edge-split operation} from (c) to (d); and an \emph{edge-coincident operation} from (d) to (c). Let $N(x)$ be the neighbor set of all neighbors of a vertex $x$. In Fig.\ref{fig:two-operations}, after split operations, then the following neighbor sets hold $N(y')\cap N(y'')=\emptyset $, $N(u')\cap N(u'')=\emptyset $ and $N(v')\cap N(v'')=\emptyset $ in the resulting graphs.

\begin{figure}[h]
\centering
\includegraphics[height=6.4cm]{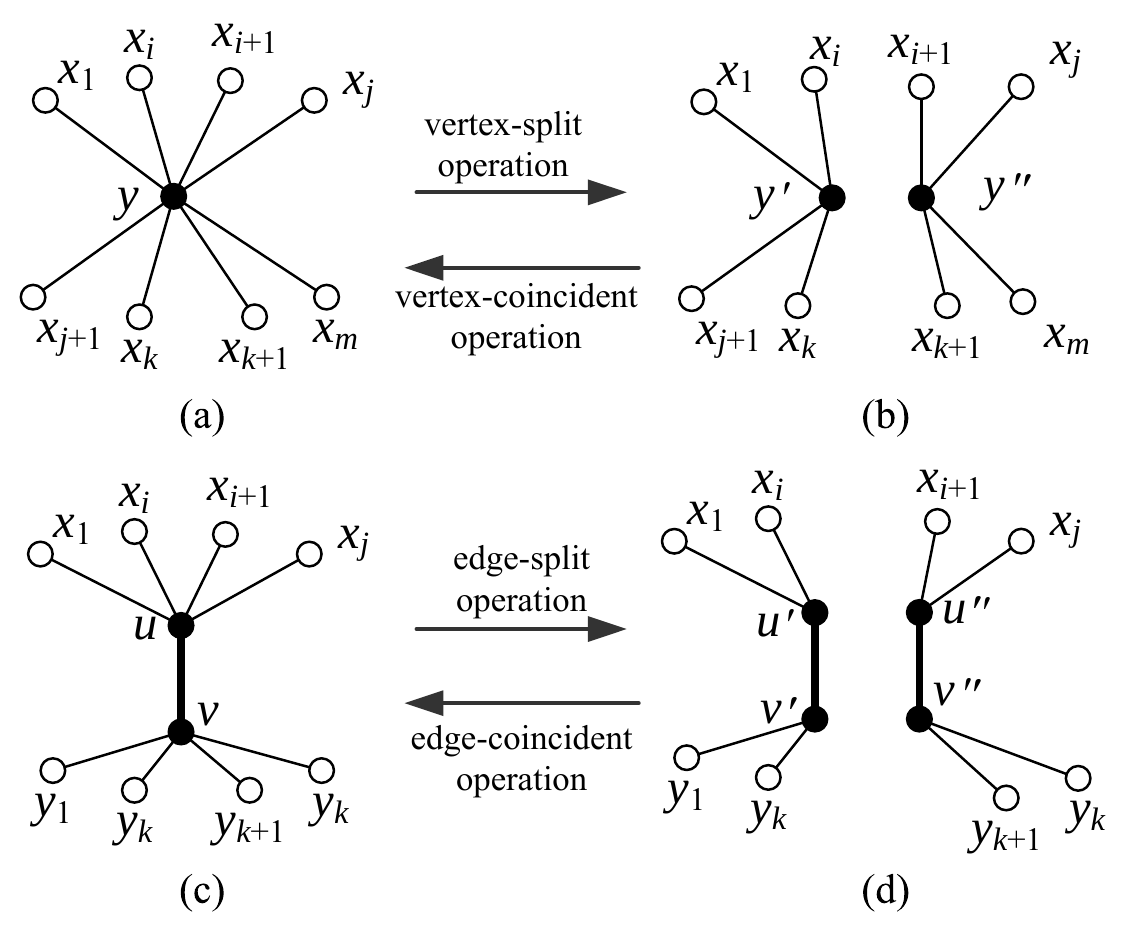}\\
\caption{\label{fig:two-operations} {\small A scheme for illustrating four graph operations: vertex-split operation; vertex-identifying operation; edge-split operation; edge-identifying operation cited from \cite{Yao-Sun-Zhang-Mu-Sun-Wang-Su-Zhang-Yang-Yang-2018arXiv}.}}
\end{figure}

We do a vertex-split operation to a vertex $y$ of a graph $H$, the resulting graph is denoted as $H\wedge y$. So, $|V(H\wedge y)|=1+|V(H)|$ and $|E(H\wedge y)|=|E(H)|$ (see Fig.\ref{fig:two-operations}(b)). The resulting graph obtained by doing an edge-split operation to an edge $uv$ of $H$ is written as $H\wedge uv$, thus, $|V(H\wedge uv)|=2+|V(H)|$ and $|E(H\wedge uv)|=1+|E(H)|$ (see Fig.\ref{fig:two-operations}(d)). Conversely, we can coincide two vertices $x,y$ of a graph $H$ into one if $N(x)\cap N(y)=\emptyset$ for obtaining a graph $H(x\circ y)$, this procedure is called a \emph{vertex-coincident operation}. If two edges $xy$ and $uv$ of satisfy $N(x)\cap N(u)=\emptyset $, $N(x)\cap N(v)=\emptyset $, $N(y)\cap N(u)=\emptyset $ and $N(y)\cap N(v)=\emptyset $, then we coincide $xy$ with $uv$ into one edge $(x,u)(y,v)$, the resulting graph is denoted by $H(xy\circ uv)$, and call this procedure as an \emph{edge-coincident operation}.

\begin{defn}\label{defn:v-split-connectivity}
\cite{Wang-Zhang-Mei-Yao2018-Split} A \emph{v-split $k$-connected graph} $H$ holds: $H\wedge \{x_i\}^{k}_1$ is disconnected, where $V^*=\{x_1,x_2,\dots,x_k\}$ is a subset of $V(H)$, each component $H_j$ of $H\wedge \{x_i\}^{k}_1$ has at least a vertex $w_j\not \in V^*$, $|V(H\wedge \{x_i\}^{k}_1)|=k+|V(H)|$ and $|E(H\wedge \{x_i\}^{k}_1)|=|E(H)|$. The smallest number of $k$ for which $H\wedge \{x_i\}^{k}_1$ is disconnected is called the \emph{v-split connectivity} of $H$, denoted as $\gamma_{vs}(H)$ (see Fig.\ref{fig:split-coincident-connectivity}).\qqed
\end{defn}

\begin{defn}\label{defn:e-split-connectivity}
\cite{Wang-Zhang-Mei-Yao2018-Split} An \emph{e-split $k$-connected graph} $H$ holds: $H\wedge \{e_i\}^{k}_1$ is disconnected, where $E^*=\{e_1,e_2,\dots,e_k\}$ is a subset of $E(H)$, each component $H_j$ of $H\wedge \{e_i\}^{k}_1$ has at least a vertex $w_j$ being not any end of any edge of $E^*$, $|V(H\wedge \{e_i\}^{k}_1)|=2k+|V(H)|$ and $|E(H\wedge \{x_i\}^{k}_1)|=k+|E(H)|$. The smallest number of $k$ for which $H\wedge \{e_i\}^{k}_1$ is disconnected is called the \emph{e-split connectivity} of $H$, denoted as $\gamma_{es}(H)$ (see Fig.\ref{fig:split-coincident-connectivity}).\qqed
\end{defn}

\begin{figure}[h]
\centering
\includegraphics[height=5.4cm]{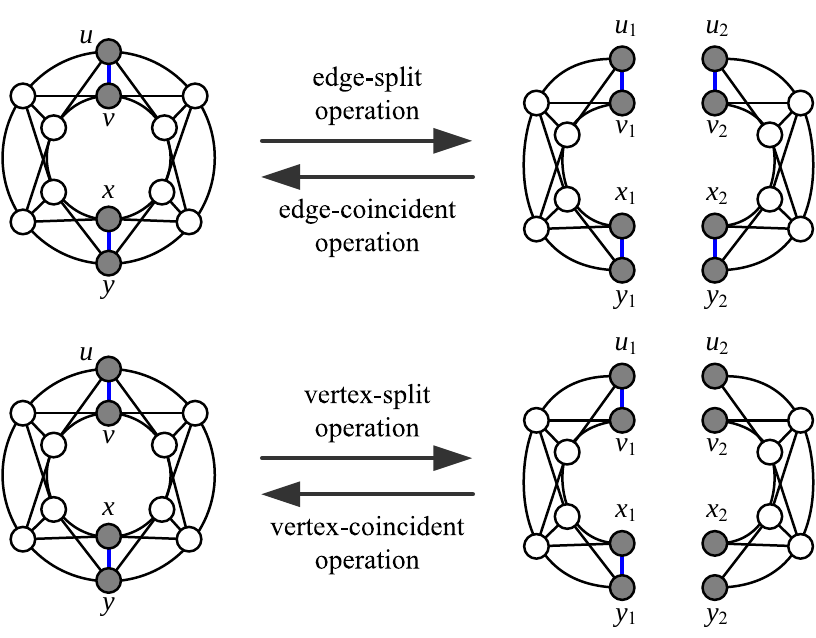}
\caption{\label{fig:split-coincident-connectivity}{\small A graph $H$ has $\gamma_{es}(H)=2$, $\gamma_{vs}(H)=4$, $\kappa(H)=4$, $\kappa'(H)=5$ and $\delta(H)=5$.}}
\end{figure}

Recall that the minimum degree $\delta(H)$, the vertex connectivity $\kappa(H)$ and the edge connectivity $\kappa'(H)$ of a simple graph $G$ hold
\begin{equation}\label{eqa:popular-connectivity}
\kappa(H)\leq \kappa'(H)\leq \delta(H).
\end{equation}
true in graph theory. However, we do not have the inequalities (\ref{eqa:popular-connectivity}) about the minimum degree $\delta(H)$, the v-split connectivity $\gamma_{vs}(H)$ and the e-split connectivity $\gamma_{es}(H)$. But, we have

\begin{thm}\label{thm:new-old-connectivity}
\cite{Wang-Zhang-Mei-Yao2018-Split} Any simple and connected graph $H$ holds $\gamma_{vs}(H)=\kappa(H)$, where $\kappa(H)$ is the popular vertex connectivity of $H$, and $\gamma_{vs}(H)$ is the v-split connectivity of $H$. Moreover, the e-split connectivity $\gamma_{es}(H)$ of $H$ satisfies $\gamma_{vs}(H)\leq 2\gamma_{es}(H)$.
\end{thm}

The vertex-split/edge-split operations and vertex-coincidence/edge-coincidence operations enable us to define some operations on Topsnut-matrices of graphs and their subgraphs. Suppose that an e-split $k$-connected graph $G$, $k=\gamma_{es}(G)$, induces an \emph{edge-split graph} $G\wedge \{e_i\}^{k}_1$, such that $G\wedge \{e_i\}^{k}_1$ has its own components $G_1,G_2,\dots ,G_m$ with $m\leq k$. Then the Topsnut-matrix of $G$ can be computed as
\begin{equation}\label{eqa:e-split-coincidence}
A_{vev}(G)=\ominus^m_{i=1}A_{vev}(G_i).
\end{equation}
For a \emph{vertex-split graph} $G\wedge \{x_i\}^{k}_1$ having its components $H_1,H_2,\dots ,H_n$ with $k=\gamma_{vs}(G)$, we have
\begin{equation}\label{eqa:v-split-coincidence}
A_{vev}(G)=\odot^n_{i=1}A_{vev}(H_i).
\end{equation}
For a \emph{mixed-split graph} $G\wedge (\{x_i\}^{a}_1\cup \{e_i\}^{b}_1)$ having its components $T_1,T_2,\dots ,T_c$, we have
\begin{equation}\label{eqa:v-split-coincidence}
A_{vev}(G)=[\ominus \odot]^c_{i=1}A_{vev}(T_i).
\end{equation}
Correspondingly, we have $G=\ominus^m_{i=1}G_i$, $G=\odot^n_{i=1}H_i$ and $G=[\ominus \odot]^c_{i=1}T_i$ in topological configuration.

\subsection{Other operations on Topsnut-matrices}

We exchange the positions of two columns $(x_i~e_i~y_i)^{-1}$ and $(x_j~e_j~y_j)^{-1}$ in $A_{vev}(G)$, so we get another Topsnut-matrix $A'_{vev}(G)$. In mathematical symbol, the \emph{column-exchanging operation} $c_{(i,j)}(A_{vev}(G))=A'_{vev}(G)$ is defined by
$$
{
\begin{split}
&\quad c_{(i,j)}(x_1 ~ x_2 ~ \cdots ~x_i~ \cdots x_j~ \cdots ~x_q)\\
&=(x_1 ~ x_2 ~ \cdots ~x_j~ \cdots x_i~ \cdots ~x_q),
\end{split}}
$$
$$
{
\begin{split}
&\quad c_{(i,j)}(e_1 ~ e_2 ~ \cdots ~e_i~ \cdots e_j~ \cdots ~e_q)\\
&=(e_1 ~ e_2 ~ \cdots ~e_j~ \cdots e_i~ \cdots ~e_q),
\end{split}}
$$
and
$$
{
\begin{split}
&\quad c_{(i,j)}(y_1 ~ y_2 ~ \cdots ~y_i~ \cdots y_j~ \cdots ~y_q)\\
&=(y_1 ~ y_2 ~ \cdots ~y_j~ \cdots y_i~ \cdots ~y_q).
\end{split}}
$$
And we exchange the positions of $x_i$ and $y_i$ of the $i$th column of $A_{vev}(G)$ by an \emph{xy-exchanging operation} $l_{(i)}$ defined as:
$$l_{(i)}(x_1 ~ x_2 ~ \cdots ~x_i~ \cdots ~x_q)=(x_1 ~ x_2 ~ \cdots ~y_i~ \cdots ~x_q)$$
and
$$l_{(i)}(y_1 ~ y_2 ~ \cdots ~y_i~ \cdots ~y_q)=(y_1 ~ y_2 ~ \cdots ~x_i~ \cdots ~y_q),$$
the resulting matrix is denoted as $l_{(i)}(A_{vev}(G))$.

Now, we do a series of column-exchanging operations $c_{(i_k,j_k)}$ with $k\in [1,m]$, and a series of xy-exchanging operation $l_{(i_s)}$ with $s\in [1,n]$ to $A_{vev}(G)$, the resulting matrix is written by $T_{(c,l)}(A_{vev}(G))$.
\begin{lem}\label{thm:2-trees-matrices-isomorphic}
Suppose $T$ and $H$ are trees of $q$ edges. If $T_{(c,l)}(A_{vev}(T))=A_{vev}(H)$, then these two trees are isomorphic to each other, that is, $T\cong H$.
\end{lem}
\begin{proof} We use induction on the number of vertices of two trees. As $q=1$, it is trivial. Assume that $T_{(c,l)}(A_{vev}(T-x))=A_{vev}(H-x')$, where $x$ is a leaf of $T$, and $x'$ is a leaf of $H$ such that two trees $T-x$ and $H-x'$ are isomorphic to each other. The condition $T_{(c,l)}(A_{vev}(T))=A_{vev}(H)$ enables us to add leaves $x,x'$ to $T-x$ and $H-x'$ respectively, such that $T-x+x\cong H-x'+x'$, since $T_{(c,l)}(A_{vev}(T))=T_{(c,l)}(A_{vev}(T-x))\odot A_{vev}(e)$ with $e=yx\in E(T)$, and $A_{vev}(H)=A_{vev}(H-x')\odot A_{vev}(e')$ with $e'=y'x'\in E(H)$.

The lemma holds true.
\end{proof}

\begin{thm}\label{thm:2-graphs-matrices-isomorphic}
Let $G$ and $Q$ be two connected graphs of $q$ edges. If there are two edge subsets $E_G\subset E(G)$ and $E_Q\subset E(Q)$, such that a spanning tree $T=G-E_G$ of $G$ and a spanning tree $H=Q-E_Q$ of $Q$ hold $T_{(c,l)}(A_{vev}(T))=A_{vev}(H)$, as well as $T_{(c,l)}(A_{vev}(G))=A_{vev}(Q)$, then these two connected graphs $G$ and $Q$ are isomorphic to each other, that is, $G\cong Q$.
\end{thm}
\begin{proof} We use induction on the number of edges of graphs. If $G$ and $Q$ are trees, we are done by Lemma \ref{thm:2-trees-matrices-isomorphic}. According $T_{(c,l)}(A_{vev}(T))=A_{vev}(H)$ which means $T\cong H$, we take an edge $e_1\in E_G$ and then add it to $T$ for a new graph $T_1=T+e_1$. Next, we take an edge $e'_1\in E_Q$ and then add it to $H$ such that $H_1=H+e'_1$ and $H_1\cong T_1$, since $T\cong H$. Go on in this way, we have $T_i\cong H_i$ with $T_i=T_{i-1}+e_i$ for $e_i\in E_G\setminus \{e_1,\dots ,e_{i-1}\}$ and $H_i=H_{i-1}+e'_i$ for $e'_i\in  E_Q\setminus \{'_1,\dots ,e'_{i-1}\}$. When $E_G\setminus \{e_1,\dots ,e_{m-1}\}=\emptyset $ and $E_Q\setminus \{e'_1,\dots ,e'_{m-1}\}=\emptyset $ for some $m$, we get $G\cong Q$, as desired.
\end{proof}

We point out that Theorem \ref{thm:2-graphs-matrices-isomorphic} is not a solution of the isomorphic problem of graphs, since $G$ and $Q$ are Topsnut-gpws, also, are labelled graphs.

Suppose that $A_{vev}(G)=(X~W~Y)^{-1}$ is a Topsnut-matrix of a Topsnut-gpw $(p,q)$-graph $G$ (see Definition \ref{defn:Topsnut-matrix}(\ref{eqa:a-formula}) and (\ref{eqa:three-vectors})). We define a Topsnut-matrix of an edge $e_i=x_iy_i$ of $G$ as $A(e_i)=(x_i~e_i~y_i)^{-1}$, and set an operation $\odot$ between $A(e_i)$ with $i\in [1,q]$. Hence, we get
\begin{equation}\label{eqa:edge-topsnut-matrix}
{
\begin{split}
A(e_i)\odot A(e_j)&= \left(
\begin{array}{ccccc}
x_{i}\\
w_{i}\\
y_{i}
\end{array}
\right)\odot
\left(
\begin{array}{ccccc}
x_{j}\\
w_{j}\\
y_{j}
\end{array}
\right)
\\
&= \left(
\begin{array}{ccccc}
x_{i}&x_{j}\\
w_{i}&w_{j}\\
y_{i}&y_{j}
\end{array}
\right)\\
&=(x_i~e_i~y_i)^{-1}\odot (x_j~e_j~y_j)^{-1},
\end{split}}
\end{equation}
so we can rewrite the Topsnut-matrix of $G$ in another way
\begin{equation}\label{eqa:dge-topsnut-matrix-0}
{
\begin{split}
A_{vev}(G)=\odot ^q_{i=1}A(e_i).
\end{split}}
\end{equation}
Thereby, we have a vev-type TB-paw
\begin{equation}\label{eqa:TB-paw-by-dges}
{
\begin{split}
D_{vev}(G)&=\uplus^q_{i=1}D_{vev}(e_i)=\uplus^q_{j=1}x_{i_j1}e_{i_j}y_{i_j}\\
&=x_{i_1}e_{i_1}y_{i_1}x_{i_2}e_{i_2}y_{i_2}\cdots x_{i_q}e_{i_q}y_{i_q}.
\end{split}}
\end{equation}
where $x_{i_1}e_{i_1}y_{i_1},x_{i_2}e_{i_2}y_{i_2},\cdots ,x_{i_q}e_{i_q}y_{i_q}$ is a permutation of $x_{1}e_{1}y_{1},x_{2}e_{2}y_{2},\cdots ,x_{q}e_{q}y_{q}$.

We can observe the following facts:

(1) If there is no $(x_i~e_i~y_i)^{-1}=(x_j~e_j~y_j)^{-1}$ for any pair of edges $e_i$ and $e_j$ of $G$, then $G$ is simple.

(2) If any edge $e_s$ corresponds another edge $e_t$ such that both $A(e_s)=(x_s~e_s~y_s)^{-1}$ and $A(e_t)=(x_t~e_t~y_t)^{-1}$ hold one of $x_s=x_t$, $y_s=y_t$, $x_s=y_t$ and $x_t=y_s$, then $G$ is connected.

(3) If $e_i=|x_i-y_i|$ for any edge $e_i$ of $G$ and $\{e_i\}^q_1=[1,q]$ (or $\{e_i\}^q_1=[1,2q-1]^o$), then $G$ is (odd-)graceful; and moreover $G$ is (odd-)elegant if $e_i=x_i+y_i~(\bmod~q)$ for any edge $e_i$ of $G$  and $\{e_i\}^q_1=[0,q-1]$ (or $\{e_i\}^q_1=[1,2q-1]^o$).

To characterize more properties of a $(p,q)$-graph $G$ by its own Topsnut-matrix $A_{vev}(G)=(X~W~Y)^{-1}$ may be interesting.

\section{Encrypting dynamic networks by every-zero graphic groups}

We propose to encrypt a dynamic network in this section although we do not have any known knowledge about such topic. However, we can foresee large scale size and varied constantly, big data base of Topsnut-gpws and changing encryption at any time in the following:

(1) A dynamic network $N(t)$ is variable as time goes on, and $N(t)$, very often, have thousands of vertices and edges at time step $t$.

(2) There is a big data base $D_{ata}(N(t))$ of Topsnut-gpws for encrypting $N(t)$, such that an edge $uv$ of $N(t)$ is labelled by a Topsnut-gpw $G_{uv}$, which can join two labels $G_{u}$ and $G_{v}$ of two ends $u$ and $v$ of the edge together to form an authentication. So the number of elements of $D_{ata}(N(t))$ should be very larger.

(3) It must be quick to encrypt $N(t)$ in short time in real practice, and substitute constantly by new Topsnut-gpws the old Topsnut-gpws to the vertices and edges of $N(t)$ at any time.

Clearly, these three difficult problems will obstruct us to realize our encryption of dynamic networks. We consider it is interesting to explore this topic by our best efforts.

As the first exploration, we will apply spanning trees of dynamic networks as the models of encryption, and use every-zero graphic groups to be as desired data base of Topsnut-gpws, and then change the Topsnut-gpws of $N(t)$ by various every-zero graphic groups under the equivalent coloring/labellings or under the different configurations of graphs (\cite{Yao-Sun-Zhao-Li-Yan-2017, Wang-Xu-Yao-Ars-2018, SH-ZXH-YB-2017-2, SH-ZXH-YB-2017-3, SH-ZXH-YB-2017-4, SH-ZXH-YB-2018-5, Yao-Liu-Yao-2017}).

Since, a Topsnut-gpw $G$ has its own Topsnut-matrices $A(G)$, and each Topsnut-matrix $A(G)$ induces vev-type TB-paws $D{vev}(A(G))$, so we can get our corresponding \emph{every-zero Topsnut-matrix groups} and \emph{every-zero TB-paw groups}, respectively, thus, these two classes of groups can help us to encrypt dynamic networks quickly and efficiently.

\subsection{A pan-odd-graceful every-zero graphic group}

\begin{figure}[h]
\centering
\includegraphics[height=10cm]{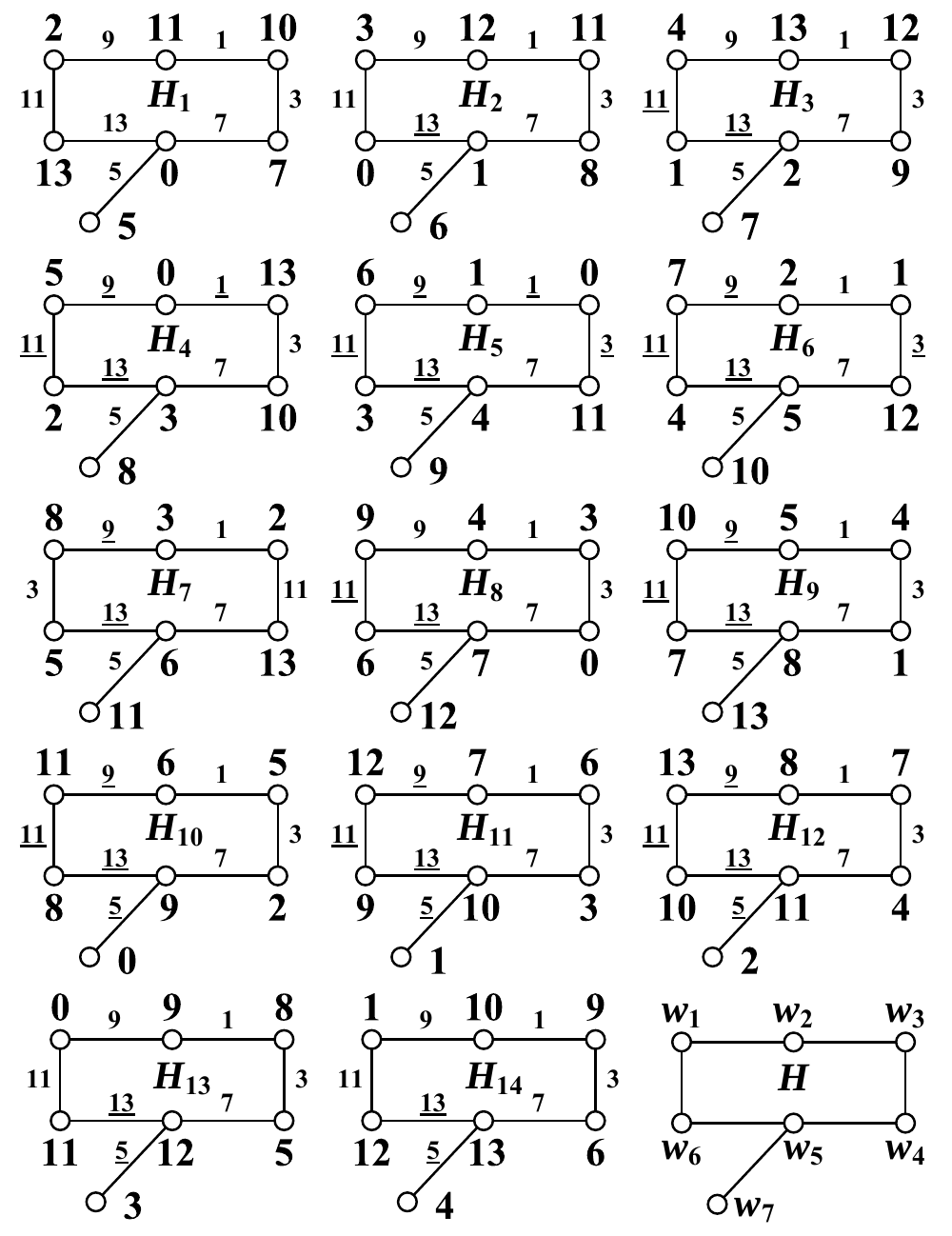}
\caption{\label{fig:odd-graceful-group}{\small A pan-odd-graceful graphic group $F_{14}(H,f)$, also, an \emph{every-zero graphic group} cited from \cite{Yao-Sun-Zhao-Li-Yan-2017}.}}
\end{figure}

Fig.\ref{fig:odd-graceful-group} shows a pan-odd-graceful graphic group $F_{14}(H,f)$ based on a $(7,7)$-graph $H$ and a pan-odd-graceful labelling $f$ of $H$. The pan-odd-graceful graphic group $F_{14}(H,f)$ contains 14 labelled graphs and satisfies: Each element $H_i\in F_{14}(H,f)$ admits a pan-odd-graceful labelling $f_i$, two elements $H_i,H_j\in F_{14}(H,f)$ hold an additive operation $H_i\oplus H_j$ defined as
\begin{equation}\label{eqa:odd-graceful-group}
f_i(w_i)+f_j(w_i)-f_k(w_i)=f_{i+j-k~(\bmod~14)}(w_i)
\end{equation}
for each vertex $w_i\in V(H)$, where $H$ is displayed in Fig.\ref{fig:odd-graceful-group}, and $H_k$ admits a pan-odd-graceful labelling $f_k$ is as the \emph{zero} of $F_{14}(H,f)$. It is easy to verify $H_i\oplus H_k=H_i$, $H_i\oplus H_j=H_j\oplus H_i$, $(H_i\oplus H_j)\oplus H_s=H_i\oplus (H_j\oplus H_s)$. So, $F_{14}(H,f)$ is an Abelian additive group (a graphic group). Notice that each element $H_i\in F_{14}(H,f)$ can be as the \emph{zero} of the graphic group $F_{14}(H,f)$, so we call $F_{14}(H,f)$ a \emph{pan-odd-graceful every-zero graphic group}.

\begin{figure}[h]
\centering
\includegraphics[height=5.4cm]{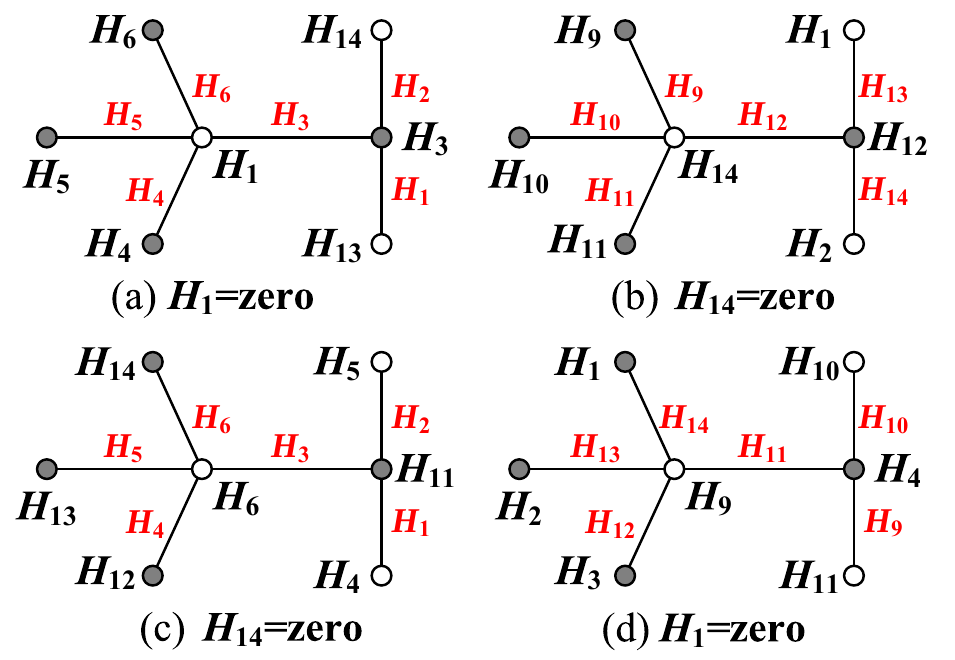}
\caption{\label{fig:group-label-tree}{\small A tree $T$, based on an every-zero graphic group $F_{14}(H,f)$ shown in Fig.\ref{fig:odd-graceful-group}, admits: (a) a graceful group-labelled $F_1$; (b) a dual group-labelled $F'_1$ of $F_1$; (c) another graceful group-labelled $F_2$; (d) a pan-edge-magic group-labelled $F_3$.}}
\end{figure}

Fig.\ref{fig:group-label-tree} shows us the following phenomenon:

(a) $F_1(uv)=F_1(u)\oplus F_1(v)$ for each edge $uv\in E(T)$, and $F_1(E(T))=\{H_1,H_2,\dots ,H_6\}$, so we say $T$ admits a graceful group-labelling.

(b) $F'_1(w)=H_j$ and $F_1(w)=H_i$ with $w\in V(T)\cup E(T)$ such that $H_j=H_{15-i}$, and $F'_1(uv)=F'_1(u)\oplus F'_1(v)$ for each edge $uv\in E(T)$.

(c) $F_2(uv)=F_2(u)\oplus F_2(v)$ for each edge $uv\in E(T)$, and $F_2(E(T))=\{H_1,H_2,\dots ,H_6\}$, so $F_2$ is a graceful group-labelling; moreover each edge $uv$ corresponds another edge $xy$ such that $i'+s+j'=21$ for $F_2(uv)=H_s$, $F_2(x)=H_{i'}$ and $F_2(y)=H_{j'}$.

(d) each edge $uv\in E(T)$ holds $i+t+j=24$ for $F_3(uv)=H_t$, $F_3(u)=H_i$ and $F_3(v)=H_j$, and each edge $uv$ corresponds another edge $xy$ such that $F_3(uv)=F_3(x)\oplus F_3(y)$.

\begin{figure}[h]
\centering
\includegraphics[height=5.4cm]{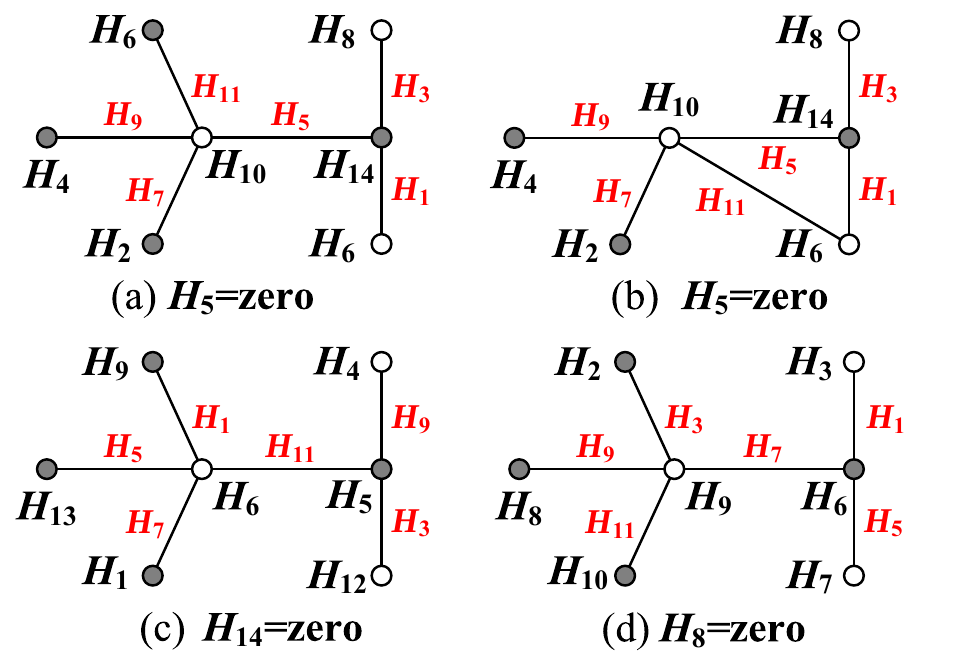}
\caption{\label{fig:group-label-tree-1}{\small (a) The group-labelled tree $T$ has two vertices labelled with the same $H_6$; (b) the graph obtained by identifying two vertices of $T$ into one admits an odd-graceful group-labelling; (c) $T$ has an odd-graceful group-labelling; (d) $T$ has another odd-graceful group-labelling.}}
\end{figure}

\subsection{Graphs labelled by every-zero graphic groups}

Before labelling graphs with graphic groups, we define a particular class of graphic groups as follows:

\begin{defn}\label{defn:graphic-group-definition}
$^*$ An \emph{every-zero graphic group} $F_{n}(H,h)$ made by a Topsnut-gpw $H$ admitting an $\varepsilon$-labelling $h$ contains its own elements $H_i$ holding $H\cong H_i$ and admitting an $\varepsilon$-labelling $h_i$ induced by $h$ with $i\in [1,n]$ and hold an additive operation $H_i\oplus H_j$ defined as
\begin{equation}\label{eqa:graphic-group-definition}
h_i(x)+h_j(x)-h_k(x)=h_{i+j-k\,(\textrm{mod}\,n)}(x)
\end{equation}
for each element $x\in V(H)$ under a \emph{zero} $H_k$.\qqed
\end{defn}

\begin{figure}[h]
\centering
\includegraphics[height=6.8cm]{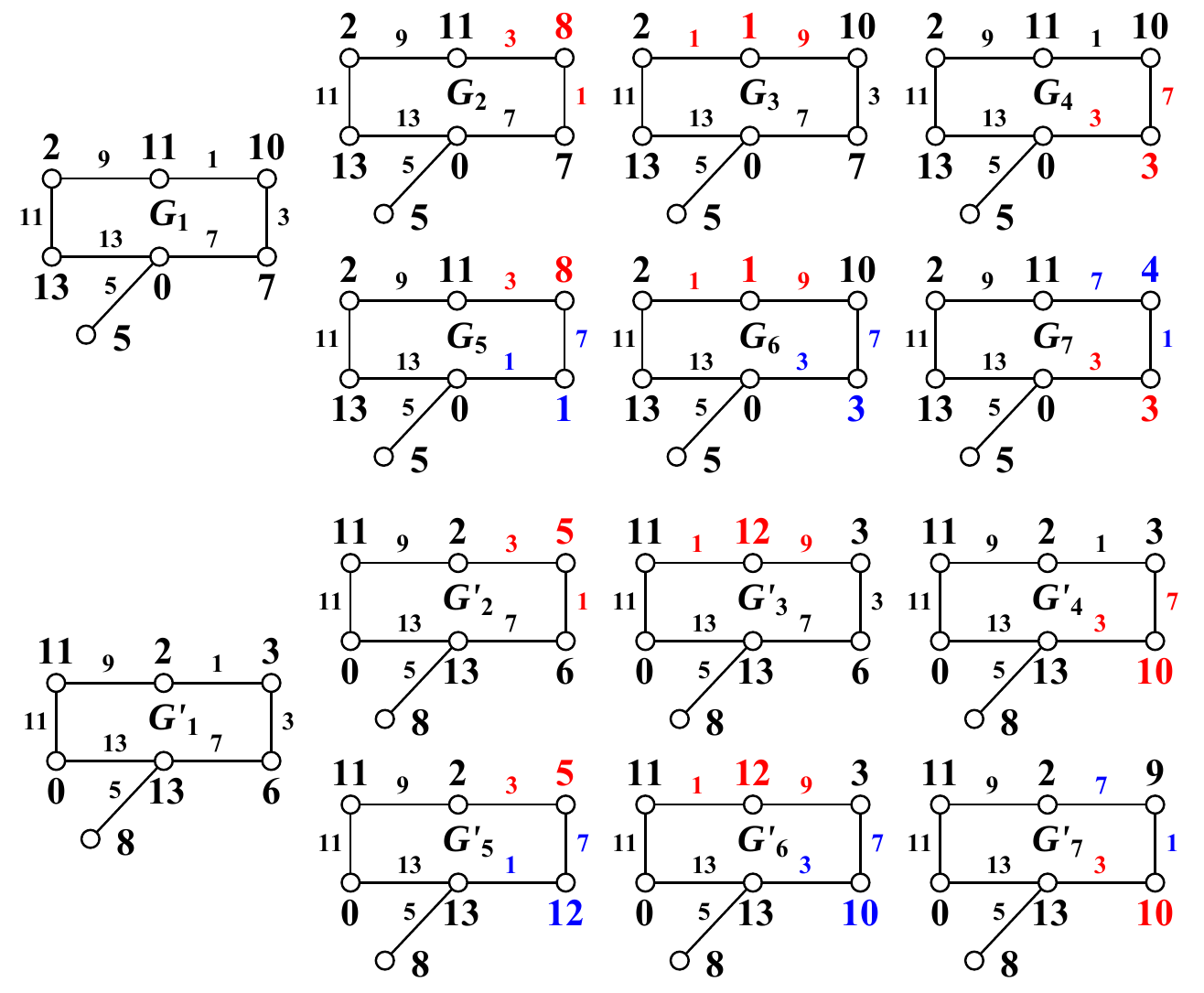}
\caption{\label{fig:dual-translation}{\small Seven pan-odd-graceful graphs $G_i$ and its dual graphs $G'_i$ with $i\in [1,7]$.}}
\end{figure}

We write $G_i(a_i,b_j)G_j$ to stand for a \emph{transformation} from $G_i$ to $G_j$, and vice versa; similarly, $G'_i(a'_i,b'_j)G'_j$ is a transformation from $G'_i$ to $G'_j$ in Fig.\ref{fig:dual-translation}. For example, $G_1(10,8)G_2$, $G_1(11,1)G_2$ and $G_1(7,3)G_2$; $G_2(7,1)G_5$, $G_3(7,3)G_6$ and $G_4(10,4)G_7$. Thereby, there are at least two every-zero graphic group $F_{n}(H,h)$ and $F_{n}(H,h')$, where $h'$ is the dual labelling of $h$.

We may have some \emph{every-zero graphic group chain} $\{F_{n}(H,h_i)\}^m_1$ holding $H_i(a_i,b_{i+1})H_{i+1}$ with $i\in [1,m-1]$. Such every-zero graphic group chains can be used to encrypt network chains, or the subnetworks of a large network.

Furthermore, we can label the vertices and edges of a graph with the elements of a given graphic group.

\begin{defn}\label{defn:graph-graceful-group-labelling}
$^*$ Let $F_{n}(H,h)$ be an every-zero graphic group. A $(p,q)$-graph $G$ admits a \emph{graceful group-labelling} (an \emph{odd-graceful group-labelling}) $F:V(G)\rightarrow F_{n}(H,h)$ such that each edge $uv$ is labelled by $F(uv)=F(u)\oplus F(v)$ under a \emph{zero} $H_k$, and $F(E(G))=\{F(uv):uv \in E(G)\}=\{H_1,H_2,\dots ,H_q\}$ (or $F(E(G))=\{F(uv):uv \in E(G)\}=\{H_1,H_3,\dots ,H_{2q-1}\}$).\qqed
\end{defn}

For understanding Definition \ref{defn:graph-graceful-group-labelling}, we present Fig.\ref{fig:group-label-tree}(a) and (c), as well as Fig.\ref{fig:group-label-tree-1}(b), (c) and (d). Since the group-labelled tree $T$ has two vertices labelled with the same $H_6$ in Fig.\ref{fig:group-label-tree-1}(a), we say that $T$ admits an \emph{odd-graceful group-coloring}. Similarly, we can define the \emph{graceful group-coloring}.

If $n=q$ (or $n=2q-1$) in Definition \ref{defn:graph-graceful-group-labelling}, we say $G$ admits a \emph{pure graceful group-labelling} (or a \emph{pure odd-graceful group-labelling}).

In general, finding the minimum number of the modular $n$ of $F_n(H,h)$ for which a $(p,q)$-graph $G$ admits a \emph{graceful group-labelling} (an \emph{odd-graceful group-labelling}) may be important and interesting. We present an example indicated in Fig.\ref{fig:encrypt-network} for encrypting a network $T$ shown in Fig.\ref{fig:group-label-tree-1}. Notice that there are many ways to join $H_i$ with $H_j$ by an edge $u_iv_j$ for $u_i\in V(H_i)$ and $v_j\in V(H_j)$, so, there are many encrypted networks like $N_{et}(T,F_{14}(H,f))$ shown in Fig.\ref{fig:encrypt-network}. Furthermore, we can get many vv-type/vev-type TB-paws from a graph $H$ having a labelling $f$ and a network $T$.

\begin{figure}[h]
\centering
\includegraphics[height=9.2cm]{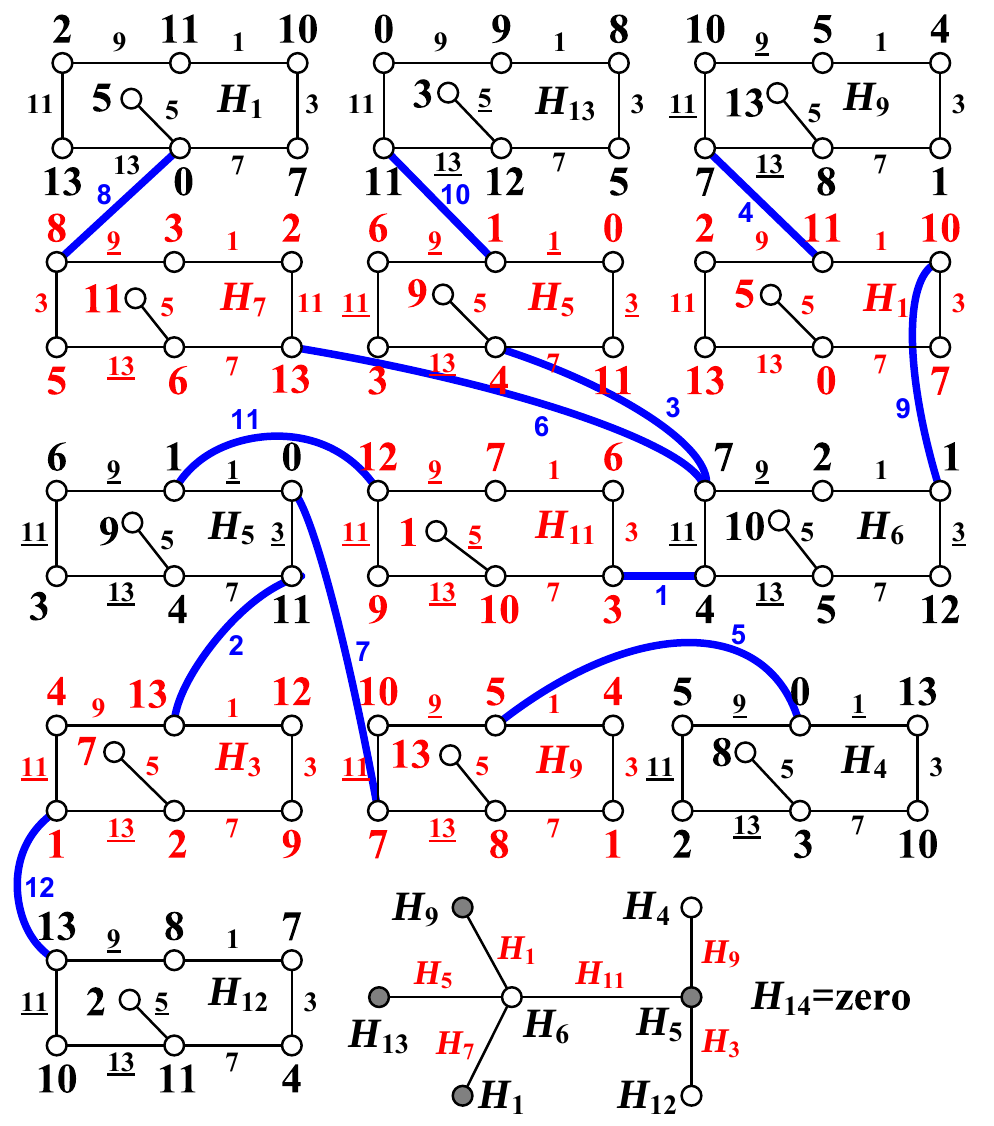}
\caption{\label{fig:encrypt-network}{\small An encrypted network $N_{et}(T,F_{14}(H,f))$ made by an every-zero graphic group $F_{14}(H,f)$ exhibited in Fig.\ref{fig:odd-graceful-group}, the labels of the joined edges in blue color form a consecutive set $[1,12]$.}}
\end{figure}

The encrypted network $N_{et}(T,F_{14}(H,f))$ can provide more vev-type TB-paws with longer bytes by the previous methods introduced, such as Path-neighbor-method, Cycle-neighbor-method, Lobster-neighbor-method and Spider-neighbor-method, as well as Euler-Hamilton-method. For example, $N^*=N_{et}(T,F_{14}(H,f))$ shown in Fig.\ref{fig:encrypt-network} distributes us the following vev-type TB-paws (we write $D_{vev}$ as $D$ for short):

{\small \noindent $D(H_1)=1311291111037705513$ (19), $D(a_1a_7)=088$ (4),

\noindent $D(H_7)=53893121113765116135$ (20), $D(a_7a_6)=1367$ (4),

\noindent $D(H_6)=41179211312755105134$ (20), $D(a_6a_5)=734$ (3),

\noindent $D(H_5)=3116911031174594133$ (19), $D(a_5a_{13})=11011$ (5),

\noindent $D(H_{13})=1111099183571253121311$ (22), $D(a_6a_{1})=1910$ (4),

\noindent $D(H_{1})=1311291111037705501313$ (22), $D(a_1a_{9})=1147$ (4),

\noindent $D(H_{9})=71110951431785138137$ (20), $D(a_6a_{11})=413$ (3),

\noindent $D(H_{11})=9111297163371051139$ (19), $D(a_{11}a_5)=413$ (3),

\noindent $D(H_{5})$ (19), $D(a_{5}a_3)=31013$ (5),

\noindent $D(H_{3})=11149131123972572131$ (20), $D(a_3a_{12})=178$ (3),

\noindent $D(H_{12})=10111398173471152111310$ (23),
$D(a_5a_{9})=325$ (3),

\noindent $D(H_{9})=71110951431785138137$ (20), $D(a_9a_{4})=325$ (3),

\noindent $D(H_{4})=21159011331073583132$ (20). }

Thereby, we have a vev-type TB-paw as follows
\begin{equation}\label{eqa:c3xxxxx}
{
\begin{split}
D(N^*)=&D(H_1)\uplus D(a_1a_7)\uplus D(H_7)\uplus D(a_7a_6)\\
&\uplus D(H_6)\uplus D(a_6a_5)\uplus D(H_5)\uplus D(a_5a_{13})\\
&\uplus D(H_{13})\uplus D(a_6a_{1})\uplus D(H_{1})\uplus D(a_1a_{9})\\
&\uplus D(H_{9})\uplus D(a_6a_{11})\uplus D(H_{11})\\
&\uplus D(a_{11}a_5)\uplus D(H_{5})\uplus D(a_{5}a_3)\uplus D(H_{3})\\
&\uplus D(a_3a_{12})\uplus D(H_{12})\uplus D(a_5a_{9})\uplus D(H_{9})\\
&\uplus D(a_9a_{4})\uplus D(H_{4}),
\end{split}}
\end{equation}
by $N_{et}(T,F_{14}(H,f))$, such that $D(N^*)$ has at least $263+44=307$ bytes in total. Clearly, there are many ways to write $D(N^*)$, since there are many ways to write $D(H_i)$ and $D(a_ia_j)$ for $1\leq i,j\leq 13$, and there are many ways to combinatoric $D(H_i)$ and $D(a_ia_j)$ for producing  $D(N^*)$.

\subsection{Encryption of Tree-networks}

The topic of encrypting dynamic networks has been proposed in \cite{Yao-Sun-Zhang-Mu-Sun-Wang-Su-Zhang-Yang-Yang-2018arXiv}. We show a general definition on every-zero graphic groups as follows:

\begin{defn}\label{defn:general-group-labelling}
$^*$ Let $F_{n}(H,h)$ be an every-zero graphic group, and $\{H_{i_j}\}^q_1$ be a subset of $F_{n}(H,h)$. Suppose that a $(p,q)$-graph $G$ admits a mapping $F:V(G)\rightarrow F_{n}(H,h)$ such that each edge $uv$ is labelled by $F(uv)=F(u)\oplus F(v)$ under a \emph{zero} $H_k$. If $F(x)\neq F(y)$ for any pair of vertices $x,y$, and $F(E(G))=\{H_{i_j}\}^q_1$, we call $F$ an \emph{$\{H_{i_j}\}^q_1$-sequence group-labelling}; if $F(w)=F(z)$ for some two distinct vertices $w,z$, and $F(E(G))=\{H_{i_j}\}^q_1$, we call $F$ an \emph{$\{H_{i_j}\}^q_1$-sequence group-coloring}. The labelled graph made by joining $F(u)$ with $F(uv)$ and joining $F(uv)$ with $F(v)$ for each edge $uv\in E(G)$ is denoted as $N_{et}(G,F_{n}(H,h))$ (see an example depicted in Fig.\ref{fig:encrypt-network}).\qqed
\end{defn}

\begin{thm} \label{them:trees-sequence-group-coloring}
For any sequence $\{H_{i_j}\}^q_1$ of an every-zero graphic group $F_{n}(H,h)$, any tree having $q$ edges admits an $\{H_{i_j}\}^q_1$-sequence group-coloring or an $\{H_{i_j}\}^q_1$-sequence group-labelling.
\end{thm}
\begin{proof} Let $T_q$ be a tree having $q$ edges, and each $H_i\in F_{n}(H,h)$ admit an $\varepsilon$-labelling $h_i$. For $q=1$, $T_1$ has two vertices $u,v$ and a unique edge $uv$. We define a labelling $F_1$ of $T_1$ under the zero $H_1$ such that $F_1(u)=H_1$, $F_1(uv)=H_s$ with $s\neq 1$, then $F_1(u)\oplus F_1(v)=F_1(uv)$, assume $F_1(v)=H_j$, moreover
\begin{equation}\label{eqa:sequence-group-coloring-1}
h_1(x)+h_j(x)-h_1(x)=h_{s~(\bmod~q)}(x)
\end{equation}
so $1+j-1=s~(\bmod~q)$, we get $F_1(v)=H_s$. Thereby, $F_1$ is an $H_s$-sequence group-labelling, since $H_1\neq H_s$.

Suppose that any tree $T_{q-1}$ of $(q-1)$ edges admits one of an $\{H_{i_j}\}^{q-1}_1$-sequence group-coloring or an $\{H_{i_j}\}^{q-1}_1$-sequence group-labelling, here $T_{q-1}=T_q-x$ for a leaf $x$ of $T$. Let $y$ be the unique adjacent vertex of $x$. Notice that $T_{q-1}$ admits an $\{H_{i_j}\}^{q-1}_1$-sequence group-coloring $F_{q-1}$ with the sequence $\{H_{i_j}\}^{q-1}_1$ and the zero $H_k$. Let $F_{q-1}(y)=H_y$. We define a new group-coloring $F_{q}$ by setting $F_{q}(w)=F_{q-1}(w)$ for each element of $V(T_{q-1})\cup E(T_{q-1})$. Let $F_{q-1}(y)=H_a$, we set $F_{q}(xy)=H_{i_q}$. Assume $F_{q}(xy)=H_b$, we will find the exact value of $b$. By
\begin{equation}\label{eqa:sequence-group-coloring-2}
h_a(u)+h_b(u)-h_k(u)=h_{i_q~(\bmod~q)}(u)
\end{equation}
that is, $a+b-k=i_q~(\bmod~q)$, we get the solution $b=k-a+i_q~(\bmod~q)$. Hence, $F_{q}(xy)=F_{q}(y)\oplus F_{q}(x)$. According to the hypothesis of induction, the theorem is really correct.
\end{proof}

By the induction proof on Theorem \ref{them:trees-sequence-group-coloring}, we can set randomly the elements of the set $\{H_1,H_2,\dots ,H_q\}=\{H_{i}\}^q_1$ on the edges of any tree $T$ having $q$ edges, where $H_i\neq H_j$ for $i\neq j$, and then label the vertices of $T$ with the elements of an every-zero graphic group $F_{n}(H,h)$. We provide a sequence group-coloring $F$ of $T$ through the following algorithm.

\vskip 0.2cm

\textbf{TREE-GROUP-COLORING algorithm.}

\textbf{Input:} A tree $T$ of $q$ edges, and $\{H_{i}\}^q_1\subseteq  F_{n}(H,h)$.

\textbf{Output:} An $\{H_{i}\}^q_1$-sequence group-coloring (or group-labelling) of $T$.

\begin{verse}
\emph{Step 1.} Select an initial vertex $u_1\in V(T)$, its neighbor set $N(u_1)=\{v_{1,1},v_{1,2},\dots ,v_{1,d_1}\}$, where $d_1$ is the degree of the vertex $u_1$; next, select $H_1$ as the zero, and label $u_1$ with $F(u_1)=H_1$ and $F(u_1v_{1,j})=H_{1,j}$ with $j\in [1,d_1]$. From
$$h_1(w)+h_z(w)-h_1(w)=h_{i_j~(\bmod ~n)}(w),$$
where $h_z$ is the labelling of $H_{1,j}$, immediately, we get solutions $z={i_j}$, that is $F(v_{1,j})=H_{1,j}$ with $j\in [1,d_1]$. Let $V_1\leftarrow V(T)\setminus \{u_1,v_{1,1},v_{1,2},\dots ,v_{1,d_1}\}$, $L_1\leftarrow \{F(u_1v_{1,j})\}^{d_1}_1$.

\emph{Step 2.} If $V_{k-1}=\emptyset$ (resp. $L_{k-1}=\{H_{i}\}^q_1$), go to Step 4.

\emph{Step 3.}  If $V_{k-1}\neq \emptyset$ (resp. $L_{k-1}\neq \{H_{i}\}^q_1$), select $u_k\in V_{k-1}$ such that $N(u_k)=\{v_{k,1},v_{k,2},\dots ,v_{k,d_k}\}$ contains the unique vertex $v_{k,1}$ being labelled with $F(v_{k,1})=H_{\alpha}$. Label $F(u_kv_{k,j})=H_{k,j}\in \{H_{i}\}^q_1\setminus L_{k-1}$ with $j\in [1,d_k]$. Assume $F(u_k)=H_{k}$, solve
$$h_k(w)+h_{\alpha}(w)-h_1(w)=h_{k,1~(\bmod ~n)}(w),$$
then $k+\alpha-1=(k,1)~(\bmod ~n)$, thus, $k=1-\alpha+(k,1)~(\bmod ~n)$. Next,  solve
$$h_k(w)+h_z(w)-h_1(w)=h_{k£¬j~(\bmod ~n)}(w),$$
where $h_z$ is the labelling of $F(v_{k,j})$ with $j\geq 2$. Then $k+z-1=(k,j)~(\bmod ~n)$, so $z=1-k+(k,j)~(\bmod ~n)$, also, $F(v_{k,j})=H_{1-k+(k,j)~(\bmod ~n)}$ with $j\in [2,d_k]$. Let $V_{k}\leftarrow V(T)\setminus \big (V_{k-1}\cup \{u_k,v_{k,2},\dots ,v_{k,d_k}\}\big )$, and $L_k\leftarrow L_{k-1}\cup \{F(u_kv_{k,j}):j\in [1,d_k]\}$, go to Step 2.

\emph{Step 4.} Return an $\{H_{i}\}^q_1$-sequence group-coloring $F$ of $T$.
\end{verse}

\vskip 0.2cm

As a consequence, the TREE-GROUP-COLORING algorithm is polynomial and efficient, and it can quickly set Topsnut-gpws to a tree-like network. In Fig.\ref{fig:encrypt-network}, we can see ``$H_1-H_7-H_6$'', called a block joined by two edges having labels 6 and 8. In real operation of encrypting a network, we can use two or more edges to join $H_1$, $H_7$ and $H_6$ together as desired as possible.

Our encrypting a network is in the way: We select a spanning tree $T$ from a network $N(t)$ at time step $t$, and encrypt $T$ by an every-zero graphic group $F_n(H,f)$ to obtain an encrypted tree-like network $N_{et}(T,F_n(H,f))$. For example, we select $T$ to be a caterpillar, or a spider, or a lobster, and so on. And furthermore we label $T$ by a determined labelling $f:V(T)\rightarrow F_n(H,f)$, such that $f(i)=H_i$, $f(j)=H_j$ and $f(ij)=H_{ij}$ obtained from $f(i)$ and $f(j)$, correspondingly, we get a vv-type/vev-type TB-paw
$$D(H_i)\uplus D(a_ib_{ij})\uplus D(H_{ij})\uplus D(a_{ij}b_j)\uplus D(H_j),$$
where $a_i$ is a vertex of $H_i$, $b_{ij}$ is a vertex of $H_{ij}$, $a_{ij}$ is a vertex of $H_{ij}$, and $b_j$ is a vertex of $H_j$ (see Fig.\ref{fig:graph-labelling-join-method}).

\begin{figure}[h]
\centering
\includegraphics[height=1.2cm]{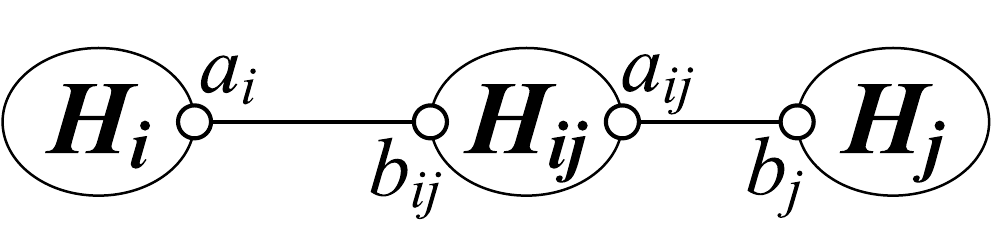}
\caption{\label{fig:graph-labelling-join-method}{\small  A scheme of joining $H_i,H_{ij}$ and $H_j$.}}
\end{figure}

Spanning trees of dynamic networks have been investigated for a long time (\cite{Ma-Bing-Yao-Physica-A-2017, Ma-Bing-Yao-Computer-2018, Ma-Bing-Yao-Physica-A-2018, Ma-Wang-Wang-Yao-Theoretical-Computer-Science-2018, Ma-Su-Hao-Yao-Physica-A-2018}), those spanning trees admitting power-law and having scale-free feature are useful for encrypting networks. The nodes having larger degrees in a scale-free network control nodes over $80$ per center (\cite{Barabasi-Bonabeau2003}), so they can be considered to form a center of public keys in encrypting dynamic networks, see Fig.\ref{fig:scale-free-tree}(b) and Fig.\ref{fig:scale-free-tree-1}(b)-(d). However, it is a big challenge to enumerate the number of spanning trees of a dynamic network $N(t)$ at time step $t$, and very difficult to figure out these non-isomorphic spanning trees, even for particular spanning trees, such as spanning trees to be: caterpillars, lobsters, spiders, trees having maximum leaves, trees having the shortest diameters, and so on.

We provide three algorithms for finding particular spanning trees in Appendices A, B and C.

In the article \cite{Douglas-Robert-J1992}, the authors have shown that a minimal connected dominating set $S$ and a spanning tree $T^{\max}$ having maximal leaves in a connected graph $G$ hold $|G|=|S|+|\mathcal {L}(T^{\max})|$. However, finding a spanning tree $T^{\max}$ having maximal leaves is a NP-problem (\cite{Garey-Johnson1979}). The authors in
\cite{Fernau-Kneis-Kratsch-Langer-Liedloff-Raible2011} applied a technique, called \emph{measure-and-conquer technique}, to distribute an exact algorithm of complex $O(1.8966^n)$ for finding a spanning tree $T^{\max}$ having maximal leaves in a network with $n$ vertices.

\begin{figure}[h]
\centering
\includegraphics[height=2.6cm]{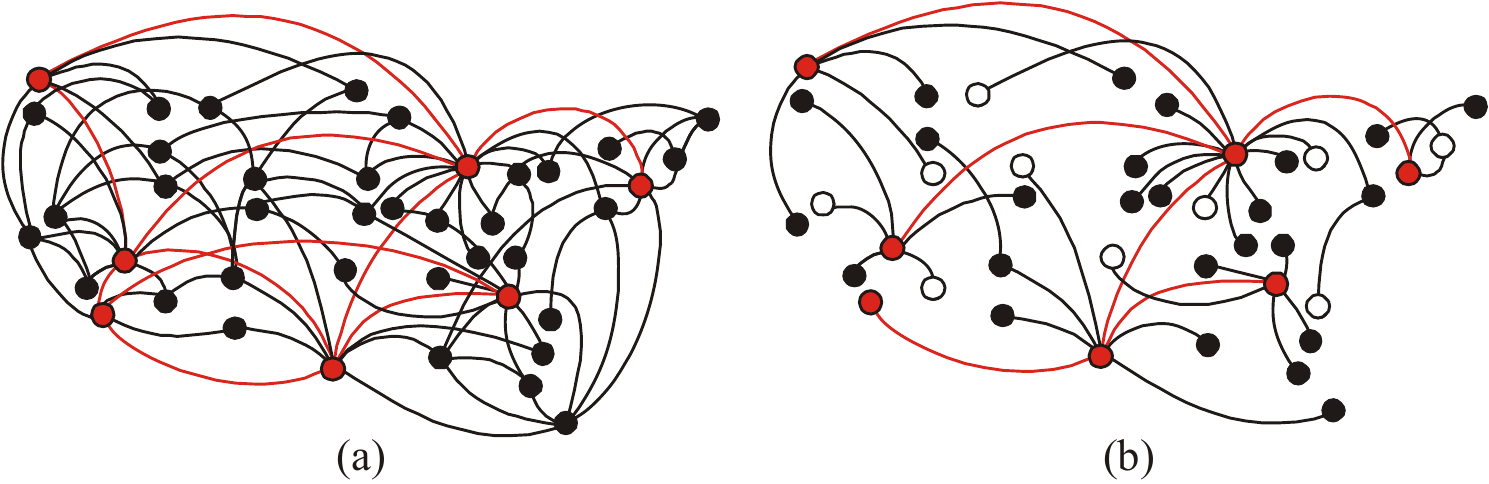}
\caption{\label{fig:scale-free-tree}{\small (a) A scale-free network $N$ \cite{A-L-Barabasi-R-Albert1999}; (b) a spanning (scale-free) tree of $N$.}}
\end{figure}

Three spanning trees pictured in Fig.\ref{fig:scale-free-tree-1}(b)-(d) are lobsters, so they admit odd-graceful labellings and odd-elegant labellings (\cite{Zhou-Yao-Chen-Tao2012, Zhou-Yao-Chen-2013}). Thereby, we have three Topsnut-gpws $T_1,T_2,T_3$ made by three spanning trees shown in Fig.\ref{fig:scale-free-tree-1}(b)-(d), and these three Topsnut-gpws $T_1,T_2,T_3$ can distribute us complex vv-type/vev-type TB-paws by the previous methods. Next, we label $T_i$ with $i\in [1,3]$ by an every-zero graphic group $F_n(H,f)$ with large scale $n$, and we get $N_{et}(T_i,F_n(H,f))$ with $i\in [1,3]$. It is not difficult to see that each $N_{et}(T_i,F_n(H,f))$ with $i\in [1,3]$ can degenerate vv-type/vev-type TB-paws in more complex.

\begin{figure}[h]
\centering
\includegraphics[height=8.4cm]{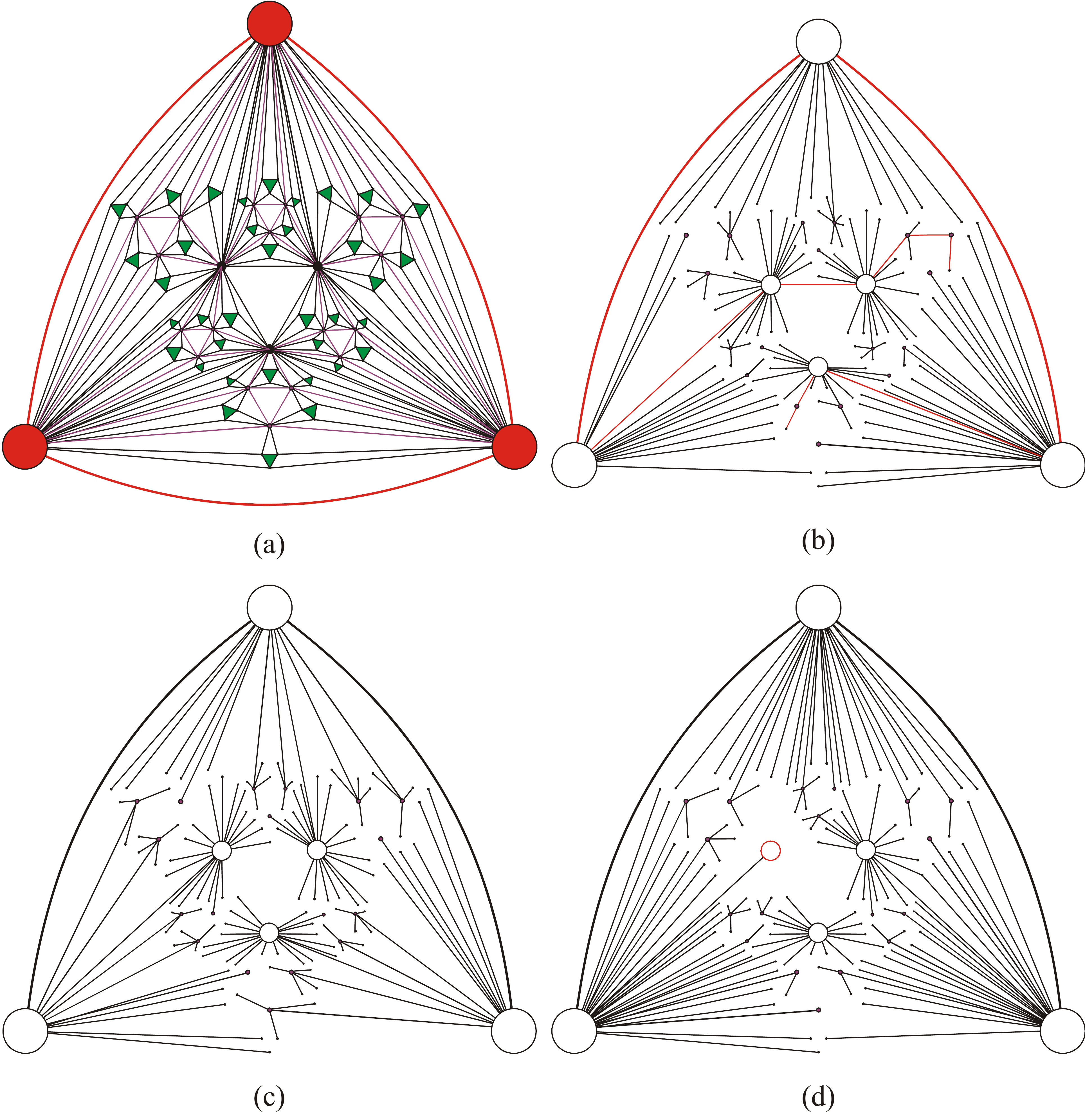}
\caption{\label{fig:scale-free-tree-1}{\small (a) A scale-free network $S$ of 132 vertices, also a Sierpinski model ( \cite{Zhang-Zhou-Fang-Guan-Zhang-2007}); (b)-(d) three spanning (scale-free) trees of $S$ having maximal leaves, they have different diameters.}}
\end{figure}

For particular sequence $\{H_i\}^{q}_1$ and particular graphs, we can determine such particular graphs admitting  $\{H_i\}^{q}_1$-sequence group-labellings.

\begin{thm}\label{thm:K(m,n)-group-labelling}
For any sequence $S=\{H_{a+(i-1)b}\}^{q}_1$, where $H_{a+(i-1)b}$ belongs to an every-zero graphic group $F_n(H,f)$, each complete bipartite graph $K_{m,n}$ with $mn=q$ admits an $S$-sequence group-labelling.
\end{thm}
\begin{proof} We write the vertex set $V(K_{m,n})=\{u_i,v_j:~i\in [1,m],j\in[1,n]\}$ and edge set $E(K_{m,n})=\{u_iv_j:~i\in [1,m],j\in[1,n]\}$ of $K_{m,n}$. Without loss of generality, $m\leq n$, we select $H_a$ as the \emph{zero}, and define a labelling $h$ of $K_{m,n}$ as:

(1) $f(u_1)=H_a$, $f(v_j)=H_{a+(j-1)b}$, and $f(u_{1}v_j)=H_{a+(j-1)b}=f(u_1)\oplus f(v_j)$ with $j\in[1,n]$.

(2) $f(u_{k+1})=H_{a+knb}$ with $k\in[1,m-1]$, $H_{a+knb}\oplus H_{a+(j-1)b}=H_{a+(kn+j-1)b}$ with $j\in[1,n]$: $H_{a+knb}\oplus H_{a}=H_{a+knb}$, $H_{a+knb}\oplus H_{a+b}=H_{a+(kn+1)b}$, $H_{a+knb}\oplus H_{a+2b}=H_{a+(kn+2)b}$, $\dots$, $H_{a+knb}\oplus H_{a+(n-1)b}=H_{a+[(k+1)n-1]b}$.

(3) $f(u_{k+1}v_j)=H_{a+(kn+j-1)b}=H_{a+knb}\oplus H_{a+(j-1)b}=f(u_{k+1})\oplus f(v_j)$ with $k\in[1,m-1]$ and $j\in[1,n]$.

It is not difficult to verify that $f$ is just an $S$-sequence group-labelling, as desired.
\end{proof}

\begin{thm}\label{thm:ring-like-group-labelling}
Suppose $S^*=\{H_{i_j}\}^{m}_1$ is a subsequence of a sequence $S=\{H_{i}\}^{q}_1$ from an every-zero graphic group $F_n(H,f)$, if the unique cycle $C_m$ of a ring-like network $N_{ring}$ of $q$ edges admits an $S^*$-sequence group-coloring (or  group-labelling), then $N_{ring}$ admits an $S$-sequence group-labelling.
\end{thm}
\begin{proof} Let $h:V(C_m)\rightarrow S^*=\{H_{i_j}\}^{m}_1$ be a an $S^*$-sequence group-labelling of $C_m$. We divide the remainder elements of the sequence $S\setminus S^*$ into $m$ groups $S_i$ with $|S_i|=|E(T_i)|$ and $i\in [1,m]$. Next, we use the TREE-GROUP-COLORING algorithm to label each tree $T_i$ after distributing the elements of $S_i$ to the edges of $T_i$ by one-vs-one based on the every-zero graphic group $F_n(H,f)$, finally, we get a desired $S$-sequence group-coloring (or  group-labelling) of the ring-like network $N_{ring}$.
\end{proof}

By the method in the proof of Theorem \ref{thm:ring-like-group-labelling}, we can prove: A \emph{generalized ring-like network} $N^*_{gring}$ has a connected graph $G$ such that the deletion of all vertices of $G$ from $N^*_{gring}$ results in a forest (a forest $H$ is disconnected graph, and each component of $H$ is just a tree). If $G$ admits an $S^*$-sequence group-labelling, then $N_{gring}$ admits an $S$-sequence group-labelling, where $S^*=\{H_{i_j}\}^{m}_1\subset S=\{H_{i}\}^{q}_1 \subseteq F_n(H,f)$.

\begin{figure}[h]
\centering
\includegraphics[height=3.4cm]{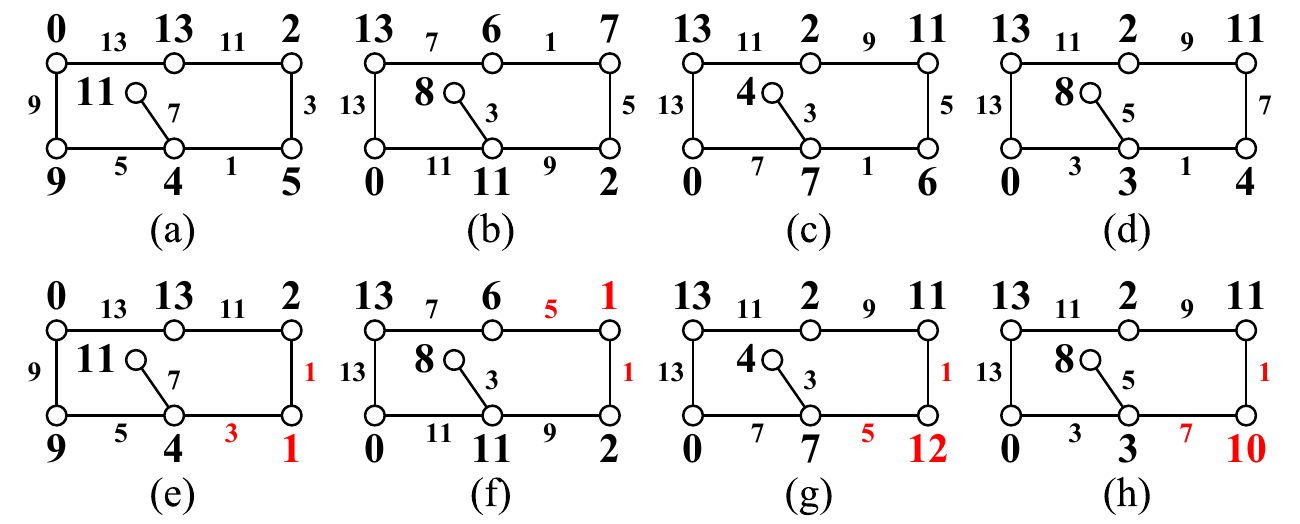}
\caption{\label{fig:H-other-odd-graceful}{\small Other eight odd-graceful labellings of $H$ displayed in Fig. \ref{fig:group-label-tree}.}}
\end{figure}

In Fig.\ref{fig:H-other-odd-graceful}(g), $H$ admits an odd-graceful labelling $f_g:V(H)\rightarrow [0,13]$ such that $f_g(E(H))=[1,13]^o$ and $\{|a-b|:~a,b\in f_g(V(H))\}=[1,13]$, we call $f_g$ a \emph{perfect odd-graceful labelling}. Similarly, we can define a \emph{perfect odd-elegant labelling}, and so on. Thereby, we propose the following new labellings:

\begin{defn}\label{defn:perfect-odd-graceful-labelling}
$^*$ Let $f$ be an odd-graceful labelling of a $(p,q)$-graph $G$, such that $f(V(G))\subset [0,2q-1]^o$ and $f(E(G))=[1,2q-1]^o$. If $\{|a-b|:~a,b\in f(V(G))\}=[1,p]$, then $f$ is called a \emph{perfect odd-graceful labelling} of $G$.\qqed
\end{defn}

\begin{defn}\label{defn:perfect-varepsilon-labelling}
$^*$ Suppose that a $(p,q)$-graph $G$ admits an $\varepsilon$-labelling $h: V(G)\rightarrow S\subseteq [0,p+q]$. If $\{|a-b|:~a,b\in f(V(G))\}=[1,p]$, we call $f$ a \emph{perfect $\varepsilon$-labelling} of $G$.\qqed
\end{defn}

\subsection{Complexity of encrypted networks by every-zero graphic groups}

The complexity of Theorem \ref{them:trees-sequence-group-coloring} is $n^2\cdot q!$, since there are: each edge labelling $f_i:E(G)\rightarrow \{H_{i}\}^q_1$, each \emph{zero} $H_j$, each initial vertex $H_k$.

Theorem \ref{them:trees-sequence-group-coloring} and the TREE-GROUP-COLORING algorithm show that any tree-like network can be encrypted by every-zero graphic groups. We point that encrypted networks $N_{et}(G,F_{n}(H,h))$ have the following advantages for withstanding decryption:
\begin{asparaenum}[\textbf{Com}-1. ]
\item An encrypted network $N_{et}(G,F_{n}(H,h))$ have a large number of vertices.
\item Each element $H_k$ of $F_{n}(H,h)$ can be considered as the ``zero'', so there exist $M$ encrypted network $N_{et}(G,F_{n}(H,h))$ for a fixed every-zero graphic group $F_{n}(H,h)$, where $M=|F_{n}(H,h)|$.
\item There are many sequences $\{H_{i}\}^q_1$ of $F_{n}(H,h)$, and there are many permutations $\{H_{i_j}\}^q_1$ of $\{H_{i}\}^q_1$ to label the edges of $G$ by the TREE-GROUP-COLORING algorithm.
\item There are many labellings $h_1,h_2,\dots ,h_m$ of $H$ to form $F_{n}(H,h_i)$ with $i\in [1,m]$ such that $h, h_1,h_2,\dots ,h_m$ belong to the same class $C_{lass}$, see Fig.\ref{fig:dual-translation} and Fig.\ref{fig:H-other-odd-graceful}.
\item $H$ may admits many labellings that do not belong to $C_{lass}$, such as graceful labelling, odd-elegant labelling, edge-magic total labelling, and so on.
\item There are many graphs $H^*$ like $H$ that can form $F_{n}(H^*,h_j)$.
\item In an encrypted network $N_{et}(G,F_{n}(H,h))$, there many ways to join $F(u)$ with $F(uv)$ by edges and to join $F(uv)$ with $F(v)$ by edges, so we have many encrypted network $N_{et}(G,F_{n}(H,h))$. Such joining method can interrupt an attack that has decrypted the Topsnut-gpws $H_i$ on some vertices.
\item There are many ways to generate vv-type/vev-type TB-paws from an encrypted network $N_{et}(G,F_{n}(H,h))$.
\item If $G=T(t)$ is a spanning tree of a dynamic network $N(t)$ at time step $t$, then $N_{et}(T(t),F_{n}(H,h))$ can be considered as a network password of $N(t)$ at time step $t$. Since there are $a(t)$ spanning trees of $N(t)$ at a fixed time step $t$, so we have $a(t)$ encrypted networks of the form $N_{et}(T(t),F_{n}(H,h))$.
\end{asparaenum}

The facts listed above indicate that it is not easy to attack encrypted networks $N_{et}(G,F_{n}(H,h))$, in other words, encrypted networks $N_{et}(G,F_{n}(H,h))$ are \emph{provable security}.

\subsection{Encrypting networks by pan-matrices}

Motivated from Topsnut-matrices $A_{vev}(G)$, we can define so-called \emph{pan-matrices} for encrypting networks.

\begin{defn}\label{defn:graphic-group-matrix}
$^*$ For an every-zero graphic group $F_n(H,h)=\{H_i,H'_i,H''_i\}^n_1$, a \emph{graphic group-matrix} $P_{vev}(G)$ of a $(p,q)$-graph $G$ is defined as $P_{vev}(G)=(X_P,W_P,Y_P)^{-1}$ with
\begin{equation}\label{eqa:two-vectors}
{
\begin{split}
&X_P=(H_1 ~ H_2 ~ \cdots ~H_q), W_P=(H'_1 ~ H'_2 ~ \cdots ~H'_q)\\
&Y_P=(H''_1 ~ H''_2 ~\cdots ~ H''_q),
\end{split}}
\end{equation}
where each edge $u_iv_i$ of $G$ with $i\in [1,q]$ is labelled by $H'_i$, and its own two ends $u_i$ and $v_i$ are labelled by $H_i$ and $H''_i$ respectively; and $G$ has another \emph{graphic group-matrix} $P_{vv}(G)$ defined as $P_{vv}(G)=(X_P,Y_P)^{-1}$, where $X_P,Y_P$ are called \emph{pan-v-vectors}, $W_P$ is called \emph{pan-e-vector}.\qqed
\end{defn}

See Fig.\ref{fig:group-label-tree} and Fig.\ref{fig:group-label-tree-1} for understanding Definition \ref{defn:graphic-group-matrix}. Notice that the  every-zero graphic group $F_{14}(H,f)$ indicated in Fig.\ref{fig:odd-graceful-group} corresponds an \emph{every-zero matrix group} $A_n(H,f)$, in which every element $A_{vev}(H_i)$ is the Topsnut-matrix of Topsnut-gpw $H_i\in F_{14}(H,f)$.

For an every-zero matrix group $$A_n(G,f)=\{A_{vev}(G_i),A_{vev}(G'_i),A_{vev}(G''_i)\}^n_1$$ we define a \emph{matrix group-matrix} $M_{vev}(G)$ of a $(p,q)$-graph $G$ as: $M_{vev}(G)=(X_M,W_M,Y_M)^{-1}$ with

$X_M=(A_{vev}(G_1) ~ A_{vev}(G_2) ~ \cdots ~A_{vev}(G_q))$,

$W_M=(A_{vev}(G'_1) ~ A_{vev}(G'_2) ~ \cdots ~A_{vev}(G'_q))$,

$Y_M=(A_{vev}(G''_1) ~ A_{vev}(G''_2) ~ \cdots ~A_{vev}(G''_q))$, \\
where each edge label $A_{vev}(G'_i)$ has its own two end labels $A_{vev}(G''_i)$ and $A_{vev}(G_i)$ with $i\in [1,q]$; and $G$ has another \emph{matrix group-matrix} $P_{vv}(G)$ defined as $M_{vv}(G)=(X_M,Y_M)^{-1}$, where $X_M,Y_M$ are called \emph{pan-v-vectors}, $W_M$ is called \emph{pan-e-vector}.

\subsection{New graphic groups made by encrypting networks}

Motivated from Fig.\ref{fig:encrypt-network}, we show the following result:

\begin{thm}\label{thm:group-labelling-by-group-labelling}
Suppose that $N=N_{et}(G,F_{n}(H,f))$ is an encrypted network labelled by a group-labelling $g_{roup}$ based on an every-zero graphic group $F_{n}(H,f)$ under a zero $H_{k}\in F_{n}(H,f)$. Then we have an every-zero graphic group $F_{n}(N,g_{roup})$.
\end{thm}
\begin{proof}By the hypothesis of the theorem, we have a group-labelling $g_{roup}:V(G)\rightarrow F_{n}(H,f)$, $g_{roup}:E(G)\rightarrow \{H_i\}^q_1$ with $\{H_i\}^q_1\subset F_{n}(H,f)$ for an every-zero graphic group $F_{n}(H,f)$, such that $g_{roup}(u)\neq g_{roup}(v)$ for each edge $uv\in E(G)$, $g_{roup}(uv)=g_{roup}(u)\oplus g_{roup}(v)\in g_{roup}(E(G))=\{H_i\}^q_1$ under the zero $H_{k}\in F_{n}(H,f)$. The resulting encrypted network is denoted as $N$. Let $N=N_1$ and $g_{roup}=g^{(1)}_{roup}$.

We construct the desired group $F_{n}(N,g_{roup})$ by adding $k$ to the lower index $i$ of the label $H_i=g_{roup}(u)$ of each vertex $u$ of $G$, and adding $k$ to the lower index $(i,j)$ of the label $H_{(i,j)}=g_{roup}(uv)$ of each edge $uv$ of $G$ with $H_j=g_{roup}(v)$ under modular $n$. So, we get new encrypted networks $N_{k+1}=\{H_{i+k}\,(\textrm{mod}\,n): H_{i}=g_{roup}(u),u\in V(G)\}\cup \{H_{(i,j)+k}\,(\textrm{mod}\,n): H_{(i,j)}=g_{roup}(uv), uv\in E(G)\}$, and write the labelling of $N_{k+1}$ by $g^{(k+1)}_{roup}$, $k\in [1,n-1]$. Thereby, we get a set $F_{n}(N,g_{roup})=\{N_{k+1}:~k\in [0,n-1]\}$

Next, we select arbitrarily an element $N_{k}\in F_{n}(N,g_{roup})$ as \emph{zero}, and define an operation $\oplus$ for $F_{n}(N,g_{roup})$ as: $N_{i}\oplus N_{j}=N_{i+j-k\,(\textrm{mod}\,n)}$ means $g^{(i)}_{roup}(w)\oplus g^{(j)}_{roup}(w)=g^{(i+j-k)}_{roup}(w)\,(\textrm{mod}\,n)$ for $w\in V(G)\cup E(G)$. For $w \in V(G)$ as $g^{(1)}_{roup}(w)=H_{s}$, we have $g^{(i)}_{roup}(w)=H_{s+i}$, $g^{(j)}_{roup}(w)=H_{s+j}$, $g^{(k)}_{roup}(w)=H_{s+k}$, and $g^{(i+j-k)}_{roup}(w)=H_{s+i+j-k}$. Thereby, we have
$$H_{s+i}\oplus H_{s+j}=H_{s+i+j-(s+k)\,(\textrm{mod}\,n)}=H_{i+j-k\,(\textrm{mod}\,n)},$$
and have proven
$$g^{(i)}_{roup}(w)\oplus g^{(j)}_{roup}(w)=g^{(i+j-k)}_{roup}(w)\,(\textrm{mod}\,n), $$
for $w\in V(G)\cup E(G)$, and it is not hard to show the \emph{Zero}, the \emph{Inverse}, the \emph{Uniqueness and Closure}, the \emph{Commutative law} and the\emph{ Associative law} on $F_{n}(H,f)$, since $F_{n}(H,f)$ is an every-zero graphic group (\cite{Wang-Xu-Yao-Ars-2018, Yao-Mu-Sun-Zhang-Su-Ma-Wang2018, Yao-Sun-Zhao-Li-Yan-2017}).

The claim of the theorem is proven. \end{proof}

An example for understanding Theorem \ref{thm:group-labelling-by-group-labelling} is shown in Fig.\ref{fig:group-of-group}.

\begin{figure}[h]
\centering
\includegraphics[height=10cm]{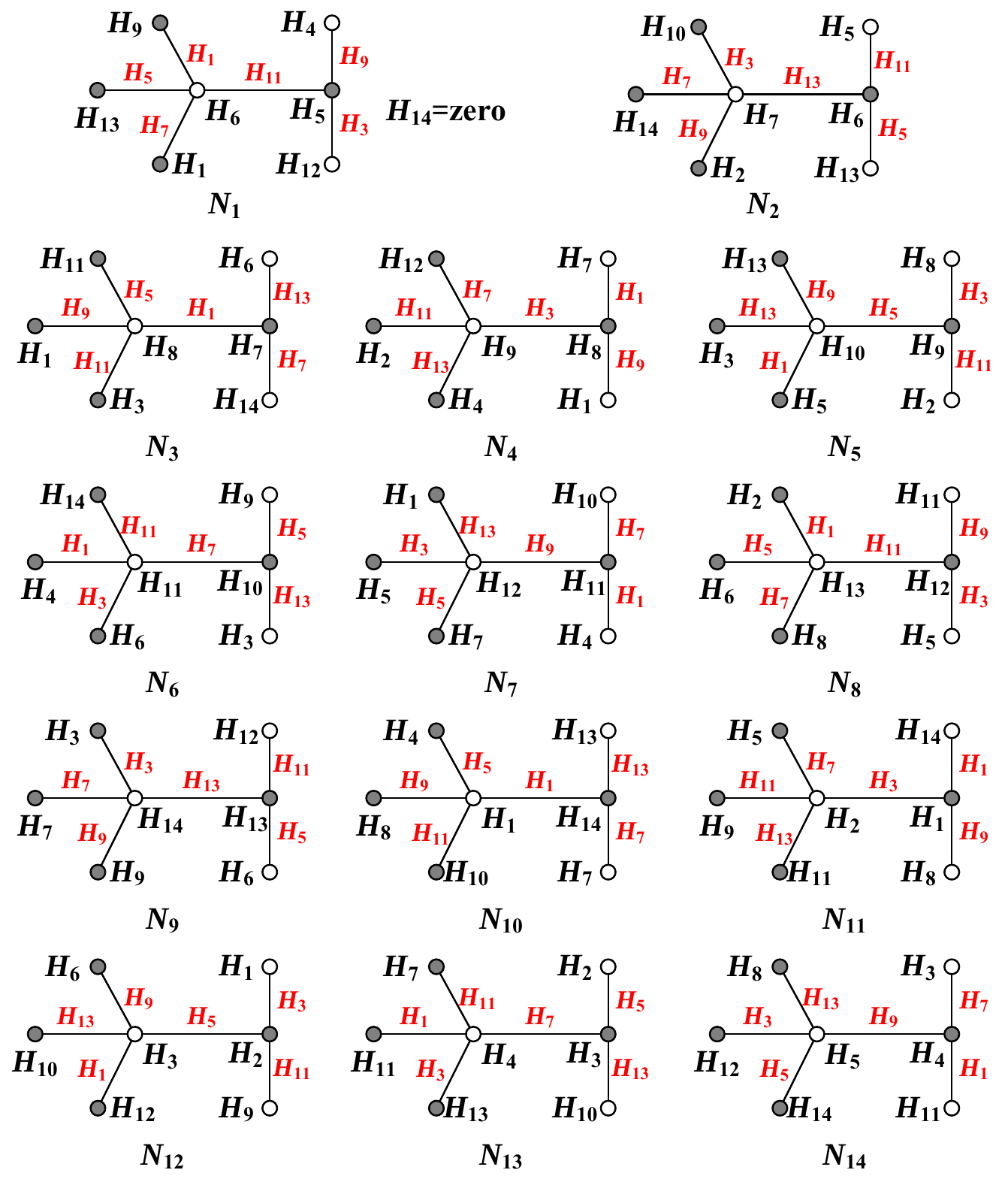}
\caption{\label{fig:group-of-group}{\small An every-zero graphic group $F_{14}(N,g_{roup})$ made by an encrypted network $N_1=N_{et}(T,F_{14}(H,f))$ pictured in Fig.\ref{fig:encrypt-network}.}}
\end{figure}

The authors in \cite{Zhang-Sun-Yao-ICMITE2017} and \cite{ZHANG-MU-SUN-YAO-IAEAC2018} propose other methods for producing every-zero graphic groups.

\section{Topsnut-matchings}

Here, for making vv-type/vev-type TB-paws, we will apply the \emph{Path-neighbor-method, Cycle-neighbor-method, Lobster-neighbor-method,Spider-neighbor-method and Euler-Hamilton-method} introduced in the previous sections.

\subsection{Examples for Topsnut-matchings}

We use a Topsnut-gpw $T$ (as a public key) presented in Fig.\ref{fig:1-example} to induce a vev-type TB-paw
$${
\begin{split}
D_{vev}(T)=&13323224125132331341451561623167\\
&223231622817918101911202220122120,
\end{split}}$$
by hands, and another Topsnut-gpw $H$ (as a private key) distributes us a vev-type TB-paw
$${
\begin{split}
D_{vev}(H)=&13141131207189161115121111721\\
&921417201223245256131910152381.
\end{split}}$$
Thus, we get a\emph{ digital authentication} $\odot \langle T,H\rangle $, which generates the authentication vev-type TB-paw $D_{vev}(\odot \langle T,H\rangle )$ as follows:
\begin{equation}\label{eqa:c3xxxxx}
D_{vev}(\odot \langle T,H\rangle )=D_{vev}(T)\uplus D_{vev}(H).
\end{equation}
or
\begin{equation}\label{eqa:c3xxxxx}
D'_{vev}(\odot \langle T,H\rangle )=D_{vev}(H)\uplus D_{vev}(T).
\end{equation}
Clearly, $D_{vev}(\odot \langle T,H\rangle )$ differs from $D'_{vev}(\odot \langle T,H\rangle )$.

Notice that a Topsnut-gpw authentication contains two parts: \emph{digital authentication}, \emph{topological structure authentication}. In this example, the topological structure authentication is the unlabelled graph $\odot \langle T,H\rangle $. We emphasize the topological structure authentication, since a Topsnut-gpw (as a public key) may match with two or more Topsnut-gpws (as private keys). See Fig.\ref{fig:2-more-matching}, a Topsnut-gpw $T$ matches with three Topsnut-gpws $G_1,G_2$ and $H$, but three matchings differ from each other in topological structures. The vev-type TB-paw $D_{vev}(T)$ matches with the following TB-paw
$${
\begin{split}
D_{vev}(G_1)=&132512422331332010211122122313\\
&31020319410415616717818946144,
\end{split}}$$
since the topological structure of $T$ is isomorphic to that of $G_1$. So, $D_{vev}(\odot \langle T,G_1\rangle )$ differs from $D_{vev}(\odot \langle T,H\rangle )$.

Unfortunately, for a given Topsnut-gpw $G$ admitting a labelling being the same as that admitted by $T$ shown in Fig.\ref{fig:2-more-matching}(a), we do ont have efficient algorithm for finding all matchings of $G$, and go on theoretical jobs on them.

\subsection{Why are Topsnut-matrices good for generating TB-paws}

\begin{defn}\label{defn:6C-labelling}
\cite{Yao-Sun-Zhang-Mu-Sun-Wang-Su-Zhang-Yang-Yang-2018arXiv} A total labelling $f:V(G)\cup E(G)\rightarrow [1,p+q]$ for a bipartite $(p,q)$-graph $G$ is a bijection holding:

(i) (e-magic) $f(uv)+|f(u)-f(v)|=k$;

(ii) (ee-difference) each edge $uv$ matches with another edge $xy$ holding $f(uv)=|f(x)-f(y)|$ (or $f(uv)=(p+q+1)-|f(x)-f(y)|$);

(iii) (ee-balanced) let $s(uv)=|f(u)-f(v)|-f(uv)$ for $uv\in E(G)$, then there exists a constant $k'$ such that each edge $uv$ matches with another edge $u'v'$ holding $s(uv)+s(u'v')=k'$ (or $(p+q+1)+s(uv)+s(u'v')=k'$) true;

(iv) (EV-ordered) $f_{\min}(V(G))>f_{\max}(E(G))$ (or $f_{\max}(V(G))<f_{\min}(E(G))$, or $f(V(G))\subseteq f(E(G))$, or $f(E(G))\subseteq f(V(G))$, or $f(V(G))$ is an odd-set and $f(E(G))$ is an even-set;

(v) (ve-matching) there exists a constant $k''$ such that each edge $uv$ matches with one vertex $w$ such that $f(uv)+f(w)=k''$, and each vertex $z$ matches with one edge $xy$ such that $f(z)+f(xy)=k''$, except the \emph{singularity} $f(x_0)=\lfloor \frac{p+q+1}{2}\rfloor $;

(vi) (set-ordered) $f_{\max}(X)<f_{\min}(Y)$ (or $f_{\min}(X)>f_{\max}(Y)$) for the bipartition $(X,Y)$ of $V(G)$.

We call $f$ a \emph{6C-labelling} of $G$.\qqed
\end{defn}

In Definition \ref{defn:6C-labelling}, it is natural, each edge $uv$ corresponds another edge $xy$ such that
\begin{equation}\label{eqa:c3xxxxx}
f(uv)+f(xy)=\min f(E(G))+\max f(E(G))
\end{equation} and each vertex $w$ corresponds another vertex $z$ such that
\begin{equation}\label{eqa:c3xxxxx}
f(w)+f(z)=\min f(V(G))+\max f(V(G)).
\end{equation}

In Fig.\ref{fig:6C-matching-matrix-1}, the Topsnut-matrix $A(T)$ matches with its dual Topsnut-matrix $A^{-1}(T)$, since the sum of each element of $A(T)$ and its corresponding element of $A^{-1}(T)$ is just 26. According to Definition \ref{defn:6C-labelling}, the Topsnut-matrix $A(T)=(X~W~Y)^{-1}$ holds the \emph{6C-restriction}: (i) $e_i+|x_i-y_i|=13$; (ii) $e_i=|x_j-y_j|$; (iii) $(|x_i-y_i|-e_i)+(|x_j-y_j|-e_j)=0$; (iv) $\min (X\cup Y)>\max W$; (v) $e_i+x_s=26$ or $e_i+y_t=26$; (vi) $\min X>\max Y$. However, the dual Topsnut-matrix $A^{-1}(T)=(X'~W'~Y')^{-1}$ holds $e'_i-|x'_i-y'_i|=13$ and $\min (X'\cup Y')< \max W'$ only.

\begin{figure}[h]
\centering
\includegraphics[height=3cm]{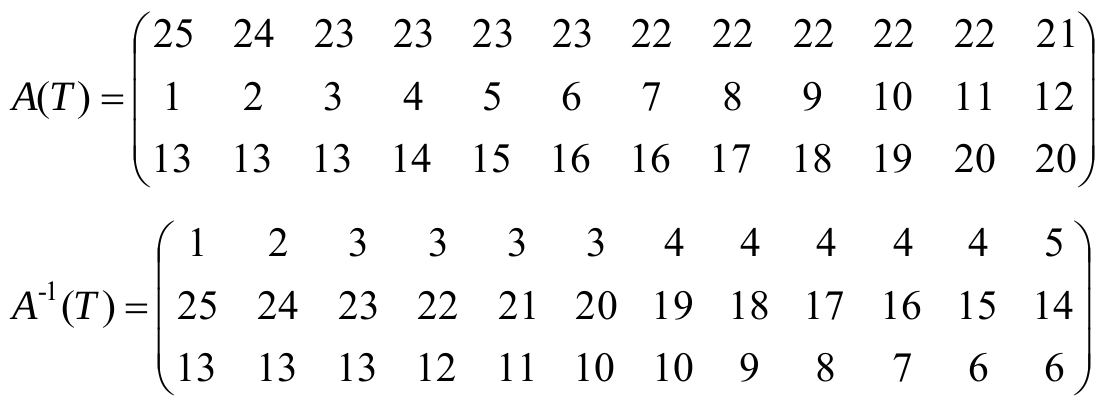}
\caption{\label{fig:6C-matching-matrix-1}{\small A Topsnut-matrix $A(T)$ of a Topsnut-gpw $T$ shown in Fig.\ref{fig:2-more-matching}(a), and the dual Topsnut-matrix $A^{-1}(T)$ of $A(T)$.}}
\end{figure}

\begin{figure}[h]
\centering
\includegraphics[height=1.5cm]{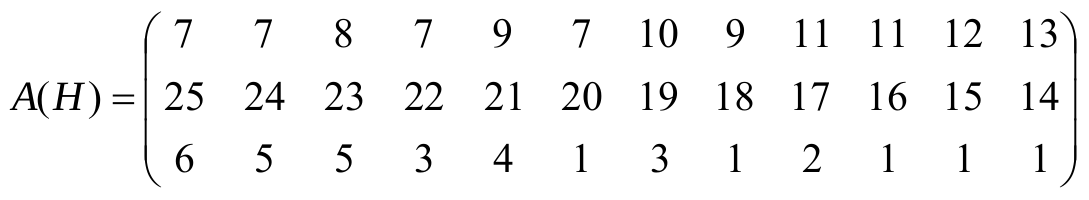}
\caption{\label{fig:6C-matching-matrix-2}{\small A Topsnut-matrix $A(H)$ of a Topsnut-gpw $H$ shown in Fig.\ref{fig:2-more-matching}(d), which matches with the Topsnut-matrix $A(T)$ of a Topsnut-gpw $T$ shown in Fig.\ref{fig:2-more-matching}(a).}}
\end{figure}

Notice that $T\not \cong H$, see Fig.\ref{fig:2-more-matching}(a) and (d). And $\odot_1\langle T,H \rangle $ obtained by coinciding the vertex $x_0$ of $G$ having $f(x_0)=13$ with the vertex $w_0$ of $H$ having $g(w_0)=13$ is a 6C-complementary matching conforming to Definition \ref{defn:6C-complementary-matching}. Moreover, the Topsnut-matrix $A(H)=(X''~W''~Y'')^{-1}$ holds the \emph{6C-restriction}: (i) $e''_i+|x''_i-y''_i|=13$; (ii) $e''_i=26-|x''_j-y''_j|$; (iii) $26-[(|x''_i-y''_i|-e''_i)+(|x''_j-y''_j|-e''_j)]=0$; (iv) $\min (X\cup Y)<\max W$; (v) $e''_i+x''_s=26$ or $e''_i+y''_t=26$; (vi) $\min X>\max Y$.

We like to use Topsnut-matrices to generate TB-paws since there are the following advantages of Topsnut-matrices:
\begin{asparaenum}[Prop-1. ]
\item Topsnut-matrices are easily saved in computer.
\item A Topsnut-matrix $A_{vev}(G)$ of a $(p,q)$-graph $G$ generates at least $\sum^{M}_{m=1}(m!)$ vv-type/vev-type TB-paws, where $M=\lfloor 3q/2\rfloor$.
\item In general, the vv-type/vev-type TB-paws generated by a Topsnut-matrix $A(\odot_1\langle T,H \rangle )$ differ from those vv-type/vev-type TB-paws of form $D(T)\uplus D(H)$ obtained from two Topsnut-matrices $A(T)$ and $A(H)$.
\item The procedure of rebuilding a Topsnut-matrix $A(T)$ by a vv-type/vev-type TB-paw $D(T)$, verifying $A(T)$ holding the \emph{6C-restriction}, and then redrawing the Topsnut-gpw $T$ by $A(T)$, is not easy to be realized, even impossible if a Topsnut-gpw possesses thousands of vertices and edges. Thereby, it is hard to reproduce a 6C-complementary matching $\odot_1\langle T,H \rangle $ when $T$ is as a \emph{public key}, $H$ is a \emph{private key} and $\odot_1\langle T,H \rangle $ is an \emph{authentication}.
\end{asparaenum}

\subsection{Looking for matchings}

We use an example to illustrate a procedure of transforming a graceful labelling to an odd-graceful labelling. The tree $T$ of $17$ vertices depicted in Fig.\ref{fig:gracefui-to-odd-tu} admits a \emph{set-ordered graceful labelling} $f$ shown in Fig.\ref{fig:gracefui-to-odd-tu}(a): $\max f(X)<\min f(Y)$ with $X=\{\textrm{black vertices}\}$ and $Y=\{\textrm{white vertices}\}$, each edge $uv$ of $T$ is balled by $f(uv)=|f(u)-f(v)|$, such that $f(E(T))=[1,16]$. First, we define a labelling $f_1$ of $T$ by setting $f_1(x)=2f(x)$ with $x\in X$, $f_1(y)=f(y)$ with $y\in Y$, and $f_1(uv)=f(uv)$ for $uv\in E(T)$ (see Fig.\ref{fig:gracefui-to-odd-tu}(b)). Second, we define another labelling $f_2$ of $T$ by setting $f_2(x)=f_1(x)$ with $x\in X$, $f_2(y)=2f_1(y)-1$ with $y\in Y$, and $f_2(uv)=f_1(uv)$ for $uv\in E(T)$ (see Fig.\ref{fig:gracefui-to-odd-tu}(c)). Finally, we have the desired odd-graceful labelling $f_3$ obtained by setting $f_3(w)=f_2(w)$ with $w\in X\cup Y$, and $f_3(uv)=|f_2(u)-f_2(v)|$ for $uv\in E(T)$ (see Fig.\ref{fig:gracefui-to-odd-tu}(d)).

\begin{figure}[h]
\centering
\includegraphics[height=6cm]{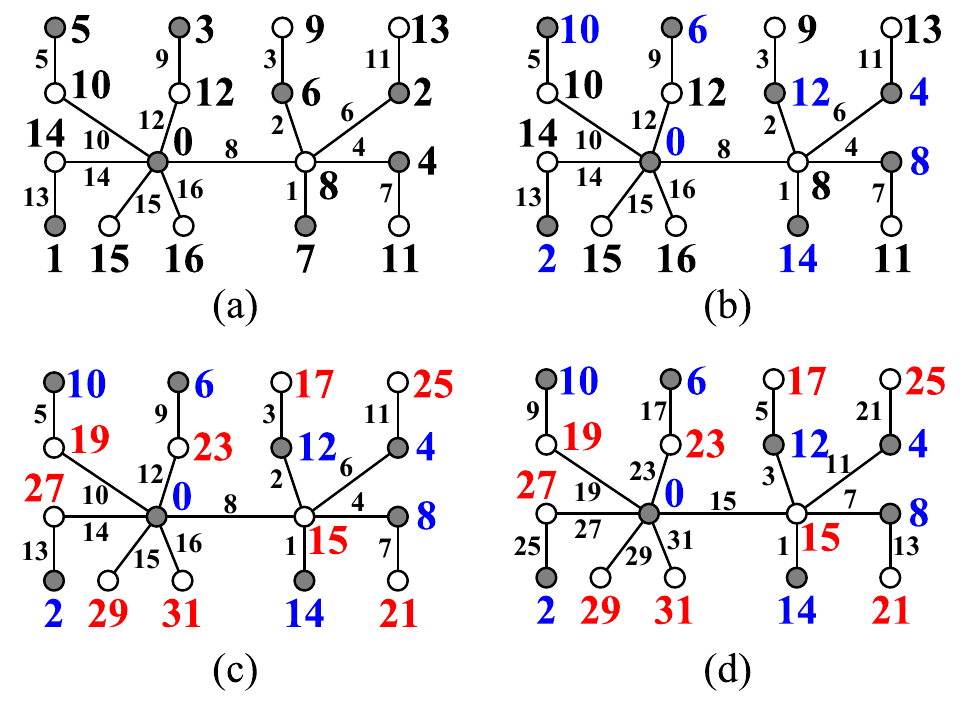}
\caption{\label{fig:gracefui-to-odd-tu}{\small A procedure of transforming a graceful labelling to an odd-graceful labelling.}}
\end{figure}

The set-ordered graceful labelling $f$ of $T$ presented in Fig.\ref{fig:gracefui-to-odd-tu}(a) induces a Topsnut-matrix $A_a$ depicted in Fig.\ref{fig:gracefui-to-odd}. And other three labellings $f_1,f_2$ and $f_3$ give us three Topsnut-matrices $A_b$, $A_c$ and $A_d$ shown in Fig.\ref{fig:gracefui-to-odd}, respectively. Thereby, we have a \emph{Topsnut-matrix chain} $A_a \rightarrow A_b\rightarrow A_c\rightarrow A_d$ and a \emph{TB-paw chain} $D_a\rightarrow D_b\rightarrow D_c\rightarrow D_d$ obtained from $A_a, A_b,A_c,A_d$, respectively. In general, Topsnut-gpws contain three basic characters:

(1) Topsnut-gpws = Topological structures (configuration, graph) plus labelling/colorings;

(2) Topsnut-matrices join Topsnut-gpws by TB-paws;

(3) TB-paws are easy for encryption.

\begin{figure}[h]
\centering
\includegraphics[height=5cm]{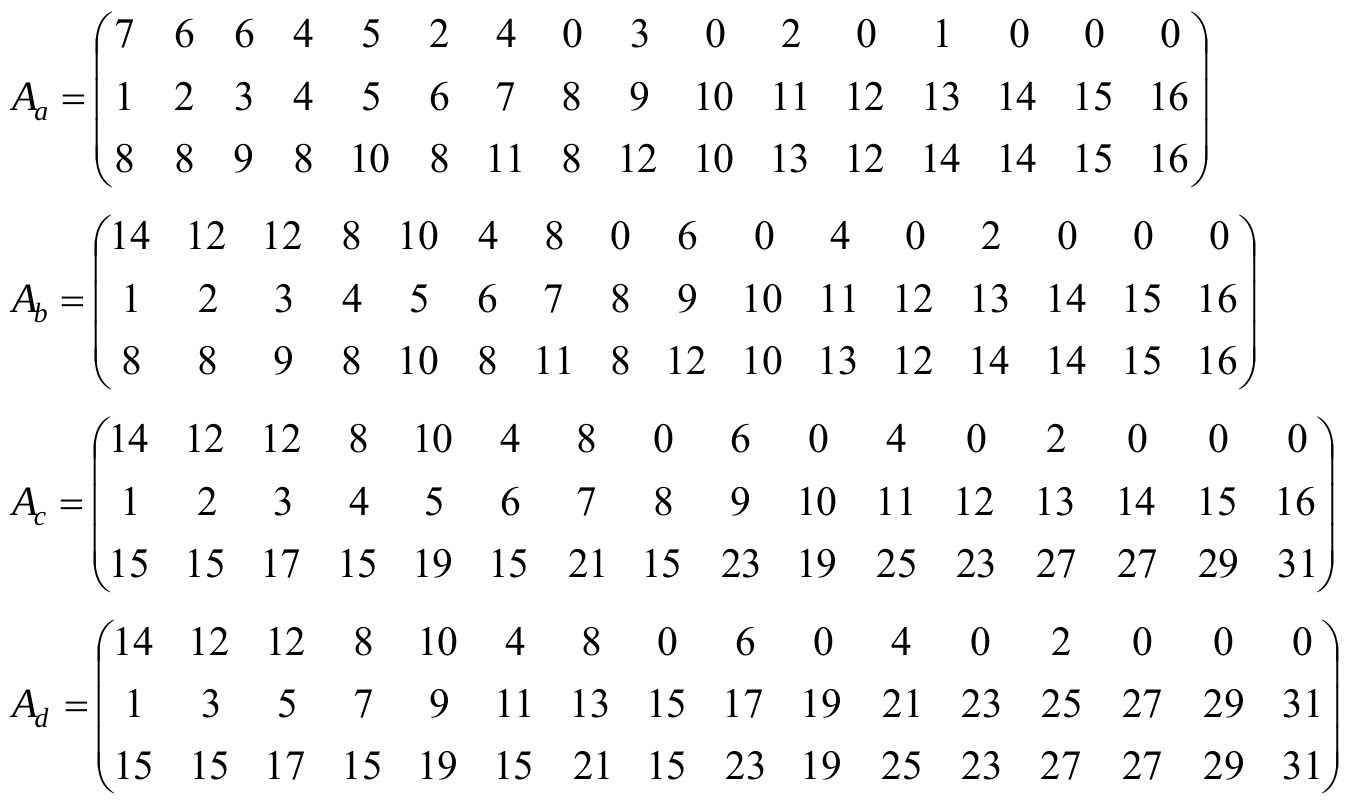}
\caption{\label{fig:gracefui-to-odd}{\small Four Topsnut-matrices corresponding four Topsnut-gpws.}}
\end{figure}

\begin{lem}\label{thm:lamma-set-ordered-vs-6C-labelling}
\cite{Yao-Sun-Zhang-Mu-Sun-Wang-Su-Zhang-Yang-Yang-2018arXiv} If a tree admits a set-ordered graceful labelling if and only if it admits a 6C-labelling.
\end{lem}

\begin{thm}\label{thm:set-ordered-vs-6C-matching}
If two trees of $p$ vertices admit set-ordered graceful labellings, then they are a 6C-complementary matching.
\end{thm}
\begin{proof} Assume that each tree $T_i$ of $p$ vertices admits a set-ordered graceful labelling $f_i$ and let $(X_i,Y_i)$ be the bipartition of $T_i$ with $i=1,2$. So, by the definition of a set-ordered graceful labelling, we have $\max f_i(X_i)<\min f_i(Y_i)$ where $X_i=\{x_{i,j}:j\in [1,s_i]\}$ and $Y_i=\{y_{i,j}:j\in [1,t_i]\}$ holding $s_i+t_i=p$ with $i=1,2$. Without loss of generality, we can set $f_i(x_{i,j})=j-1$ for $j\in [1,s_i]$, $f_i(y_{i,j})=s_i+j-1$ for $j\in [1,t_i]$ and $f_i(x_{i,s}y_{i,t})=f_i(y_{i,t})-f_i(x_{i,s})=s_i+t-s$ for each edge $x_{i,s}y_{i,t}\in E(T_i)$, and $f_i(E(T_i))=[1,p]$ for $i=1,2$.

We define another labelling $f^*_1$ of $T_1$ as: $f^*_1(w)=p+f_1(w)$ for $w\in V(T_1)$ and $f^*_1(x_{1,s}y_{1,t})=p+1-f(x_{1,s}y_{1,t})$ for each edge $x_{1,s}y_{1,t}\in E(T_1)$. So, we can compute $f^*_1(V(T_1))=[p,2p-1]$, $f^*_1(E(T_1))=[1,p-1]$.

Next, we define another labelling $f^*_2$ of $T_2$ as: $f^*_2(w)=f_2(w)+1$ for $w\in V(T_2)$ and $f^*_2(x_{2,i}y_{2,j})=p+f_2(x_{2,i}y_{2,j})$ for each edge $x_{2,i}y_{2,j}\in E(T_2)$. Thereby, we get $f^*_2(V(T_2))=[1,p]$, $f^*_2(E(T_2))=[p+1, 2p-1]$.

Notice that $f^*_1(V(T_1))\setminus \{p\}=f^*_2(E(T_2))$, $f^*_1(E(T_1))=f^*_2(V(T_2))\setminus \{p\}$, and by Lemma \ref{thm:lamma-set-ordered-vs-6C-labelling}, we have proven the theorem.
\end{proof}

\begin{defn}\label{defn:image-labelling}
$^*$ Let $f_i:V(G)\rightarrow [a,b]$ be a labelling of a $(p,q)$-graph $G$ and define each edge $uv\in E(G)$ has its own label as $f_i(uv)=|f_i(u)-f_i(v)|$ with $i=1,2$. If each edge $uv\in E(G)$ holds $f_1(uv)+f_2(uv)=k$ true, where $k$ is a positive constant, we call $f_1$ and $f_2$ are a pair of \emph{image-labellings}, and $f_i$ a \emph{mirror-image} of $f_{3-i}$ with $i=1,2$.\qqed
\end{defn}

A tree $T$ appeared in Fig.\ref{fig:graceful-mirror-labelling} admits a pair of \emph{set-ordered graceful image-labellings} (a) and (b). We can consider a pair of image-labellings as a \emph{matching labelling} too.

\begin{figure}[h]
\centering
\includegraphics[height=6cm]{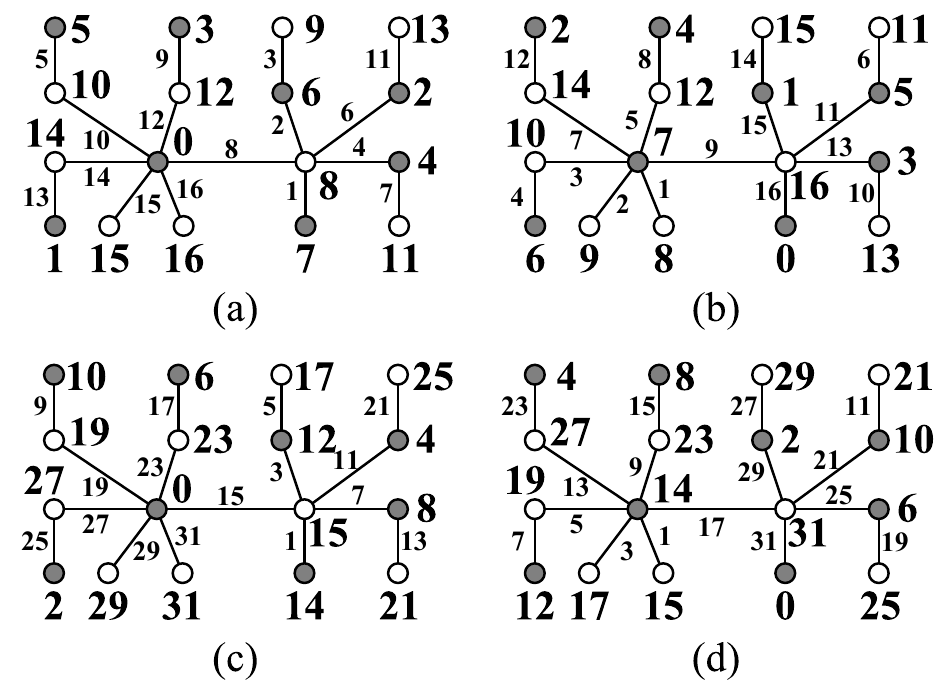}
\caption{\label{fig:graceful-mirror-labelling}{\small (a) and (b) are a pair of set-ordered graceful image-labellings with $f_a(uv)+h_b(uv)=17$; (c) and (d) are a pair of set-ordered odd-graceful image-labellings with $f_c(uv)+h_d(uv)=32$.}}
\end{figure}

\begin{lem}\label{thm:graceful-image-labelling}
If a tree $T$ admits a set-ordered graceful labelling $f$, then $T$ admits another set-ordered graceful labelling $g$ such that $f$ and $g$ are a pair of image-labellings.
\end{lem}
\begin{proof} Suppose that $(X,Y)$ is the bipartition of a tree $T$ with $p$ vertices, where $X=\{x_{i}:i\in [1,s]\}$ and $Y=\{y_{j}:j\in [1,t]\}$ holding $s+t=|V(T)|=p$. By the hypothesis of the theorem, $T$ admits a set-ordered graceful labelling $f$ such that $f(x_i)=i-1$ for $i\in [1,s]$, $f(y_j)=s+j-1$ for $j\in [1,t]$ and $f(x_iy_j)=f(y_j)-f(x_i)=s+j-i$ for each edge $x_iy_j\in E(T)$. We define another labelling $g$ of $T$ as: $g(x_i)=s-1-f(x_i)$ for $i\in [1,s]$, $g(y_j)=t+2s-1-f(y_{j})$ for $j\in [1,t]$, then
\begin{equation}\label{eqa:c3xxxxx}
{
\begin{split}
g(x_iy_j)&=g(y_j)-g(x_i)\\
&=t+2s-1-f(y_{j})-[s-1-f(x_i)]\\
&=t+s-[f(y_{j})-f(x_i)]\\
&=t+s-f(x_iy_j)
\end{split}}
\end{equation}
for each edge $x_iy_j\in E(T)$. So, $$f(x_iy_j)+g(x_iy_j)=t+s=|E(T)|+1,$$ a constant, as desired.
\end{proof}

In \cite{Yao-Liu-Yao-2017}, the authors have proven the following mutually equivalent labellings:

\begin{thm} \label{thm:connections-several-labellings}
\cite{Yao-Liu-Yao-2017} Let $T$ be a tree on $p$ vertices, and let $(X,Y)$ be its
bipartition. For all values of integers $k\geq 1$ and $d\geq 1$, the following assertions are mutually equivalent:

$(1)$ $T$ admits a set-ordered graceful labelling $f$ with $f(X)<f(Y)$.

$(2)$ $T$ admits a super felicitous labelling $\alpha$ with
$\alpha(X)<\alpha(Y)$.

$(3)$ $T$ admits a $(k,d)$-graceful labelling $\beta$ with
$\beta(x)<\beta(y)-k+d$ for all $x\in X$ and $y\in Y$.

$(4)$ $T$ admits a super edge-magic total labelling $\gamma$ with
$\gamma(X)<\gamma(Y)$ and a magic constant $|X|+2p+1$.

$(5)$ $T$ admits a super $(|X|+p+3,2)$-edge antimagic total
labelling $\theta$ with $\theta(X)<\theta(Y)$.

$(6)$ $T$ has an odd-elegant labelling $\eta$ with
$\eta(x)+\eta(y)\leq 2p-3$ for every edge $xy\in E(T)$.

$(7)$ $T$ has a $(k,d)$-arithmetic labelling $\psi$ with
$\psi(x)<\psi(y)-k+d\cdot |X|$ for all $x\in X$ and $y\in Y$.

$(8)$ $T$ has a harmonious labelling $\varphi$ with
$\varphi(X)<\varphi(Y\setminus \{y_0\})$ and $\varphi(y_0)=0$.
\end{thm}

By Lemma \ref{thm:graceful-image-labelling} and Theorem \ref{thm:connections-several-labellings}, if a tree $T$ admits set-ordered graceful labelling, then we have the following results and present Fig.\ref{fig:image-labelling-1} and Fig.\ref{fig:image-labelling-2} for illustrating these results:

\begin{thm}\label{thm:10-image-labellings}
If a tree $T$ admits set-ordered graceful labelling, then $T$ admits a pair of $SKD$  image-labellings, where $SKD\in \{$ set-ordered graceful, set-ordered odd-graceful, edge-magic graceful, set-ordered felicitous, set-ordered odd-elegant, super set-ordered edge-magic total, super set-ordered edge-antimagic total, set-ordered $(k,d)$-graceful, $(k,d)$-edge antimagic total, $(k,d)$-arithmetic total, harmonious,  $(k,d)$-harmonious $\}$.
\end{thm}

The above results on image-labellings are illustrated in Fig.\ref{fig:image-labelling-1}, Fig.\ref{fig:image-labelling-2} and Fig.\ref{fig:group-group-matching}, however, we omit the proofs of them here.

\begin{figure}[h]
\centering
\includegraphics[height=11cm]{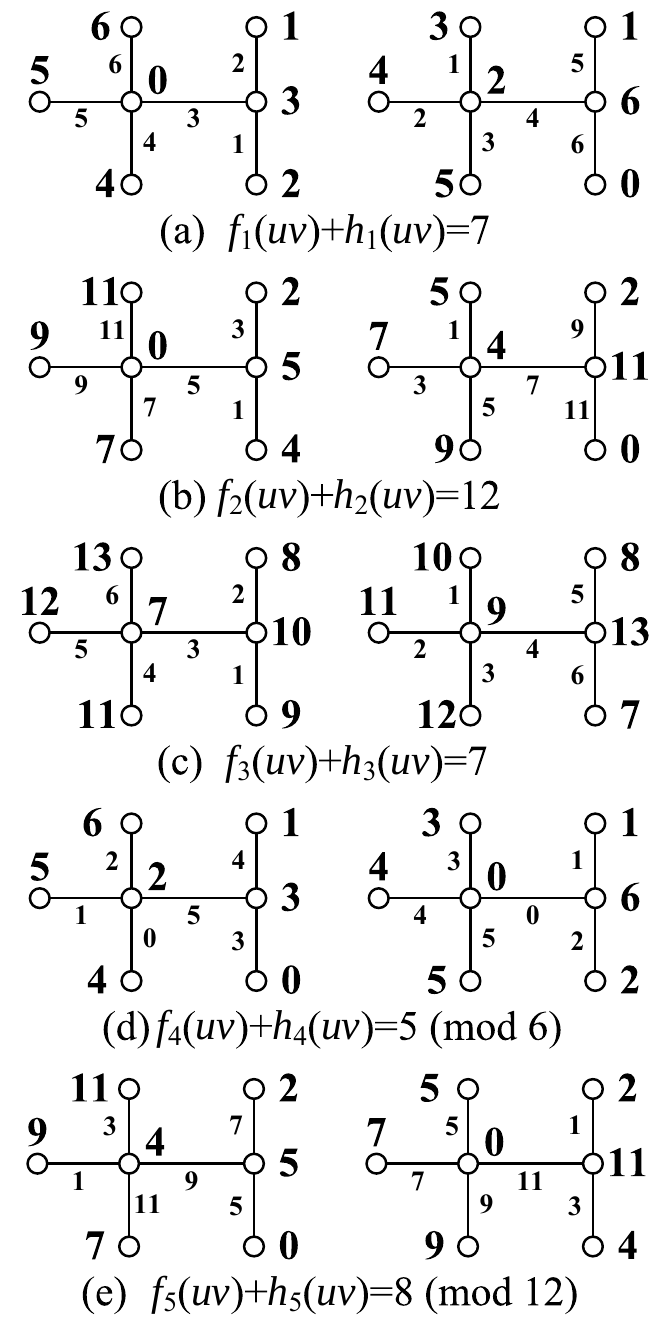}
\caption{\label{fig:image-labelling-1}{\small A tree $T$ admits (\cite{Gallian2016,Zhou-Yao-Chen-Tao2012,Zhou-Yao-Chen-2013}): (a) a pair of set-ordered graceful image-labellings $f_1$ and $h_1$; (b) a pair of set-ordered odd-graceful image-labellings $f_2$ and $h_2$; (c) a pair of edge-magic graceful image-labellings $f_3$ and $h_3$; (d) a pair of set-ordered felicitous image-labellings $f_4$ and $h_4$; (e) a pair of set-ordered odd-elegant image-labellings $f_5$ and $h_5$.}}
\end{figure}

\begin{figure}[h]
\centering
\includegraphics[height=12cm]{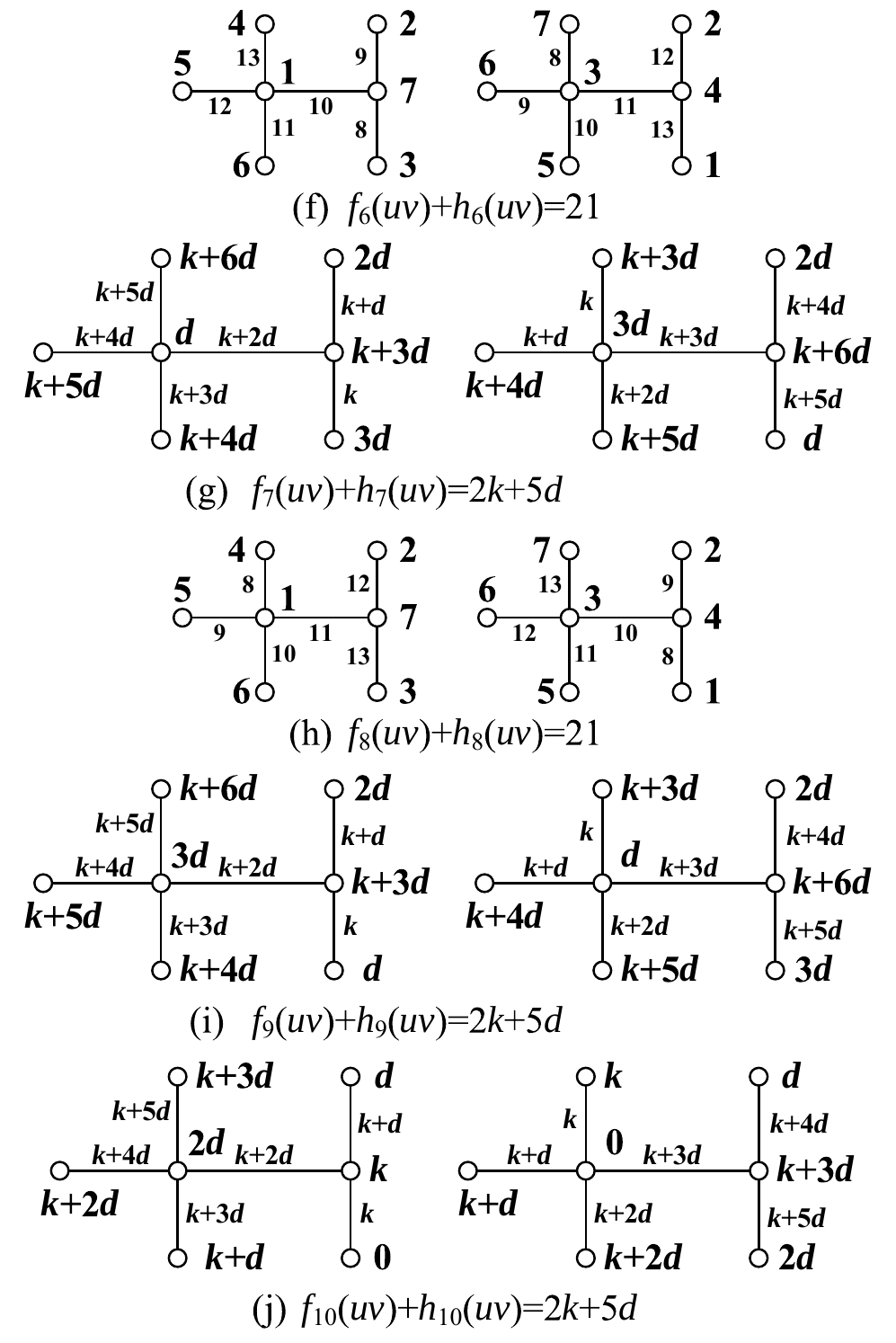}
\caption{\label{fig:image-labelling-2}{\small A tree $T$ admits (\cite{Gallian2016,Zhou-Yao-Chen-Tao2012,Zhou-Yao-Chen-2013}): (f) a pair of super set-ordered edge-magic total image-labellings $f_6$ and $h_6$; (g) a pair of set-ordered $(k,d)$-graceful image-labellings $f_7$ and $h_7$; (h) a pair of super set-ordered edge-antimagic total image-labellings $f_8$ and $h_8$; (i) a pair of $(k,d)$-edge antimagic total image-labellings $f_9$ and $h_9$; (j) a pair of $(k,d)$-arithmetic image-labellings $f_{10}$ and $h_{10}$.}}
\end{figure}

Motivated from the definitions of harmonious labelling and $(k,d)$-harmonious labelling in \cite{Yao-Zhang-Sun-Mu-Wang-Zhang2018}, we present two new labellings as follows:

\begin{defn} \label{defn:twin-k-d-harmonious-labellings}
$^*$ A $(p,q)$-graph $G$ admits two $(k,d)$-harmonious labellings $f_i:V(G)\rightarrow X_0\cup X_{k,d}$ with $i=1,2$, where $X_0=\{0,d,2d, \dots ,(q-1)d\}$ and $X_{k,d}=\{k,k+d,k+2d, \dots ,k+(q-1)d\}$, such that each edge $uv\in E(G)$ is labelled as $f_i(uv)-k=[f_i(uv)+f_i(uv)-k~(\textrm{mod}~qd)]$ with $i=1,2$. If $f_1(uv)+f_2(uv)=2k+(q-1)d$, we call $f_1$ and $f_2$ a pair of \emph{$(k,d)$-harmonious image-labellings} of $G$ (see Fig.\ref{fig:image-labelling-2}).\qqed
\end{defn}

\begin{defn} \label{defn:twin-k-d-harmonious-labellings}
$^*$ A $(p,q)$-graph $G$ admits a $(k,d)$-labellings $f$, and another $(p',q')$-graph $H$ admits another $(k,d)$-labellings $g$. If $(X_0\cup X_{k,d})\setminus f(V(G)\cup E(G))=g(V(H)\cup E(H))$, then $g$ is called a complementary $(k,d)$-labelling of $f$, and both $f$ and $g$ are a \emph{twin $(k,d)$-labellings} of $G$ (see Fig.\ref{fig:image-labelling-2}).\qqed
\end{defn}

\begin{figure}[h]
\centering
\includegraphics[height=7.2cm]{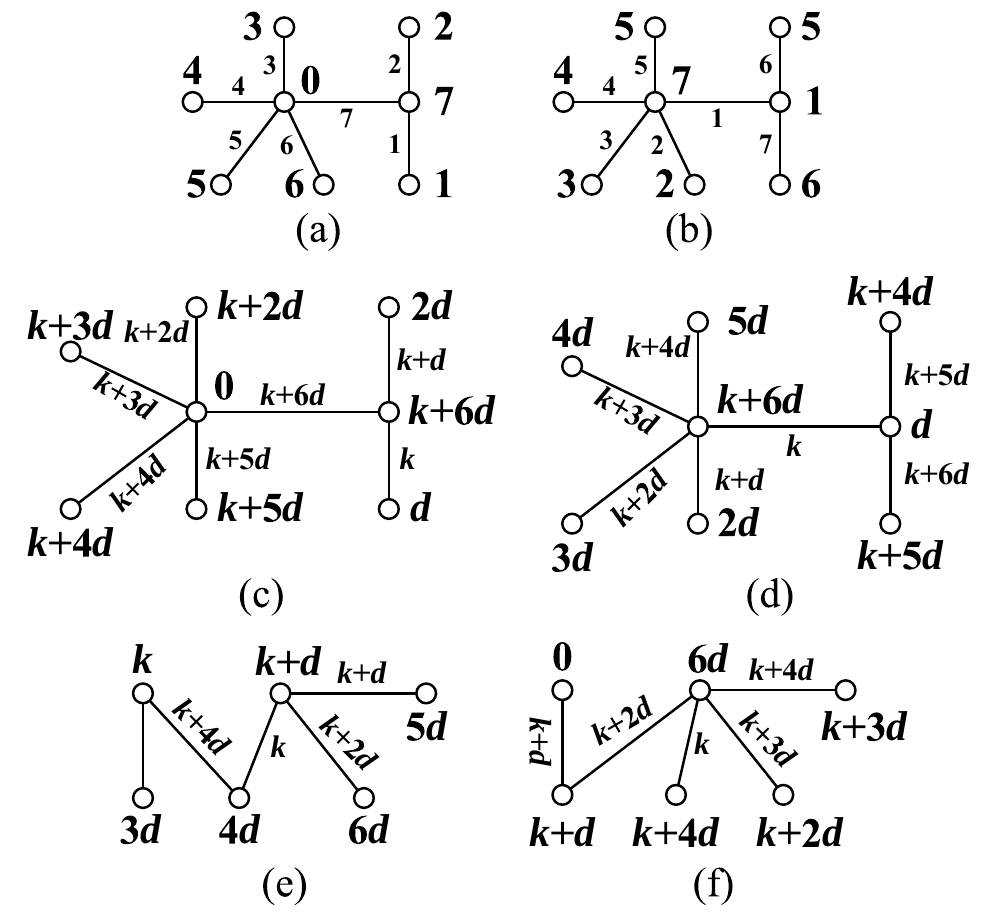}
\caption{\label{fig:group-group-matching}{\small A tree $T$ holds: (a) and (b) are a pair of harmonious image-labellings; (c) and (d) are a pair of $(k,d)$-harmonious image-labellings; (c) and (e) are a twin $(k,d)$-harmonious labellings; (d) and (f) are a twin $(k,d)$-harmonious labellings.}}
\end{figure}

\section{Other techniques for producing TB-paws}

\subsection{A new 6C-labelling, reciprocal-inverse labellings}

We introduce a new 6C-labelling, called \emph{odd-6C-labelling}, and two examples exhibited in Fig.\ref{fig:2-6C} are for understanding odd-6C-labellings.

\begin{defn}\label{defn:odd-6C-labelling}
$^*$ A $(p,q)$-graph $G$ admits a total labelling $f:V(G)\cup E(G)\rightarrow [1,4q-1]$. If this labelling $f$ holds:

(i) (e-magic) $f(uv)+|f(u)-f(v)|=k$, and $f(uv)$ is odd;

(ii) (ee-difference) each edge $uv$ matches with another edge $xy$ holding $f(uv)=2q+|f(x)-f(y)|$;

(iii) (ee-balanced) let $s(uv)=|f(u)-f(v)|-f(uv)$ for $uv\in E(G)$, then there exists a constant $k'$ such that each edge $uv$ matches with another edge $u'v'$ holding $s(uv)+s(u'v')=k'$ (or $(p+q+1)+s(uv)+s(u'v')=k'$) true;

(iv) (EV-ordered) $f_{\max}(V(G))<f_{\min}(E(G))$, and $\{|a-b|:a,b\in f(V(G))\}=[1,2q-1]$;

(v) (ve-matching) there exists two constant $k_1,k_2$ such that each edge $uv$ matches with one vertex $w$ such that $f(uv)+f(w)=k_1~(\textrm{or }k_2)$;

(vi) (set-ordered) $f_{\max}(X)<f_{\min}(Y)$ for the bipartition $(X,Y)$ of $V(G)$.

We call $f$ an \emph{odd-6C-labelling} of $G$.\qqed
\end{defn}

\begin{figure}[h]
\centering
\includegraphics[height=8.6cm]{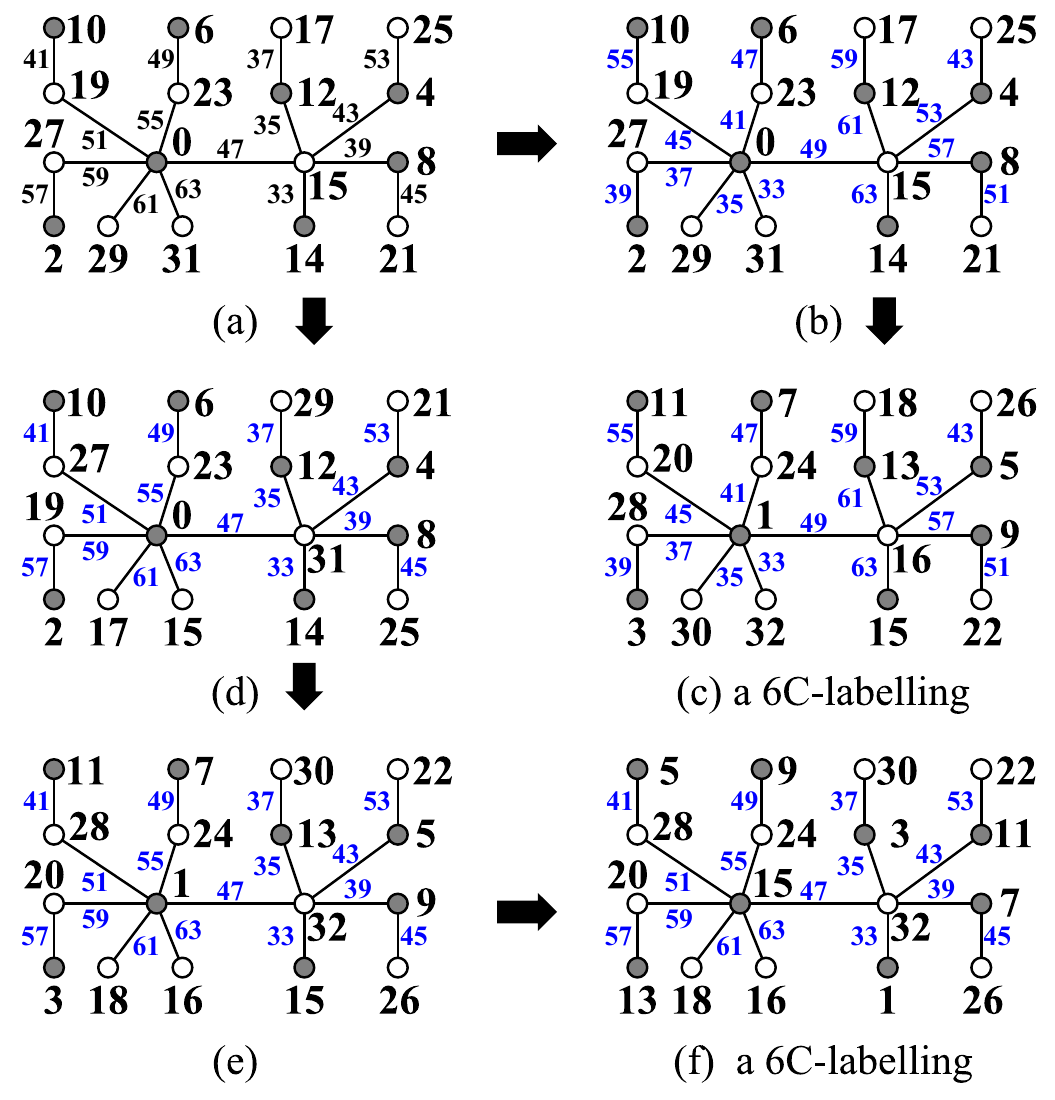}
\caption{\label{fig:2-6C}{\small The Topsnut-gpw (a) made by adding $32$ to each edge label of a Topsnut-gpw (c) shown in Fig.\ref{fig:graceful-mirror-labelling}. (a)$\rightarrow$(b)$\rightarrow$(c) is a procedure of obtaining a 6C-labellings from a pan-odd-graceful total labelling; (a)$\rightarrow$(d)$\rightarrow$(e)$\rightarrow$(f) is a procedure of obtaining another 6C-labellings from a pan-odd-graceful total labelling.}}
\end{figure}

By the way, we have discover Fig.\ref{fig:2-6C}(c) and Fig.\ref{fig:2-6C}(f) have their reciprocal-inverse matchings depicted in Fig.\ref{fig:odd-inverse-matching}.

\begin{figure}[h]
\centering
\includegraphics[height=6cm]{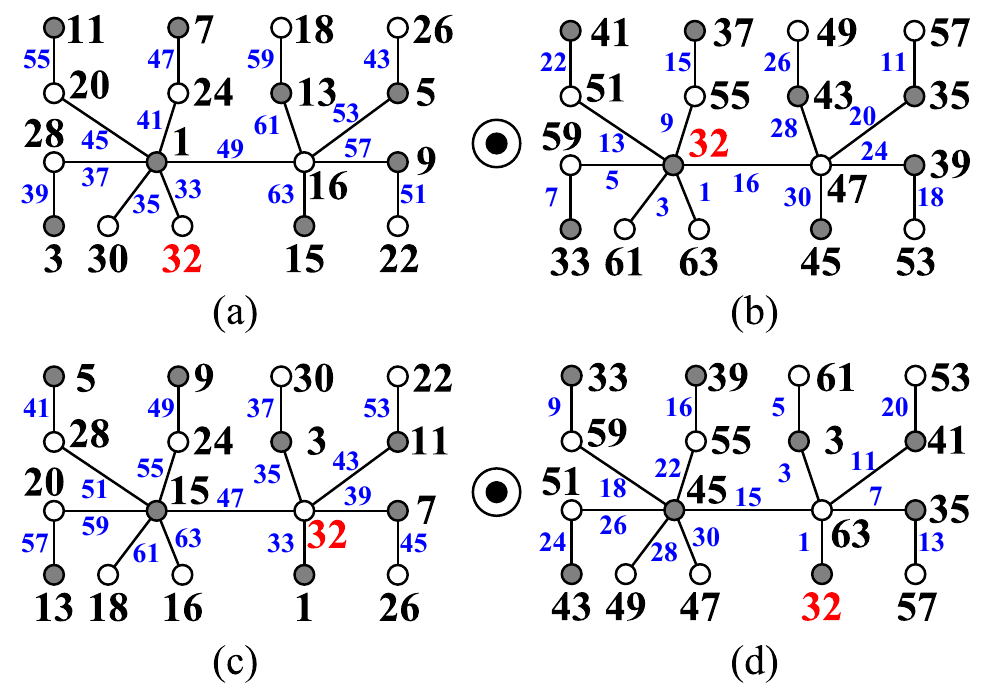}
\caption{\label{fig:odd-inverse-matching}{\small (a) and (b) are a reciprocal-inverse matching, where (a) is Fig.\ref{fig:2-6C}(c); (c) and (d) form a reciprocal-inverse matching, where (c) is Fig.\ref{fig:2-6C}(f).}}
\end{figure}

\begin{defn}\label{defn:6C-complementary-matching}
\cite{Yao-Sun-Zhang-Mu-Sun-Wang-Su-Zhang-Yang-Yang-2018arXiv}  For a given $(p,q)$-tree $G$ admitting a 6C-labelling $f$, and another $(p,q)$-tree $H$ admits a 6C-labelling $g$, if they hold $f(V(G))\setminus X^*=g(E(H))$, $f(E(G))=g(V(H))\setminus X^*$ and $f(V(G))\cap g(V(H))=X^*=\{z_0\}$ with $z_0=\lfloor \frac{p+q+1}{2}\rfloor $, then $f$ and $g$ are pairwise \emph{reciprocal-inverse}. The graph $\odot_1\langle G,H \rangle $ obtained by coinciding the vertex $x_0$ of $G$ having $f(x_0)=z_0$ with the vertex $w_0$ of $H$ having $g(w_0)=z_0 $ is called a \emph{6C-complementary matching}.\qqed
\end{defn}

\begin{defn}\label{defn:reciprocal-inverse}
\cite{Yao-Sun-Zhang-Mu-Sun-Wang-Su-Zhang-Yang-Yang-2018arXiv}  Suppose that a $(p,q)$-graph $G$ admits a total labelling $f:V(G)\cup E(G)\rightarrow [1,p+q]$, and a $(q,p)$-graph $H$ admits another total labelling $g:V(H)\cup E(H)\rightarrow [1,p+q]$. If $f(E(G))=g(V(H))\setminus X^*$ and $f(V(G))\setminus X^*=g(E(H))$ for $X^*=f(V(G))\cap g(V(H))$, then $f$ and $g$ are \emph{reciprocal-inverse} (or \emph{reciprocal complementary}) to each other, and $H$ (or $G$) is an \emph{inverse matching} of $G$ (or $H$).\qqed
\end{defn}

\begin{figure}[h]
\centering
\includegraphics[height=5.4cm]{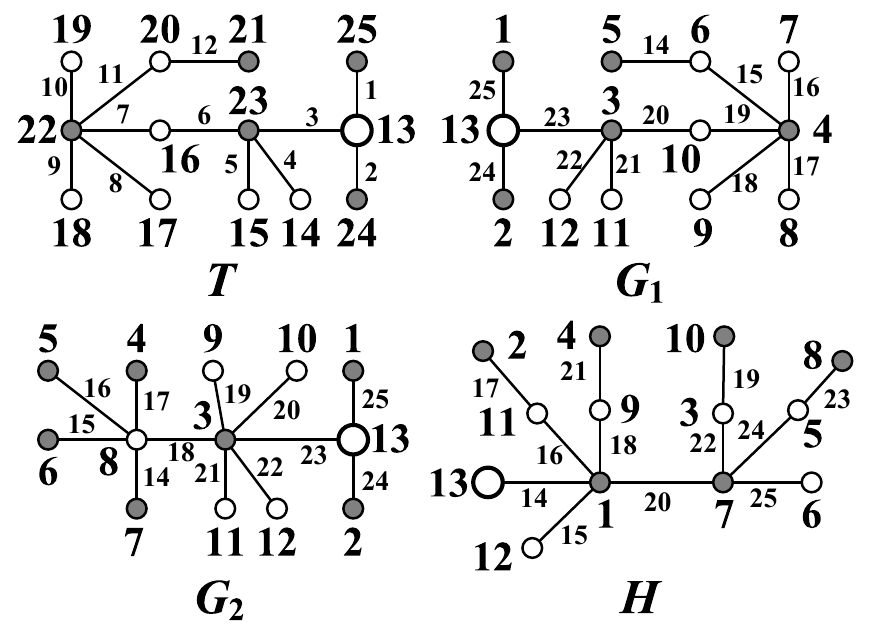}
\caption{\label{fig:2-more-matching}{\small A Topsnut-gpw $T$ depicted in Fig.\ref{fig:1-example}(a) has three inverse matchings $G_1,G_2$ and $H$, and there are three 6C-complementary matchings $\odot_1\langle G,H \rangle $, $\odot_1\langle G,G_1\rangle $ and $\odot_1\langle G,G_2\rangle $.}}
\end{figure}

\begin{thm}\label{thm:self-matching-6C-labelling}
For two reciprocal-inverse labellings $f$ and $g$ defined in Definition \ref{defn:6C-complementary-matching}, if $\{|a-b|:a,b\in f(V(G))\}=[1,q]$, then $\{|c-d|:c,d\in g(V(H))\}=[1,q]$.
\end{thm}

\begin{thm}\label{thm:inverse-reciprocal}
If two trees of $p$ vertices admit set-ordered graceful labellings, then they are inverse matching to each other under the edge-magic graceful labellings.
\end{thm}
\begin{proof} By the hypothesis of the theorem, we have known that each tree $T_i$ of $p$ vertices admits a set-ordered graceful labelling $f_i$ and let $(X_i,Y_i)$ be the bipartition of $T_i$ with $i=1,2$. The definition of a set-ordered graceful labelling means $\max f_i(X_i)<\min f_i(Y_i)$ where $X_i=\{x_{i,j}:j\in [1,s_i]\}$ and $Y_i=\{y_{i,j}:j\in [1,t_i]\}$ as well as $s_i+t_i=p$ with $i=1,2$. Since each $f_i$ is a graceful labelling, so we set $f_i(x_{i,j})=j-1$ for $j\in [1,s_i]$, $f_i(y_{i,j})=s_i+j-1$ for $j\in [1,t_i]$ and $f_i(x_{i,s}y_{i,t})=f_i(y_{i,t})-f_i(x_{i,s})=s_i+t-s$ for each edge $x_{i,s}y_{i,t}\in E(T_i)$, and $f_i(E(T_i))=[1,p]$ with $i=1,2$.

We define $g_i(x_{i,j})=f_i(x_{i,j})+1$ for $x_{i,j}\in X_i$, $g_i(y_{i,j})=f_i(y_{i,t_i-j+1})+1$ for $y_{i,j}\in Y_i$, and $g_i(x_{i,s}y_{i,t})=f_i(x_{i,s}y_{i,t})+p$ for $x_{i,s}y_{i,t}\in E(T_i)$ with $i=1,2$. Notice that $f_i(y_{i,j})+f_i(y_{i,t_i-j+1})=2s_i+t_i-1=s_i+p-1$ with $i=1,2$. Thus,
$${
\begin{split}
&\quad g_i(x_{i,s})+g_i(x_{i,s}y_{i,t})+g_i(y_{i,t})\\
&=f_i(x_{i,s})+1+f_i(x_{i,s}y_{i,t})+p+f_i(y_{i,t_i-t+1})+1\\
&=f_i(x_{i,s})+f_i(x_{i,s}y_{i,t})+p+s_i+p-f_i(y_{i,t})+1\\
&=2p+s_i+1
\end{split}}
$$ with $i=1,2$, since $f_i(x_{i,s}y_{i,t})=f_i(y_{i,t})-f_i(x_{i,s})$. Thereby, each $g_i$ is an edge-magic graceful labelling, such that $g_i(V(T_i))=[1,p]$ and $g_i(E(T_i))=[p+1,2p-1]$ with $i=1,2$. Again, we define another labelling $h_2$ of $T_2$ as: $h_2(w)=2p-g_2(w)$ for $w\in V(T_2)\cup E(T_2)$. Hence, we have $h_2(E(T_2))=[1,p-1]$ and $h_2(V(T_2))=[p,2p-1]$, so $g_1(V(T_1))\setminus \{p\}=h_2(E(T_2))$, $g_1(E(T_1))=h_2(V(T_2))\setminus \{p\}$. By Definition \ref{defn:reciprocal-inverse}, $g_1$ and $h_2$ are reciprocal-inverse to each other.
\end{proof}

\subsection{Random Topsnut-sequences for encrypting large scale of files}

Let $T_{i+1}=T_i+L_i$ be a \emph{recursive tree}, where $L_i$ is a leaf set, $T_i+L_i$ is a result of adding randomly leaves of $L_i$ to the tree $T_i$. In other words, $T_i\subset T_{i+1}$, that is, $T_i$ is a subgraph of $T_{i+1}$. If each tree $T_i$ admits a set-ordered graceful labelling $f_i$ with $i\in [1,n]$, we say $\{T_i\}^n_1$ a \emph{set-ordered graceful recursive sequence}. If each tree $H_i$ is obtained by adding randomly leaves of a leaf set $L'_i$ to $T_i\in \{T_i\}^n_1$ with $i\in [1,n]$, we call $\{H_i\}^n_1$ a \emph{leaf-adding associated sequence}, $H_i$ a \emph{leaf-adding associated matching}. By \cite{Zhou-Yao-Chen-Tao2012} and \cite{Zhou-Yao-Chen-2013}, each $H_i\in \{H_i\}^n_1$ admits an odd-graceful labelling and an odd-elegant labelling.

As known, each recursive tree $T_i$ of $\{T_i\}^n_1$ induces a Topsnut-matrix $A_{vev}(T_i)$, and $A_{vev}(T_i)$ distributes a TB-paw $D(T_i)$, so $T_i$ is a public key; a leaf-adding associated matching $H_i$ of $T_i$ corresponds a Topsnut-matrix $A_{vev}(H_i)$, and $A_{vev}(H_i)$ induces a TB-paw $D(H_i)$, so $H_i$ can be considered as a private key. Thereby, we get a pair of matching TB-paws $D(T_i)$ and $D(H_i)$ with $i\in [1,n]$, moreover, two random TB-paw sequences $\{D(T_i)\}^n_1$ and $\{D(H_i)\}^n_1$ can be used to encrypt large scale of files, since $\{T_i\}^n_1$ and $\{H_i\}^n_1$ have random property.

\begin{thm}\label{thm:set-ordered-graceful-recursive-sequence}
If each tree $T_i$ is a caterpillar and $T_i\subset T_{i+1}$ with $i\in [1,n-1]$, then we have a set-ordered graceful recursive sequence $\{T_i\}^n_1$ and its leaf-adding associated sequence $\{H_i\}^n_1$ such that each each $H_i\in \{H_i\}^n_1$ admits an odd-graceful labelling and an odd-elegant labelling.
\end{thm}

We define a parameter sequence
$$\{(k_i,d_i)\}^m_1=\{(k_1,d_1), (k_2,d_2),\dots ,(k_m,d_m)\}$$
and introduce a \emph{Topsnut-gpw sequence} $\{G_{(k_i,d_i)}\}^m_1$ made by an integer sequence $\{(k_i,d_i)\}^m_1$ and a $(p, q)$-graph $G$, where each Topsnut-gpw $G_{(k_i,d_i)}\cong G$. Let $$S(k_i,d_i)^q_1=\{k_i, k_i + d_i, \dots , k_i +(q-1)d_i\}$$ be a recursive set for integers $k_i\geq 1$, $d_i\geq 1$. Each Topsnut-gpw $G_{(k_i,d_i)}\in \{G_{(k_i,d_i)}\}^m_1$ admits one labelling of four parameter labellings defined in Definition \ref{defn:3-parameter-labellings}.

\begin{defn} \label{defn:3-parameter-labellings}
\cite{Gallian2016} (1) A \emph{$(k_i,d_i)$-graceful labelling} $f$ of $G_i$
hold $f(V(G_i))\subseteq [0, k_i + (q-1)d_i]$, $f(x)\neq f(y)$ for
distinct $x,y\in V(G_i)$ and $\pi(E(G_i))=\{|\pi(u)-\pi(v)|;\ uv\in
E(G_i)\}=S(k_i,d_i)^q_1$.

(2) A labeling $f$ of $G_i$ is said to be $(k_i,d_i)$-\emph{arithmetic} if $f(V(G_i))\subseteq [0, k_i+(q-1)d_i]$, $f(x)\neq f(y)$ for distinct $x,y\in
V(G_i)$ and $\{f(u)+f(v): uv\in E(G_i)\}=S(k_i,d_i)^q_1$.

(3) A $(k_i,d_i)$-\emph{edge antimagic total
labelling} $f$ of $G_i$ hold $f(V(G_i)\cup E(G_i))=[1,p+q]$ and
$\{f(u)+f(v)+f(uv): uv\in E(G_i)\}=S(k_i,d_i)^q_1$, and furthermore $f$ is
\emph{super} if $f(V(G_i))=[1,p]$.

(4) A \emph{$(k_i,d_i)$-harmonious labelling} of a $(p,q)$-graph $G_i$ is defined by a mapping $h:V(G)\rightarrow [0,k+(q-1)d_i]$ with $k_i,d_i\geq 1$, such that $f(x)\neq f(y)$ for any pair of vertices $x,y$ of $G$, $h(u)+h(v)(\bmod^*~qd_i)$ means that $h(uv)-k=[h(u)+h(v)-k](\bmod~qd_i)$ for each edge $uv\in E(G)$, and the edge label set $h(E(G))=\{k_i,k_i+d_i,\dots, k_i+(q-1)d_i\}$ holds true.\qqed
\end{defn}

\begin{figure}[h]
\centering
\includegraphics[height=5cm]{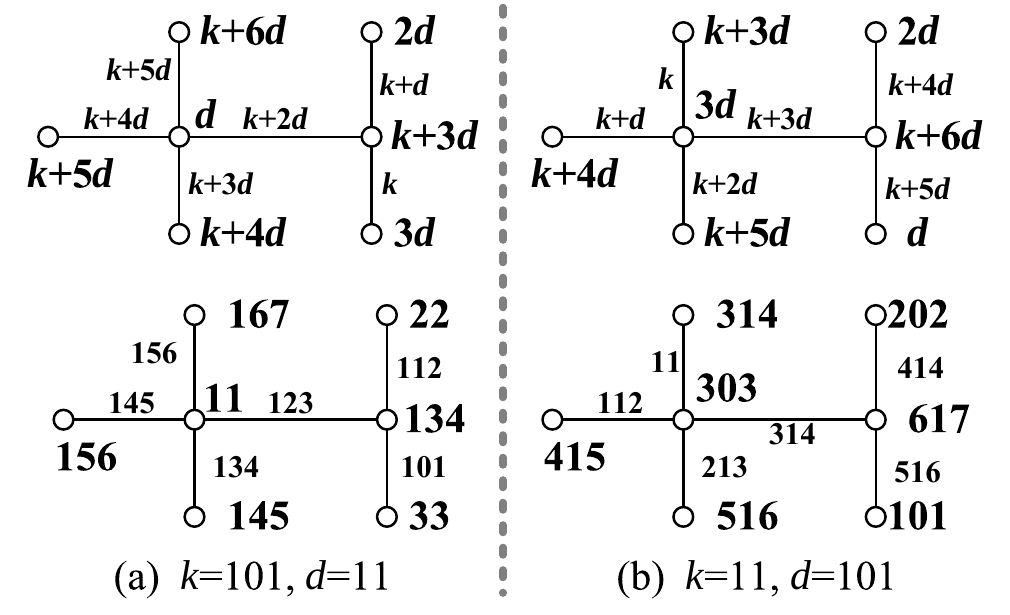}
\caption{\label{fig:k-d-sequence-11}{\small Based on Fig.\ref{fig:image-labelling-2}(g), there are: (a) A $(101,11)$-graceful image-labellings; (b) a $(11,101)$-graceful image-labellings.}}
\end{figure}

Fig.\ref{fig:k-d-sequence-11}(a) distributes a TB-paw
$${
\begin{split}
D_a=&11123134134145145156156\\
&16711134123111122210133
\end{split}}
$$
and another TB-paw
$${
\begin{split}
D_b=&3031131411241521351631461\\
&7303617516101414202314303.
\end{split}}
$$

We can make TB-paws, like $D_a$ and $D_b$, having bytes as more as we desired.

The complex of a Topsnut-gpw sequence $\{G_{(k_i,d_i)}\}^m_1$ is:

(i) $\{(k_i,d_i)\}^m_1$ is a random sequence or a sequence with many restrictions.

(ii) $G_{(k_i,d_i)}\cong G$ is a regularity.

(iii) Each $G_{(k_i,d_i)}\in \{G_{(k_i,d_i)}\}^m_1$ admits randomly one labelling in Definition \ref{defn:3-parameter-labellings}.

(iv) Each $G_{(k_i,d_i)}\in \{G_{(k_i,d_i)}\}^m_1$ has its matching $H_{(k_i,d_i)}\in \{H_{(k_i,d_i)}\}^m_1$ under the meaning of image-labelling, inverse labelling and twin labelling, and so on.

\vskip 0.4cm

\emph{Applying Topsnut-gpw sequences in encrypting graphs/networks.} In the subsection of ``Graphs labelled by every-zero graphic groups'', we have proposed a new topic of encrypting graphs (or networks, dynamic networks). Encrypting graphs/networks can be related with Topsnut-gpw sequences $\{G_{(k_i,d_i)}\}^m_1$.
\begin{defn} \label{defn:Topsnut-gpw-sequences-graph-labellings}
Let $\{(k_i,d_i)\}^m_1$ be a sequence with integers $k_i\geq 0$ and $d_i\geq 1$, and $G$ be a $(p,q)$-graph with $p\geq 2$ and $q\geq 1$. We define a labelling $F:V(G)\rightarrow \{G_{(k_i,d_i)}\}^m_1$, and $F(u_iv_j)=(|k_i-k_j|, ~d_i+d_j~(\textrm{mod}~M))$ with $F(u_i)=G_{(k_i,d_i)}$ and $F(v_j)=G_{(k_j,d_j)}$ for each edge $u_iv_j\in E(G)$. Then

(1) If $\{|k_i-k_j|:~u_iv_j\in E(G)\}=[1,2q-1]^o$ and $\{d_i+d_j~(\textrm{mod}~M)):~u_iv_j\in E(G)\}=[0,2q-3]^o$, we call $F$ a \emph{twin odd-type graph-labelling} of $G$.

(2) If $\{|k_i-k_j|:~u_iv_j\in E(G)\}=[1,q]$ and $\{d_i+d_j~(\textrm{mod}~M)):~u_iv_j\in E(G)\}=[0,2q-3]^o$, we call $F$ a \emph{graceful odd-elegant graph-labelling} of $G$.

(3) If $\{|k_i-k_j|:~u_iv_j\in E(G)\}$ and $\{d_i+d_j~(\textrm{mod}~M)):~u_iv_j\in E(G)\}$ are generalized Fibonacci sequences, we call $F$ a \emph{twin Fibonacci-type graph-labelling} of $G$.\qqed
\end{defn}

Clearly, we can define more types Topsnut-gpw sequence graph-labelling for the requirements of real application.

\subsection{Twin matchings}

A phenomenon about twin labellings was proposed and discussed in \cite{Wang-Xu-Yao-2017-Twin}, that is, the twin odd-graceful labellings are natural-inspired as keys and locks. In fact, each type of twin labellings can be considered as a matching. We have other twin labellings, such as image-labellings, inverse labellings. We view many examples for twin labellings, and want to discover that twin labellings have some properties like quantum entanglement.

\begin{defn}\label{defn:twin-labellings}
$^*$ Suppose $f:V(G)\rightarrow [0,2q-1]$ is an odd-graceful labelling of a $(p,q)$-graph $G$ and $g:V(H)\rightarrow [1,2q]$ is a labelling of another $(p',q')$-graph $H$ such that each edge $uv\in E(H)$ has its own label defined as $h(uv)=|h(u)-h(v)|$ and the edge label set $f(E(H))=[1,2q-1]^o$. We say $(f,g)$ to be a \emph{twin odd-graceful labellings}, $H$ a \emph{twin odd-graceful matching} of $G$. \qqed
\end{defn}

We point out that Definition \ref{defn:twin-labellings} contains the definition of twin odd-graceful labellings defined in \cite{Wang-Xu-Yao-2017-Twin}, since we consider the case of non-tree bipartite graphs having twin odd-graceful labellings. If $f(V(G))\cap f(V(H))\neq \emptyset$ in Definition \ref{defn:twin-labellings}, we coincide the vertex $x$ of $G$ having $f(x)=g(y)$ with the vertex $y$ of $H$ into one, until the resulting graph has no two vertices being labelled with the same integer. We denote this graph as $G\odot H$. Clearly, the edges of $G\odot H$ are labelled by two groups of $1,3,5,\dots, 2q-1$. We are interesting on looking for all twin odd-graceful matchings of $G$. It is not hard to see that if $G$ admits different odd-graceful labellings $f_1,f_2,\dots ,f_m$, each $f_i$ may induce twin odd-graceful matchings $H_{i,1},H_{i,2},\dots ,H_{i,m_i}$ of $G$ with $i\in [1,m]$ (see Fig.\ref{fig:looking-twin-odd-graceful-1} and Fig.\ref{fig:looking-twin-odd-graceful-2}).

\begin{lem}\label{thm:tree-twin-odd-graceful}
Suppose that each tree $T_i$ of $p$ vertices admits a set-ordered odd-graceful labelling $f_i$ with $i=1,2$. If $f_1(V(T_1))=f_2(V(T_2))$, then $T_i$ is a twin odd-graceful (resp. odd-elegant) matching of $T_{3-i}$ with $i=1,2$.
\end{lem}
\begin{proof} Let $(X_i,Y_i)$ be the bipartition of $V(T_i)$ with $i=1,2$. So, each label of $f_i(X_i)$ is even, and each label of $f_i(Y_i)$ is odd, and $\max f_i(X_i)<\min f_i(Y_i)$ for $i=1,2$. From $f_1(V(T_1))=f_2(V(T_2))$ and parity, we have $f_1(X_1)=f_2(X_2)$, $f_1(Y_1)=f_2(Y_2)$. We define another labelling $g$ of $T_2$ as: $g(w)=f_2(w)$ for $w\in V(T_2)$, so $g(uv)=|g(u)-g(v)|=|f_2(u)-f_2(v)|$ for each edge $uv\in E(T_2)$. We can see $g(V(T_2))\subset [1, 2|E(T_2)|]=[1,2(p-1)]$, $g(E(T_2))=[1,2(p-1)]^o$. Thereby, $T_2$ with the set-ordered odd-graceful labelling $g$ is just a twin odd-graceful matching of $T_1$.

The above proof, also, show that $T_i$ is really a twin odd-graceful matching of $T_{3-i}$ with $i=1,2$.

Since the proof for odd-elegant matching is similar with above proof, we omit it here.
\end{proof}

Lemma \ref{thm:tree-twin-odd-graceful} implies that any tree admitting a set-ordered odd-graceful labelling is a twin odd-graceful matching of itself (see Fig.\ref{fig:looking-twin-odd-graceful-0} (a) and (b)). In general, a tree admitting a set-ordered odd-graceful labelling may have two or more twin odd-graceful matchings in trees, or non-tree graphs, or disconnected graphs (see Fig.\ref{fig:looking-twin-odd-graceful-0} (c) and (d)). Fig.\ref{fig:looking-twin-odd-graceful-0} distributes three twin odd-graceful matchings $T\odot T_i$ with $i=1,2,3$, after coinciding two vertices labbelled with 15 into one.

\begin{figure}[h]
\centering
\includegraphics[height=5.6cm]{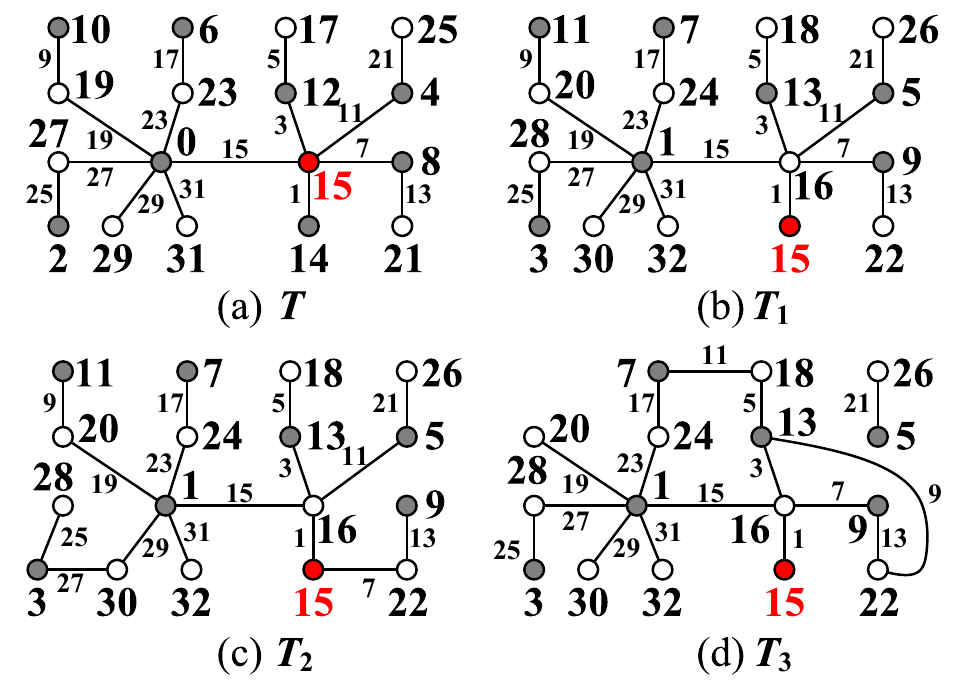}
\caption{\label{fig:looking-twin-odd-graceful-0}{\small (a) A tree $T$ admits a set-ordered odd-graceful labelling; (b) $T$ is a twin odd-graceful matching of itself by Lemma \ref{thm:tree-twin-odd-graceful}; (c) another twin odd-graceful matching of $T$; (d) a disconnected twin odd-graceful matching of $T$.}}
\end{figure}

If $G$ is a non-tree $(p,q)$-graph admitting an odd-graceful labelling, we do not have some efficient methods for determining twin odd-graceful matchings of $G$. We present some examples depicted in Fig.\ref{fig:looking-twin-odd-graceful-1}. Notice that two odd-graceful $(7,7)$-graphs $G_1$ and $G_2$ have their twin odd-graceful matchings with $H_{1j}\cong H_{2j}$ for $j\in [1,6]$, in other words, the twin odd-graceful matchings of $G_1$ and $G_2$ keep isomorphic configuration, so $G_1$ and $G_2$ are \emph{twisted} under the isomorphic configuration of their own twin odd-graceful matchings. However, the twin odd-graceful matching $H_{3j}~(j\in [1,6])$ of the odd-graceful $(7,7)$-graph $G_3$ shown in Fig.\ref{fig:looking-twin-odd-graceful-2} are not isomorphic to that of $G_1$ and $G_2$ appeared in Fig.\ref{fig:looking-twin-odd-graceful-1}. The above examples tell us that finding twin odd-graceful matchings of a non-tree graph is not a slight work.

\begin{figure}[h]
\centering
\includegraphics[height=7.8cm]{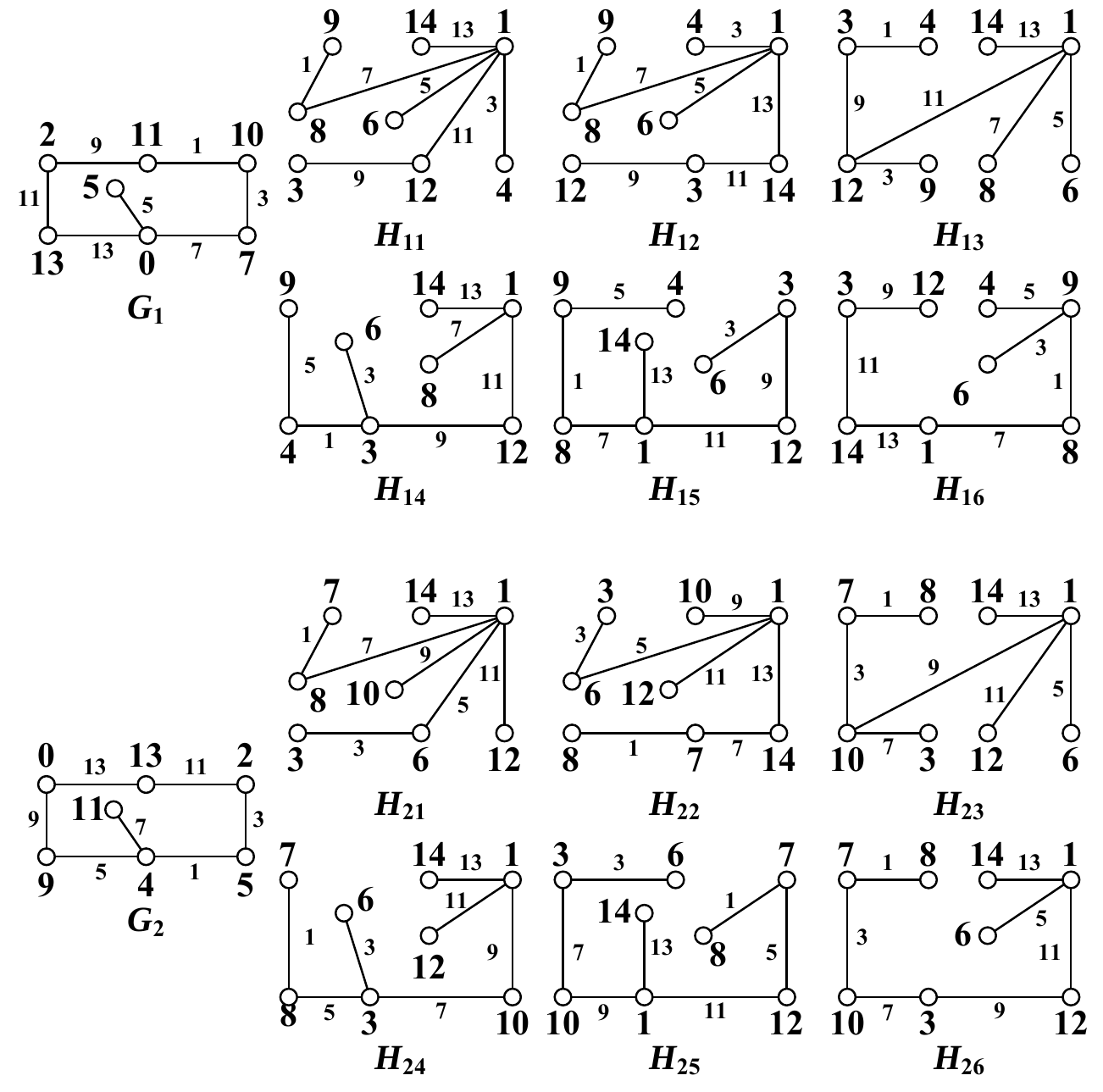}
\caption{\label{fig:looking-twin-odd-graceful-1}{\small Each odd-graceful $(7,7)$-graph $G_i$ has its own twin odd-graceful matching $H_{ij}$ for $j\in [1,6]$ and $i=1,2$.}}
\end{figure}

\begin{figure}[h]
\centering
\includegraphics[height=4cm]{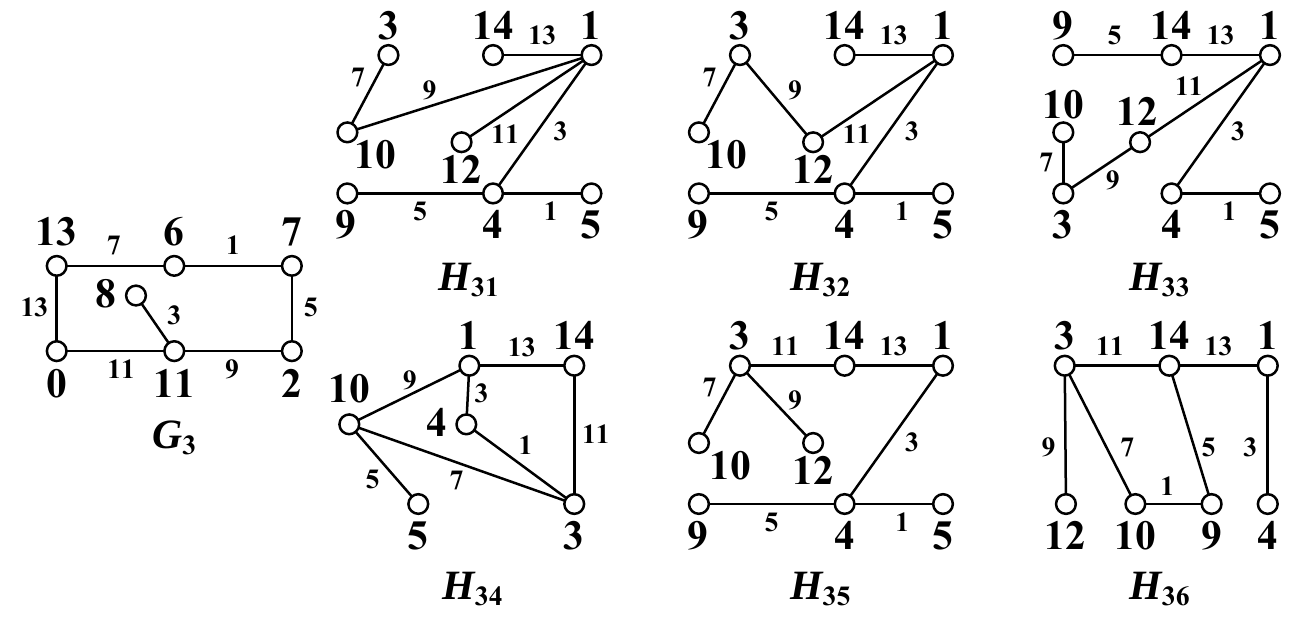}
\caption{\label{fig:looking-twin-odd-graceful-2}{\small The odd-graceful $(7,7)$-graph $G_3$ has its own twin odd-graceful matching $H_{3j}$ for $j\in [1,6]$ that differ from those $H_{ij}$ shown in Fig.\ref{fig:looking-twin-odd-graceful-1} under configuration meaning.}}
\end{figure}

\subsection{Every-zero Topsnut-matrix groups}

An every-zero Topsnut-matrix group $M_{n}(A(G),f)$ of a $(p,q)$-graph $G$ is made by a Topsnut-matrix $A_1(G)=(X~W~Y)^{-1}$ of $G$ in the way: $A_i(G)=(X_i~W~Y_i)^{-1}$ with $i\in [1,n]$, where $X_i=(x_1+(i-1),~x_2+(i-1),~\cdots ~,x_q+(i-1))~(\textrm{mod}~n)$ and $Y_i=(y_1+(i-1),~y_2+(i-1),~\cdots ~,x_y+(i-1))~(\textrm{mod}~n)$, and $W=(e_1,~e_2,~\cdots ~,e_q)$. An example is shown in Fig.\ref{fig:Topsnut-matrix-group}. We define an operation $\oplus$ in the form $A_i(G)\oplus A_j(G)$, where $A_i(G), A_j(G)\in M_{n}(A(G),f)$. The addition operation
\begin{equation}\label{eqa:Topsnut-matrix-group11}
A_i(G)\oplus A_j(G)=A_{i+j-k~(\textrm{mod}~n)}(G)
\end{equation}under a zero $A_k(G)\in M_{n}(A(G),f)$ means that $X_i\oplus X_j$ and $Y_i\oplus Y_j$ defined by
\begin{equation}\label{eqa:Topsnut-matrix-group22}
x_i+x_j-x_k=x_{i+j-k~(\textrm{mod}~n)}
\end{equation}
and
\begin{equation}\label{eqa:Topsnut-matrix-group33}
y_i+y_j-y_k=y_{i+j-k~(\textrm{mod}~n)}.
\end{equation}
By (\ref{eqa:Topsnut-matrix-group11}), (\ref{eqa:Topsnut-matrix-group22}) and (\ref{eqa:Topsnut-matrix-group33}), we can prove that $M_{n}(A(G),f)$ is really an every-zero group, we call it an \emph{every-zero Topsnut-matrix group}.

\begin{figure}[h]
\centering
\includegraphics[height=10cm]{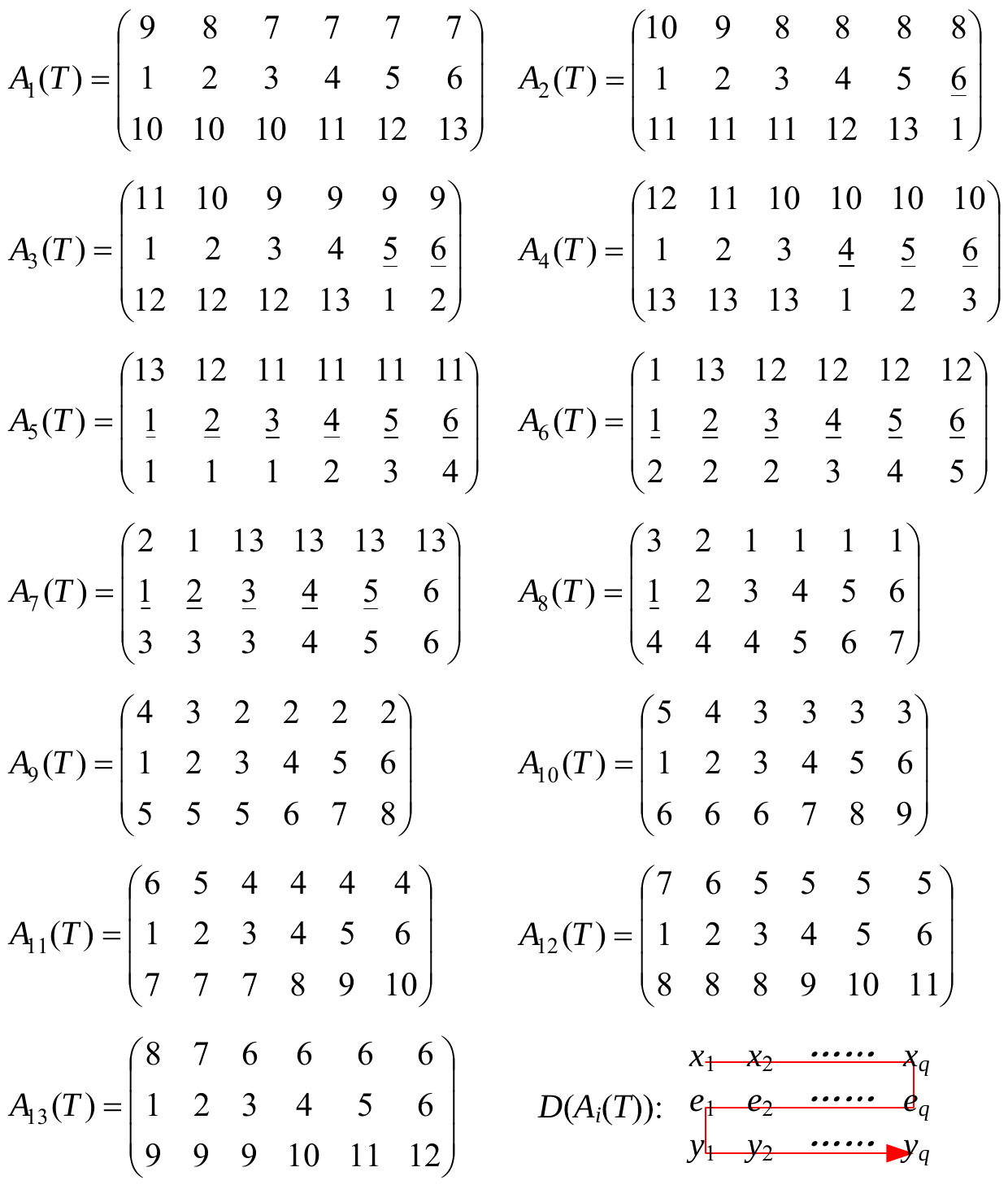}
\caption{\label{fig:Topsnut-matrix-group}{\small An every-zero Topsnut-matrix group $M_{13}(A(T),f)$.}}
\end{figure}

Since there are $N_{tbp}(G)=3 q(1+3q)\cdot q!\cdot 2^{q-1}$ every-zero Topsnut-matrix groups in total, we can use them to encrypt communities of a dynamic network or networks. According to $D(A_i(T))$ shown in Fig.\ref{fig:Topsnut-matrix-group}, we get an \emph{every-zero TB-paw group} $D_{13}(T,f)=\{D(A_i(T)):i\in [1,13]\}$, where each $D(A_i(T))$ is a vev-type TB-paw made by $A_i(T)$ by a fixed method (red narrow-line) shown in Fig.\ref{fig:Topsnut-matrix-group}.

\subsection{Matchings in graphic groups}

We present: ``\emph{an every-zero odd-graceful graphic group $F_n(H,f)=\{H_i\}^n_1$ matches with another odd-graceful graphic group $F_n(L,h)=\{L_i\}^n_1$ if $(H_i,L_i)$ is a twin odd-graceful matching, and $(f_i,h_i)$ is a twin odd-graceful labellings, where $H_i$ admits an odd-graceful labelling $f_i$, and $L_i$ admits pan-odd-graceful labelling $h_i$.}'' An example is shown in Fig.\ref{fig:odd-graceful-group} and Fig.\ref{fig:odd-matching-graphic-group}: an every-zero odd-graceful graphic group $F_{14}(H,f)=\{H_i\}^n_1$ shown in Fig.\ref{fig:odd-graceful-group} matches with another every-zero odd-graceful graphic group $F_n(L,h)=\{L_i\}^n_1$ shown in Fig.\ref{fig:odd-matching-graphic-group}, since $\odot \langle H_i,L_i\rangle$ is a twin odd-graceful matching with $i\in [1,14]$, in which $\odot \langle H_1,L_1\rangle$ is disconnected, and others $\odot \langle H_i,L_i\rangle$ are connected. We can consider $F_{14}(H,f)=\{H_i\}^n_1$ as a group of public keys, and $F_n(L,h)=\{L_i\}^n_1$ as a group of private keys. Correspondingly, Topsnut-matrices $A(H_i)$, $A(L_i)$ and $A(\odot \langle H_i,L_i\rangle)$ distribute three TB-paws $D(H_i)$ (as a public key), $D(L_i)$ (as a private key) and $D(\odot \langle H_i,L_i\rangle)$ (as an authentication).

\begin{figure}[h]
\centering
\includegraphics[height=10cm]{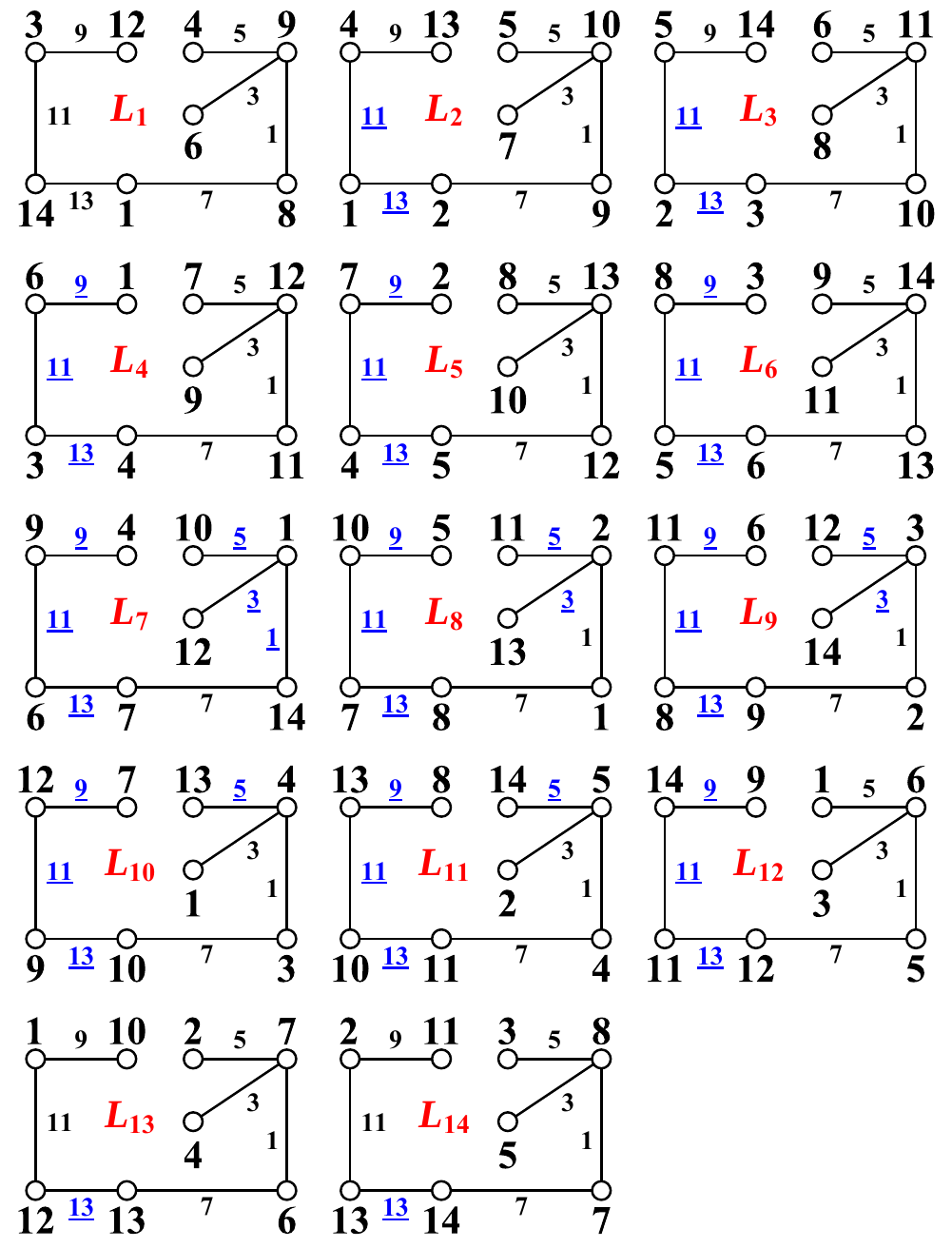}
\caption{\label{fig:odd-matching-graphic-group}{\small An every-zero graphic group $F_{14}(L,h)$ matches with the every-zero graphic group shown in Fig.\ref{fig:odd-graceful-group}, where $L_1=H_{16}$ shown in Fig.\ref{fig:looking-twin-odd-graceful-1}.}}
\end{figure}

Moreover, we have discovered that encrypting a network $T$ by $F_{14}(H,f)=\{H_i\}^{14}_1$ and $F_{14}(L,h)=\{L_i\}^{14}_1$, respectively, the results are the same, see Fig.\ref{fig:group-group-matching}.

\begin{figure}[h]
\centering
\includegraphics[height=2.6cm]{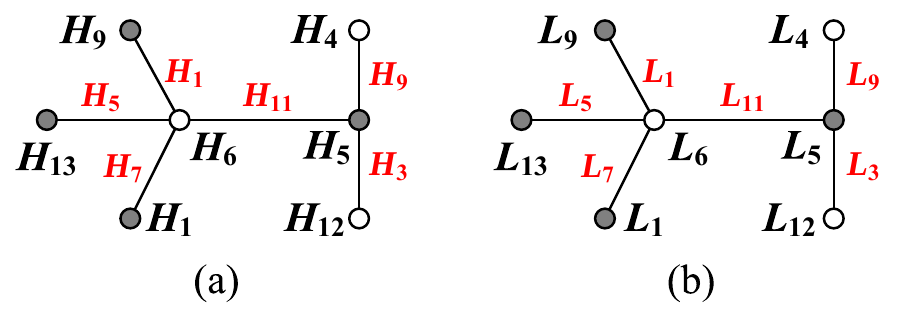}
\caption{\label{fig:group-group-matching}{\small A tree-like network $T$ was encrypted by two every-zero graphic groups $F_{14}(H,f)$ and $F_{14}(L,h)=\{L_i\}^{14}_1$: (a) $N_{et}(T,F_{14}(H,f))$; (b) $N_{et}(T,F_{14}(L,h))$.}}
\end{figure}

\begin{thm} \label{thm:group-group-matching}
Two every-zero odd-graceful graphic groups $F_n(H,f)=\{H_i\}^n_1$ and $F_n(L,h)=\{L_i\}^n_1$ match to each other. Then $F_n(\odot \langle H_i,L_i\rangle, f_i\odot h_i)$ obtained by coinciding each twin odd-graceful matchings $(H_i,L_i)$ together, also, is an every-zero twin odd-graceful graphic group.
\end{thm}

\begin{cor} \label{thm:group-group-matching-cor}
Two every-zero Topsnut-matrix groups $M_{n}(A(G),f)$ and $M_{n}(A(H),h)$ have that $(G,H)$ is a twin odd-graceful matching, $(A(G),A(H))$ is a twin odd-graceful Topsnut-matrix matching. Then $F_n(\odot \langle A(G),A(H)\rangle, f\odot h)$ is an every-zero Topsnut-matrix group.
\end{cor}

\subsection{Noise of TB-paws}

There are many well-known methods for encrypting TB-paws in cryptography. Here, a TB-paw $D$ is encrypted to be another TB-paw $E_{cr}(D)$,  where $E_{cr}$ means ``\emph{encrypt}'', we say $E_{cr}(D)$ a noise of $D$, and then the procedure clearing the noise from $E_{cr}(D)$ to obtain the original TB-paw $D$ is denoted as $D=E_{cr}^{-1}(E_{cr}(D))$.

\emph{Topsnut-gpws made by colors, miscellaneous configurations and various lines.} By comparing Fig.\ref{fig:k-d-sequence-11} with Fig.\ref{fig:colored-Topsnut-gpw}, we can see that the space of pan-Topsnut-gpws is greater than that of Topsnut-gpws, and less than that of colored pan-Topsnut-gpws. Let pentacle=pc, pentagon=pg, rectangle=r, triangle=t, circle=c,
yellow=y, purple=p, red=re, blue=bl, green=g, orange=o and black=b in Fig.\ref{fig:colored-Topsnut-gpw}(b), we can get a colored TB-paw
$${
\begin{split}
D^c_{b}=&op303r11opgbl314bbl112tpo415pb213\\
&gpcbl516blbl314rpcbl617303617blbl\\
&314rop303og414cbb202gb516pgblr101.
\end{split}}
$$
Some researching works on this topic were introduced in \cite{Hongyu-Wang-2018-Doctor-thesis} and \cite{Yao-Mu-Sun-Zhang-Wang-Su-2018}.

\begin{figure}[h]
\centering
\includegraphics[height=3cm]{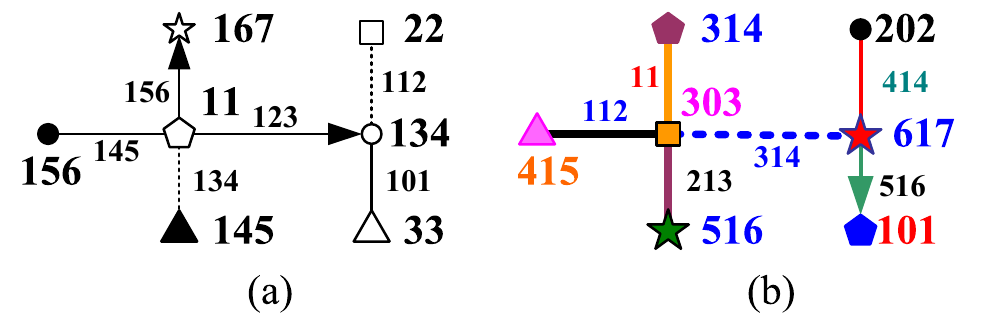}
\caption{\label{fig:colored-Topsnut-gpw}{\small (a) A pan-Topsnut-gpw in black and white; (b) a colored pan-Topsnut-gpw.}}
\end{figure}

\emph{TB-paws made by noise.} We add English letters into the following TB-paw
$${
\begin{split}
D_a=&11123134134145145156156\\
&16711134123111122210133,
\end{split}}
$$
from Fig.\ref{fig:colored-Topsnut-gpw}(a), in order to obtain a noised TB-paw as
$${
\begin{split}
E_{cr}(D_a)=&1v11bf2313ert41h34145cv14515xs6156yu\\
&1w67cdb111fh3412iuj31111j222op1013w3,
\end{split}}
$$
Or we by $x=11$, $y=22$, $z=33$, $a=34$ and $b=56$ replace the same numbers of $D_a$, thus, we get a shorted TB-paw as follows
$${
\begin{split}
E'_{cr}(D_a)=&x1231a1a1451451b1b167x1a123xx2y101z.
\end{split}}
$$
Clearing the noise from $E'_{cr}(D_a)$ needs the substitution of letters ``$x,y,z,a,b$''. Clearly, combining two methods introduce above will produce moore noised TB-paws.

\subsection{Topsnut-networks, graphs labelled by Topsnut-gpws}

We introduce pan-Topsnut-matchings on graphs as follows.

\begin{defn}\label{defn:pan-matching-graphs}
\cite{Yao-Sun-Zhang-Mu-Sun-Wang-Su-Zhang-Yang-Yang-2018arXiv} Let $P_{ag}$ be a set of graphs. A $(p,q)$-graph $G$ admits a graph-labelling $F:V(G)\rightarrow H_{ag}$, and induced edge label $F(uv)=F(u)(\bullet)F(v)$ is just a graph having a pan-matching, where $(\bullet)$ is an operation. Here, a \emph{pan-matching} may be: a perfect matching of $k_{uv}$ vertices, $k_{uv}$-cycle, $k_{uv}$-connected, $k_{uv}$-edge-connected, $k_{uv}$-colorable, $k_{uv}$-edge-colorable, total $k_{uv}$-colorable, $k_{uv}$-regular, $k_{uv}$-girth, $k_{uv}$-maximum degree, $k_{uv}$-clique, $\{a,b\}_{uv}$-factor, v-split $k_{uv}$-connected, e-split $k_{uv}$-connected. We call the graph $\langle G(\bullet) P_{ag}\rangle$ obtained by joining $F(u)$ with $F(uv)$ and joining $F(uv)$ with $F(v)$ for each edge $uv\in E(G)$ a \emph{pan-matching graph}.\qqed
\end{defn}

If each $H_i$ of a \emph{graph set} $P_{ag}$ of graphs admits a labelling $f_i$, we can label a $(p,q)$-graph $G$ in the way: $F:V(G)\rightarrow P_{ag}$, such that each edge $u_iv_j\in E(G)$ is balled by $F(u_iv_j)=F(u_i)\bullet F(v_j)=H_i\bullet H_j=\alpha_{i,j}(f_i,f_j)$, where $\alpha(f_i,f_j)$ is a \emph{reversible function} with $f_j=\alpha_{i,j}(f_i)$, and $f_i=\alpha^{-1}_{i,j}(f_j)$.

\vskip 0.4cm

\subsubsection{Self-similar Topsnut-networks} There are many self-similar networks in the world. We will construct self-similar Topsnut-networks and then label them to generate more complex Topsnut-gpws for producing TB-paws. Let us star with an example depicted in Fig.\ref{fig:self-similar-00}.

\emph{Step 1.} We use a labelling $F_1$ to label the vertices of $H$ shown in Fig.\ref{fig:odd-graceful-group} by the elements of an every-zero graphic group $F_{14}(H,f)$ pictured in Fig.\ref{fig:odd-graceful-group}, such that each edge $uv\in E(H)$ is labelled as
\begin{equation}\label{eqa:c3xxxxx}
H_{uv}=F_1(uv)=F_1(u)\oplus F_1(v)=H_u\oplus H_v
\end{equation}
for $H_u, H_v\in F_{14}(H,f)$, and $F_1(E(H))=\{F(uv):uv\in V(H)\}=\{H_1,H_3,H_5,H_7,H_9,H_{11}, H_{13}\}$. The labelled well graph is denoted as $G$ (see Fig.\ref{fig:self-similar-00}), and we say $G$ admits an \emph{odd-graceful group-cloring/group-labelling}. Next, we  join some vertex $x_u$ of $H_u$ with some vertex $x_{uv}$ of $H_{uv}$ by an edge $x_ux_{uv}$, and join some vertex $y_v$ of $H_v$ with some vertex $y_{uv}$ of $H_{uv}$ by an edge $y_vy_{uv}$. Finally, we have constructed a large graph $G_1=\langle H\leftarrow F_{14}(H,f)\rangle $ (see Fig.\ref{fig:self-similar-00}), and write $E(G_1)=E_{1,1}\cup E_{1,2}$, where each edge of $E_{1,1}$ is not labelled, each edge of $E_{1,2}$ is labelled. We say each subgraph $H_i$ of $G_1$ to be a \emph{block}.

\emph{Step 2.} Next, we define a labelling $F_2$ to label the vertices of each block $H_i$ by the elements of the every-zero graphic group $F_{14}(H,f)$, such that each edge $xy\in E(H_i)$ of the block $H_i$ of $G_1$ is labelled as
$$H_{xy}=F_2(xy)=F_2(x)\oplus F_2(y)=H_x\oplus H_y,$$
under the \emph{zero} $H_i$ (we call this case as \emph{self-zero} hereafter, see Fig.\ref{fig:self-similar-11})  and join some vertex $s_x$ of $H_x$ with some vertex $t_{xy}$ of $H_{xy}$ by an edge $s_xt_{xy}$, and join some vertex $a_y$ of $H_y$ with some vertex $b_{xy}$ of $H_{xy}$ by an edge $a_yb_{xy}$. The resulting graph is denoted as $G_2$. Thereby, $E(G_2)=E_{2,1}\cup E_{2,2}$, where each edge of $E_{2,1}$ is not labelled, each edge of $E_{2,2}$ is labelled. This procedure is called ``doing a graph-labelling to $G_1$ by $F_{n}(H,f)$ under the self-zero'', we write $G_2=\langle  G_1\leftarrow F_{14}(H,f)\rangle $.

\emph{Step 3.} Go on in this way, we get self-similar Topsnut-networks $G_1,G_2,\dots ,G_n$ with $G_{i+1}=\langle  G_i\leftarrow F_{14}(H,f)\rangle $ for $i\in [1,n-1]$ (see Fig.\ref{fig:self-similar-22}), in which each $G_j$ is similar with the origin graph $H$ exhibited in Fig.\ref{fig:self-similar-00}.

\begin{figure}[h]
\centering
\includegraphics[height=12.4cm]{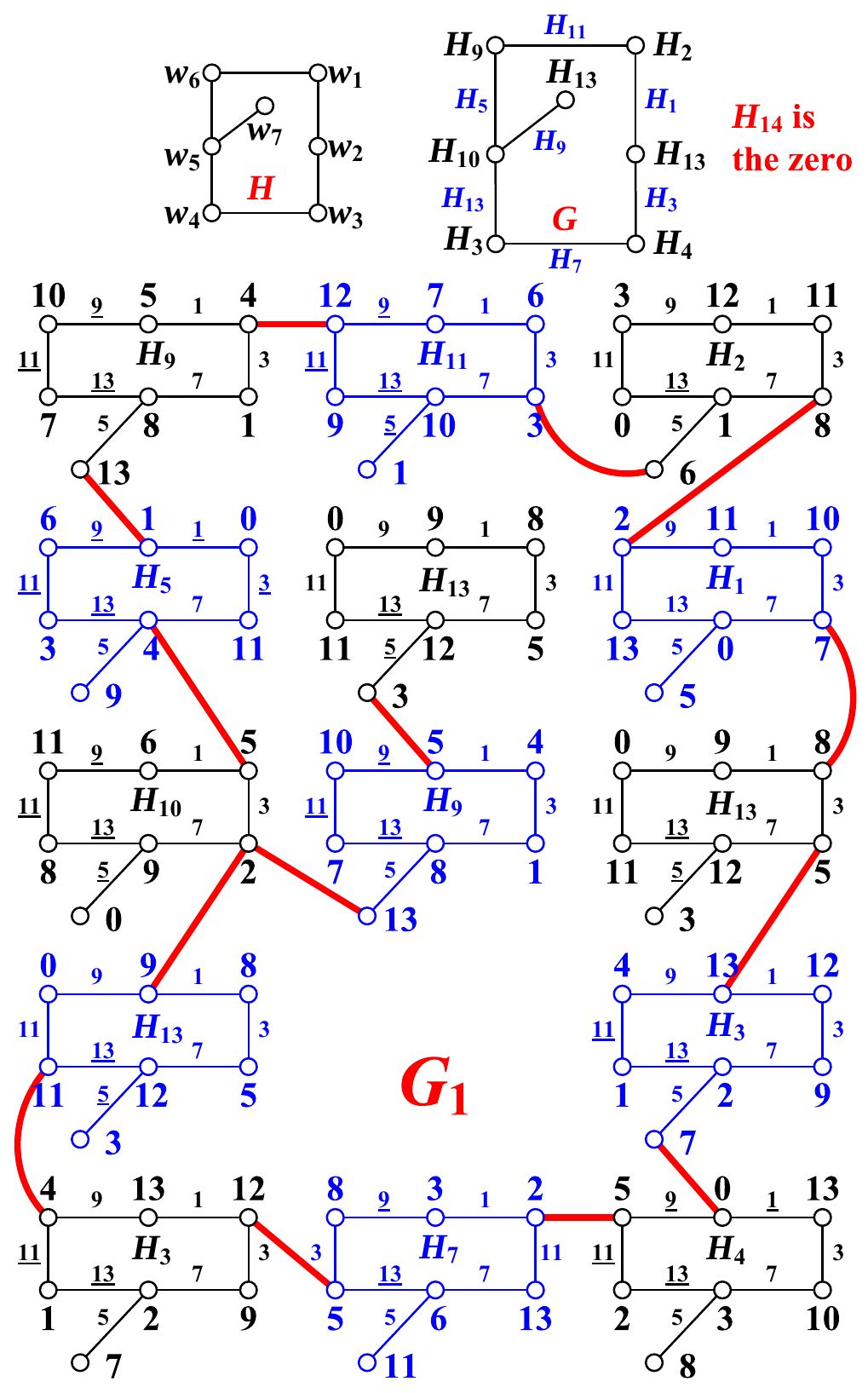}
\caption{\label{fig:self-similar-00}{\small $H$ is shown in Fig.\ref{fig:odd-graceful-group}; $G$ is the resulting from $H$ labelled by the elements of an every-zero graphic group $F_{14}(H,f)$ pictured in Fig.\ref{fig:odd-graceful-group} and admits an odd-graceful group-coloring (there are two $H_{13}$); and $G_1$ is a \emph{self-similar Topsnut-network}.}}
\end{figure}

\begin{figure}[h]
\centering
\includegraphics[height=10.4cm]{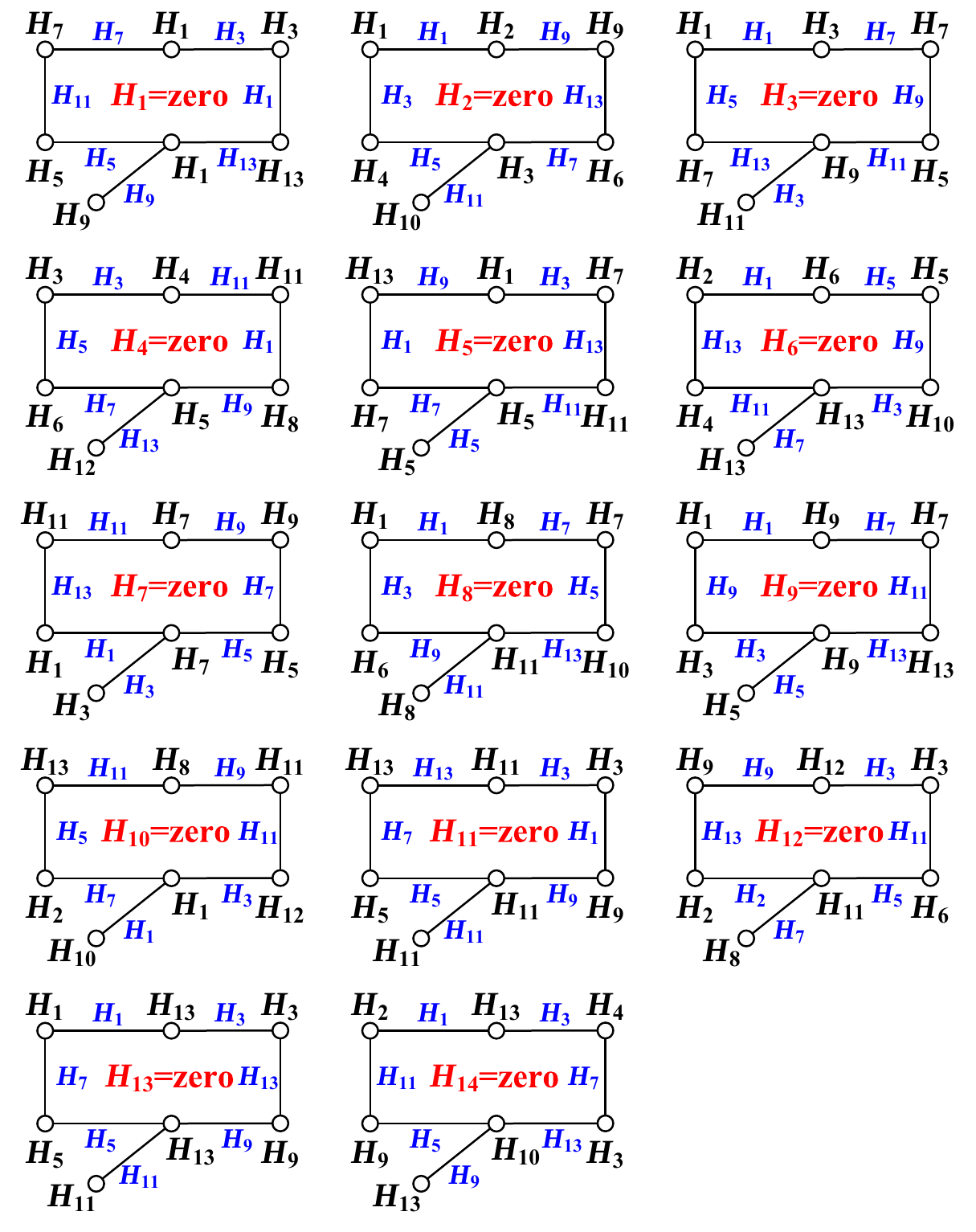}
\caption{\label{fig:self-similar-11}{\small In an every-zero graphic group $F_{14}(H,f)$, $H$ admits an odd-graceful group-coloring under each self-zero $H_i$ with $i\in [1,14]$.}}
\end{figure}

\begin{figure}[h]
\centering
\includegraphics[height=6cm]{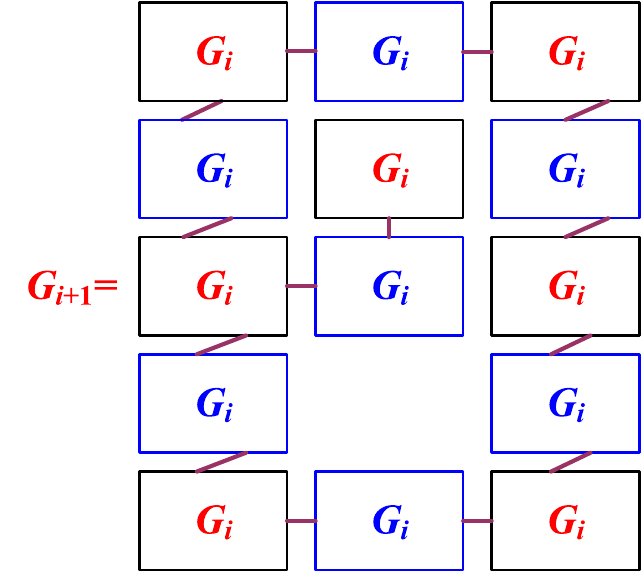}
\caption{\label{fig:self-similar-22}{\small A scheme of the construction of the self-similar Topsnut-network $G_{i+1}=\langle G_i\leftarrow F_{14}(H,f)\rangle$.}}
\end{figure}

We have the construction of the self-similar Topsnut-network $G_{i+1}=\langle G_{i}\leftarrow F_{14}(H,f)\rangle$ shown in Fig.\ref{fig:self-similar-22} based on the self-zeros shown in Fig.\ref{fig:self-similar-11}. It is noticeable, $H$ admits an odd-graceful group-coloring under each self-zero $H_i$ of the every-zero graphic group $F_{14}(H,f)$. By Theorem \ref{them:trees-sequence-group-coloring}, we have
\begin{thm} \label{them:odd-graceful-group-coloring-each-self-zero}
If the graph  $H$ in an every-zero graphic group $F_{n}(H,f)$ is a tree, then $H=H_1$ and each self-similar Topsnut-network $\langle \langle  H_i\leftarrow F_{n}(H,f)\rangle$ with $i\in [1,n]$ admits an odd-graceful group-coloring/labelling under each self-zero $H_i\in F_{n}(H,f)$.
\end{thm}

Notice that each  self-similar Topsnut-network $\langle \langle  H_i\leftarrow F_{n}(H,f)\rangle$ admits an odd-graceful group-coloring/labelling differs from $\langle G_i\leftarrow F_{14}(H,f)\rangle$. In general, we use an every-zero graphic group $F_{n}(H,f)$ to label the graph $H$ to produce a self-similar Topsnut-networks $G_1$, we call such procedure as ``doing a graph-labelling to $H$ by $F_{n}(H,f)$ under the self-zero'', denoted as $G_1=\langle H \leftarrow F_{n}(H,f)\rangle $, so ``doing a graph-labelling to $G_1$ by $F_{n}(H,f)$ under the self-zero'' gives us $G_2=\langle  G_1\leftarrow F_{n}(H,f)\rangle $, $\cdots$, ``doing a graph-labelling to $G_i$ by $F_{n}(H,f)$ under the self-zero'' gives us $G_{i+1}=\langle  G_i\leftarrow F_{n}(H,f)\rangle $ for $i\in [1,n-1]$.  Clearly, $E(G_i)=E_{i,1}\cup E_{i,2}$ such that $\bigcup ^{i-1}_{k=1}E_{k,1}\subset E_{i,1}$.

There are the following advantages about self-similar Topsnut-networks for the difficulty of Topsnut-gpws:

(1) We can use a zero $H_k$ in the construction of self-similar Topsnut-networks for all ``doing a graph-labelling to $G_i$ by $F_{n}(H,f)$'', so a  self-similar Topsnut-network may differ from other self-similar Topsnut-network.

(2) Since here are many  ways to join two blocks, so  there are many self-similar Topsnut-networks $G_i=\langle  G_{i-1}\leftarrow F_{n}(H,f)\rangle $ at time step $i$.

(3) Self-similar Topsnut-networks have giant numbers of vertices and edges based on  $F_{n}(H,f)$ with the smaller numbers of vertices and edges of $H$.

(4) We can relabel those unlabelled edges of a  self-similar Topsnut-networks $G_{i+1}=\langle  G_i\leftarrow F_{n}(H,f)\rangle $ for $i\in [1,n-1]$ for generating Topsnut-gpws.

\vskip 0.4cm

\subsubsection{Graphs labelled by planer graphs} We present the following methods for labelling graphs.

\emph{$\bullet$ Edge-magic total graph-labelling.}  We build up a connection between popular labellings and graph-labellings as follows.
\begin{defn}\label{defn:Edge-magic-total-graph-labelling}
$^*$ Let $M_{pg}(p,q)$ be the set of maximal planar graphs $H_i$ of $i+3$ vertices with $i\in [1,p+q]$, where each face of each planar graph $H_i$ is a triangle. We use a \emph{total labelling} $f$ to label the vertices and edges of a $(p,q)$-graph  $G$ with the elements of $M_{pg}(p,q)$, such that $i+ij+j=k$ (a constant), where $f(u_i)=H_i$, $f(u_iv_j)=H_{ij}$ and $f(v_j)=H_j$ for each edges $u_iv_j\in E(G)$. We say $f$ an \emph{edge-magic total graph-labelling} of $G$ based on $M_{pg}(p,q)$.\qqed
\end{defn}

Definition \ref{defn:Edge-magic-total-graph-labelling} tells us there are many graph-labellings that are similar with popular labellings introduced in \cite{Bondy-2008} and \cite{Yao-Sun-Zhang-Mu-Sun-Wang-Su-Zhang-Yang-Yang-2018arXiv}.

\begin{thm}\label{thm:edge-magic-total-graph-labelling-Mpg}
If a tree of $p$ vertices admits a set-ordered graceful labelling, then it  admits an edge-magic total graph-labelling based on $M_{pg}(p,p-1)$.
\end{thm}

\emph{$\bullet$ Four-coloring triangularly edge-identifying graph-labelling.}   In \cite{YAO-SUN-WANG-SU-XU2018arXiv}, the authors introduce the triangularly edge-identifying operation and triangular edge-subdivision operation. Let $F_{\textrm{TPG}}$ be the set of planar graphs such that each one of $F_{\textrm{TPG}}$ has its outer face to be triangle and a proper 4-coloring. In Fig.\ref{fig:Triangular-operation}, a planar graph $\Delta(T_l,T_r,T_b)$ admits a 4-coloring obtained by three 4-colorings $f_l$, $f_r$ and $f_b$, so $\Delta(T_l,T_r,T_b)\in F_{\textrm{TPG}}$. The procedure of building up $\Delta(T_l,T_r,T_b)$ is called a \emph{triangularly edge-identifying operation}. Conversely, subdividing $G(T_l,T_r,T_b)$ into $T_l,T_r$ and $T_b$ is called a \emph{triangular edge-subdivision operation}. We have:

\begin{defn}\label{defn:4-coloring-triangularly-edge-identifying-graph-labelling}
$^*$ A $(p,q)$-graph  $G$ admits a total labelling $h:V(G)\cup E(G)\rightarrow F_{\textrm{TPG}}$, such that each edge $u_iv_j\in E(G)$ holds that  $f(u_i)=T_i$, $f(u_iv_j)=T_{ij}$ and $f(v_j)=T_j$ induce a planar graph $\Delta(T_i,T_{ij},T_j)$ admitting a 4-coloring, and then we say $G$ admits a \emph{4-coloring triangularly edge-identifying graph-labelling}.\qqed
\end{defn}
\begin{figure}[h]
\centering
\includegraphics[height=2cm]{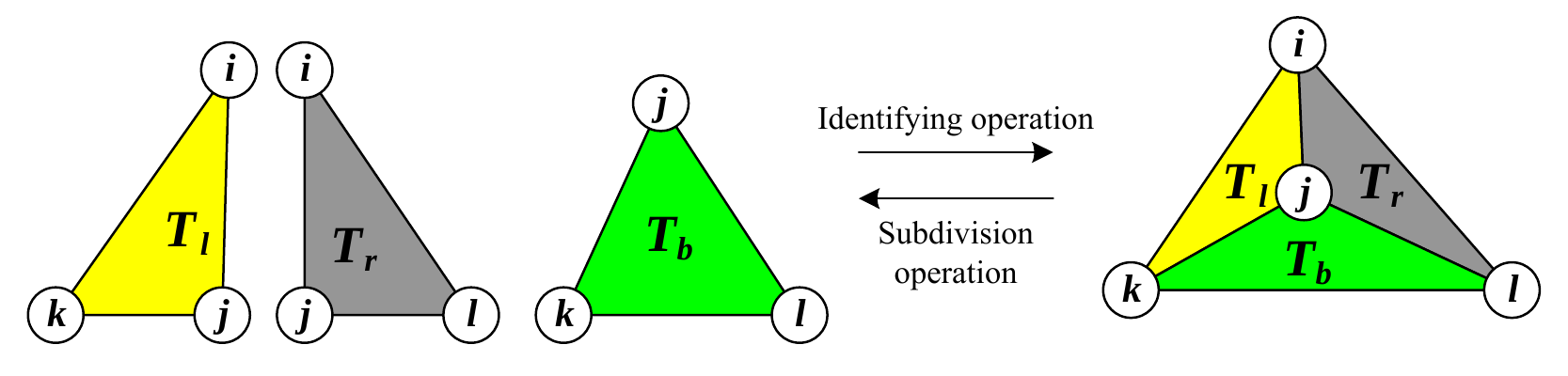}\\
\caption{\label{fig:Triangular-operation}{\footnotesize The scheme for illustrating the triangularly edge-identifying and triangular edge-subdivision operations. }}
\end{figure}

By Definition \ref{defn:4-coloring-triangularly-edge-identifying-graph-labelling},  we use the $(p,q)$-graph  $G$ admitting a 4-coloring triangularly edge-identifying graph-labelling to construct graphs in the following ways:

(1) Notice that $T_i,T_{ij},T_j$ are 4-colorable, for each edge $u_iv_j\in E(G)$,  we coincide a vertex $u_i$ of $T_i$ with another vertex $x_{ij}$ of $T_{ij}$ when $u_i$ and $x_{ij}$ are colored with the same number, and we coincide a vertex $v_j$ of $T_j$ with a vertex $y_{ij}$ of $T_{ij}$ if $v_j$ and $y_{ij}$ are colored with the same number. The resulting graph $\odot\langle G, F_{\textrm{TPG}}\rangle$ is 4-coloring, also, a Topsnut-gpw. Clearly, there are many graphs $\odot\langle G, F_{\textrm{TPG}}\rangle$.

(2) For each edge $u_iv_j\in E(G)$, we coincide $f(u_i)=T_i$, $f(u_iv_j)=T_{ij}$ and $f(v_j)=T_j$ into a planar graph $\Delta(T_i,T_{ij},T_j)$ (see Fig. \ref{fig:Triangular-operation}), the resulting graph denoted as $\odot \langle G, \Delta(T_i,T_{ij},T_j),F_{\textrm{TPG}}\rangle$ is 4-colorable. An example is exhibited in Fig.\ref{fig:Triangular-operation-11}.

By induction, it is not hard to show:

\begin{thm}\label{thm:4-coloring-triangularly-edge-identifying-graph-labellin}
Any tree admits a 4-coloring triangularly edge-identifying graph-labelling.
\end{thm}

\begin{figure}[h]
\centering
\includegraphics[height=5.6cm]{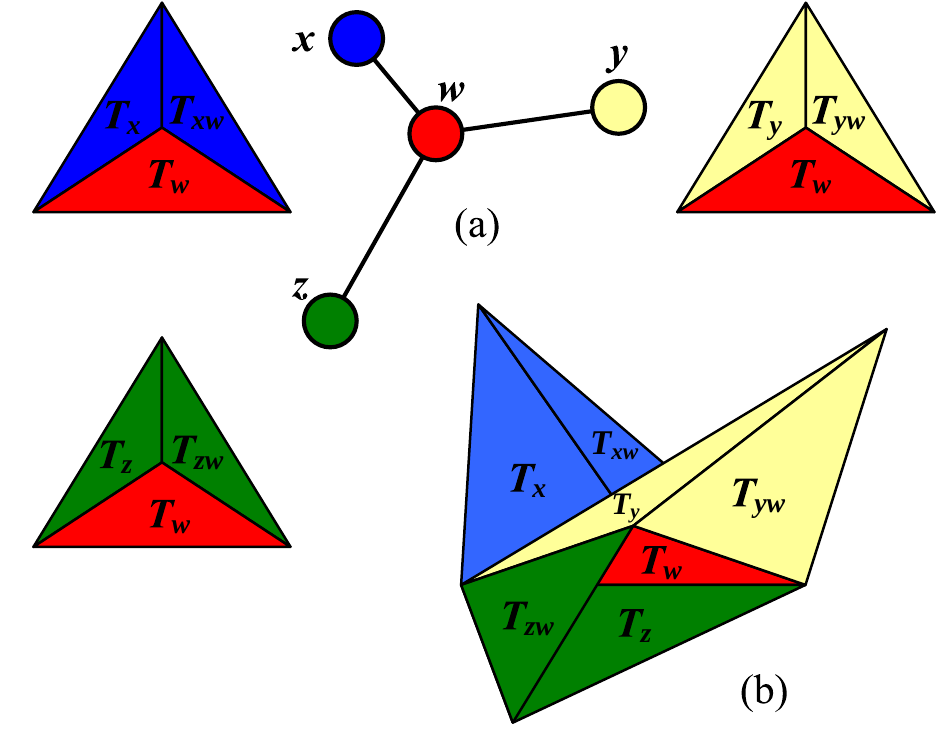}\\
\caption{\label{fig:Triangular-operation-11}{\footnotesize (a) A star $K_4$; (b) $\odot \langle K_4, \Delta(T_i,T_{ij},T_j)\rangle$. }}
\end{figure}

\emph{$\bullet$ Semi-planar graph-labelling.}

A  \emph{semi-maximal planar graph} has a unique \emph{no-triangular face} (is not a triangle), other faces are  triangles; the \emph{bound} of the  unique no-triangular face is denoted as $C$, so we write this semi-maximal planar graph as $G^{C}$   (\cite{Jin-Xu-(1)-2016, Jin-Xu-(2)-2016, Jin-Xu-(3)-2016, Jin-Xu-(4)-2016}).

In Fig.\ref{fig:semi-planer-maximum}, both graphs $G^C$ and $\overline{G}^C$ are two semi-maximal planar graphs, and $G$ (see Fig.\ref{fig:semi-planer-maximum}(c)) is a maximal planar graph obtained by coinciding $G^C$ with $\overline{G}^C$ in one edge by one edge on the cycle $C$. Conversely, we do an edge-split operation  to each edge of the cycle $C$ of $G$, the resulting graphs are just $G^C$ and $\overline{G}^C$ (see Fig.\ref{fig:semi-planer-maximum}(a) and (b)).

We use a total coloring/labelling $F$ to label the vertices and edges of a $(p,q)$-graph $H$ with the elements of a set $S_{emi}(n)$ of semi-maximal planar graphs of orders $\leq n$, such that $$F(uv)=G_{uv}=G^C_{uv}\cup \overline{G}^C_{uv}=F(u)\cup F(v)$$ for each edge $uv\in E(H)$, where $G^C_{uv}, \overline{G}^C_{uv}\in S_{emi}(n)$ (see Fig.\ref{fig:semi-planer-maximum}(d)). We say that the $(p,q)$-graph $H$ admits a \emph{semi-planar graph-labelling} $F$.

Suppose that each semi-maximal planar graph of $S_{emi}(n)$ admits a 4-coloring. Then, we coincide a vertex $x$ of $G^C_{uv}$ with a vertex $x'$ of $G$ if $x$ and $x'$ are colored with the same number, and coincide a vertex $y$ of $\overline{G}^C_{uv}$ with a vertex $y'$ of $G$ if $y$ and $y'$ are colored with the same number. Finally, we get an encrypted network $N_{et}(H, S_{emi}(n))$.

\begin{figure}[h]
\centering
\includegraphics[height=9cm]{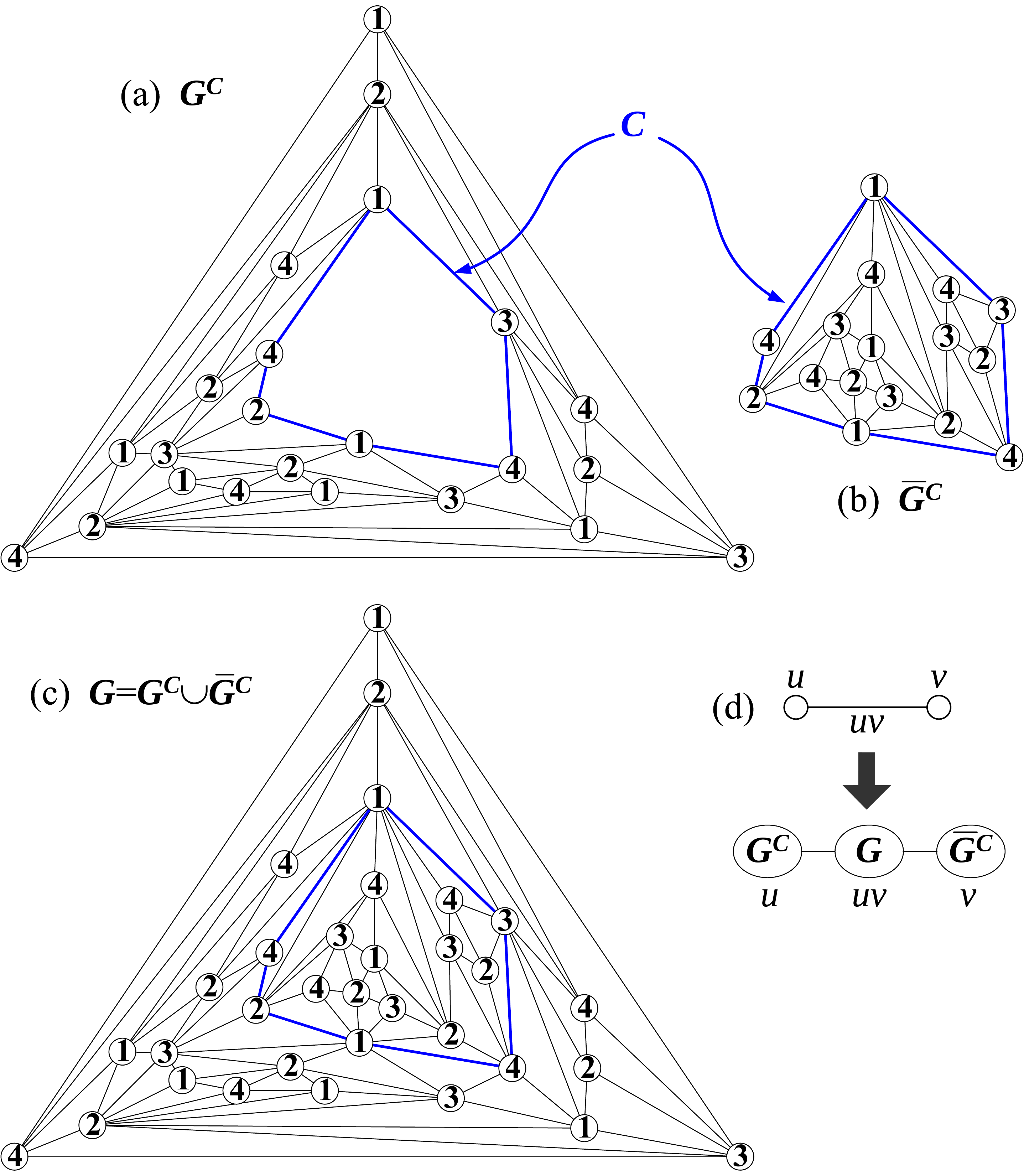}\\
\caption{\label{fig:semi-planer-maximum}{\footnotesize (a) and (b) are two semi-maximal planar graphs $G^C$ and $\overline{G}^C$; (c) a maximal planar graph $G=G^C\cup \overline{G}^C$; (d) an edge $uv$ and its two ends are labelled by $G^C$, $\overline{G}^C$ and $G$. }}
\end{figure}

\section{Further researching problems}

We present the following problems for further researching the translation from Topsnut-gpws to vv-type/vev-type TB-paws. It may be interesting to study these problems, since our researching on them is only a beginning, and the gained results are more junior.

\begin{asparaenum}[FRP-1. ]
\item Determine the number of non-isomorphic caterpillars of $p$ vertices. How many set-ordered odd-graceful/odd-elegant labellings does a caterpillar admit?
\item  Does each lobster admit a multiple edge-meaning vertex labelling defined in Definition \ref{defn:multiple-meanings-vertex-labelling}?
\item The fact of lobsters admitting odd-graceful/odd-elegant labellings were proven by caterpillars admitting set-ordered odd-graceful/odd-elegant labellings, determine set-ordered odd-graceful/odd-elegant labellings admitted by lobsters. Furthermore, how many labellings do lobsters admit?
\item Since a Topsnut-gpw $G$ is a network, find: (i) all possible non-isomorphic caterpillars of $G$; (ii) all possible non-isomorphic lobsters of $G$; (iii) all possible non-isomorphic spanning trees with the maximum number of leaves in $G$; (iv) all possible non-isomorphic generalized sun-graphs (like $T+u_1u_n$ presented in Cycle-neighbor-method) of $G$.
\item Let $P_3\times P_q$ be a lattice in  $xoy$-plan. There are points $(i,j)$ on the lattice $P_3\times P_q$ with $i\in[1,3]$ and $j\in [1,q]$. If a continuous fold-line $L$ with initial point $(a,b)$ and terminal point $(c,d)$ on $P_3\times P_q$ is internally disjoint and contains all points $(i,j)$ of $P_3\times P_q$, we call $L$ a \emph{total TB-paw line}. Find all possible total TB-paw lines. In general, let $\{L_i\}^m_1=\{L_1,L_2,\dots, L_m\}$ be a set of $m$ continuous disjoint fold-lines on $P_3\times P_q$, where each $L_i$ has own initial point $(a_i,b_i)$ and terminal point $(c_i,d_i)$. If $\{L_i\}^m_1$ contains all points $(i,j)$ of $P_3\times P_q$, we call $\{L_i\}^m_1$ a \emph{group of TB-paw lines}, here it is not allowed $(a_i,b_i)=(c_i,d_i)$ for each $L_i$. Find all possible groups $\{L_i\}^m_1$ of TB-paw lines for $m\in [1,3q]$.
\item If each spanning tree of a connected graph $G$ is a lobster, we call $G$ to be \emph{lobster-pure} or a \emph{lobster-graph}. Find the necessary and sufficient conditions for the graph $G$ to be lobster-pure.
\item Adding randomly $k$ leaves to a given tree produces new trees. Enumerate these new trees.
\item Find all multiple-meaning vertex matching partitions $(V,E)$ of $[0,2q-1]$ (see Definition \ref{defn:multiple-meanings-vertex-labelling}). Suppose that $E$ and $V=[0,q-1]$ are two subsets of $[0,2q-1]$, such that each $c\in E$ corresponds $a,b\in V$ to form an \emph{ev-matching} $(acb)$. We call $(V,E)$ a \emph{multiple-meaning vertex matching partition} of $[0,2q-1]$, for each $c\in E$ and its ev-matching $(acb)$, if: (1) $a+c+b=$a constant $k$, and $E=[1,q]$; (2) $a+c+b=$a constant $k'$, and $E=[p,p+q-1]$; (3) $c=a+b~(\bmod~q)$, and $E=[0,q-1]$; (4) $|a+b-c-f(uv)|=$a constant $k''$, and $E=[1,q]$; (5) $c=$an odd number for each $c\in E$ holding $E=[1,2q-1]^o$, and $\{a+c+b:c\in E\textrm{ and its ev-matching }(acb)\}=[\alpha,\beta]$ with $\beta-\alpha+1=q$.
\item Find all graceful-intersection (an odd-graceful-intersection) total set-matching partition of $[1,q]^2$~(or $[1,2q-1]^2)$ (see Definition \ref{defn:graceful-odd-graceful-total-set-labelling}). Let $V$ and $E$ be two subsets of $[1,q]^2$~(or $[1,2q-1]^2)$, such that each set $c\in E$ corresponds two sets $a,b\in V$ to form an \emph{ev-matching} $(acb)$. Suppose that $c=a\cap b$ and $d_c\in c$ is a \emph{representative} of $c$. If $\{d_c:~c\in E\}=[1,q]$ (or $[1,2q-1]^o$), then we call $(V,E)$ a \emph{graceful-intersection (an odd-graceful-intersection) total set-matching partition} of $[1,q]^2$~(or $[1,2q-1]^2)$.
\item For a sequence $\{H_{i_j}\}^q_1$ of an every-zero graphic group $F_{n}(H,h)$, determine: (1) an $\{H_{i_j}\}^q_1$-sequence group-labelling of a tree having $q$ edges; (2) an $\{H_{i_j}\}^q_1$-sequence group-labelling or an $\{H_{i_j}\}^q_1$-sequence group-coloring of a connected $(p,q)$-graph. Does the graph  $H$ in $F_{n}(H,h)$ admit an odd-graceful (an $\{H_{i_j}\}^q_1$-sequence) group-coloring under each self-zero $H_i\in F_{n}(H,f)$ with $\{i_1,i_2,i_3,\dots ,i_{q}\}=[1,2q-1]^o$?
\item Identifying $C_n$ by the non-adjacent identifying operation, how many Euler's graphs can we get?
\item For a tree $T$ of $q$ edges, and an every-zero graphic group $F(H)$ having at least $q$ elements, determine that $T$ admits an $\{H_i\}^q_1$-sequence group-labelling for any $\{H_i\}^q_1$.
\item For $\{S_i\}^q_1$ with $S_i\in [1,q+1]^2$, here it is allowed $S_i=S_j$ for some $i\neq j$, find a vertex labelling $f:V(G)\rightarrow [1,q+1]^2$ of a $(p,q)$-graph $G$, and induces $f(u_iv_i)=f(u_i)\cap f(v_i)=S_i$, where: (1) consecutive sets $S_i=[a_i,b_i]$, $a_{i+1}=a_i+1$, $b_{i+1}=b_i+1$; (2) \emph{Fibonacci sequences} $|S_1|=1$, $|S_2|=1$, and $|S_{i+1}|=|S_{i-1}|+|S_{i}|$; (3) generalized rainbow sequence $S_i=[a,b_i]$ with $b_i<b_{i+1}$, the \emph{regular rainbow sequence} $S_i=[1,i]$ with $i\in [2,q]$; (4) $|S_i|=i$, where $S_i=\{a_{i,1}, a_{i,2}, \dots a_{i,i}\}$; (5) $f:V(G)\rightarrow [1,q+1]^2$, $f(uv)=f(u)\cup f(v)$, and $f(V(G))\cup f(E(G))=[1,N]^2$ with $N\leq q+1$.
\item For a simple and connected graph $H$, determine its v-split connectivity $\gamma_{vs}(H)$ and its e-split connectivity $\gamma_{es}(H)$. Characterize connected graphs having e-split $k$-connectivity.
\item If $H$ admits some odd-graceful/odd-elegant labelling $f$, we ask: (1) Does it admit some perfect odd-graceful/odd-elegant labellings? (2) How many graphs $G$ matching with $H$ are there, where $G$ admits odd-graceful/odd-elegant labelling $g$ such that $g(E(G))=f(V(H))\setminus X^*$ and $g(V(G))\setminus X^*=f(E(H))$ for $X^*=g(V(G))\cap f(V(H))$?
\item Is a spanning tree having maximum leaves the same as a spanning tree having the shortest diameter in a network? Or characterize them two.
\item Given a matrix $A_{3\times q}$ with integer elements, by what condition $A_{3\times q}$ is a Topsnut-matrix of some Topsnut-gpw $G$?
\item Given a string $D$ with positive integers, how to construct a matrix $A_{3\times q}$ by $D$ such that $A_{3\times q}$ is just a Topsnut-matrix of some Topsnut-gpw $G$?
\item Applying random walks, Markov chains to Topsnut-gpws for generating random vv-type/vev-type TB-paws.
\item A spider with three legs of length 2 is denoted as $A_{2,2,2}$. If each spanning tree of a graph $G$ is a caterpillar, we call $G$ to be \emph{caterpillar-pure}, and $G$ is a \emph{caterpillar-graph}. Jamison \emph{et al.}                                \cite{Jamison-McMorris-Mulder-2003} have shown: ``\emph{A connected graph is caterpillar-pure if and only if it does not contain any aster $A_{2,2,2}$ as a (not necessarily induced) subgraph}''. We hope to describe caterpillar-graphs in detail and concrete.
\item Suppose that two trees $T,H$ of $p$ vertices are 6C-complementary matching to each other. We use two every-zero graphic groups $F_n(T,f)$ and $F_n(H,g)$ to encrypt a connected $(p,q)$-graph $G$ respectively, and we get two encrypted networks $N_{et}(G,F_n(T,f))$ and $N_{et}(G,F_n(H,g))$. Does $F_n(T,f)$ matches with $F_n(H,g)$? And, moreover does $N_{et}(G,F_n(T,f))$ matches with $N_{et}(G,F_n(H,g))$?
\item For a given non-tree $(p,q)$-graph $G$ admitting odd-graceful labellings $f_1,f_2,\dots ,f_m$ in total, find all possible twin odd-graceful matchings of $G$.
\item Topsnut-gpw sequences $\{G_{(k_i,d_i)}\}^m_1$ can encrypt graphs/networks. Determine what graphs/networks can be encrypted by what Topsnut-gpw sequences $\{G_{(k_i,d_i)}\}^m_1$.
\item Collect possible image-labellings about labellings of graph theory, and determine graphs admitting these image-labellings.
\item Define new graph-labellings, like the \emph{edge-magic total graph-labelling} defined in Definition \ref{defn:Edge-magic-total-graph-labelling}, and determine graphs/networks admitting these new graph-labellings.
\end{asparaenum}

\section{Conclusion}

We have found the number of all Topsnut-matrices $A_{vev}(G)$ of a $(p,q)$-graph $G$ and the number $D_{TBp}(G)$ of vv-type/vev-type TB-paws generated from the Topsnut-matrices. Based on Topsnut-configurations, we have shown Path-neighbor-method, Cycle-neighbor-method, Lobster-neighbor-method, Spider-neighbor-method and Euler-Hamilton-method for generating vv-type/vev-type TB-paws by providing efficient and polynomial algorithms. We use Topsnut-matrices to make vv-type/vev-type TB-paws, the results on this method indicate that Topsnut-matrices are powerful in deriving vv-type/vev-type TB-paws, since there are many random ways in Topsnut-matrices. Another important method is to encrypt a network by an every-zero graphic group, and we have list advantages about encrypted networks. It is noticeable, our methods for generating vv-type/vev-type TB-paws can transformed into algorithms, we, also, introduce the LOBSTER-algorithm and the TREE-GROUP-COLORING algorithm.

Our algorithms enables us to transform Topsnut-gpws made by caterpillars, lobsters and spiders, as well as generalized trees on them, into TB-paws. It is noticeable, the complexity of encrypted networks by every-zero graphic groups tells the provable security of encrypted networks, especially, scale-free tree-like networks in which a few number of vertices control other vertices over $80$ percenter of networks. We discussed every-zero graphic group $F_n(H,f)=\{H_i\}^n_1$, every-zero Topsnut-matrix group $M_{n}(A(G),f)$, every-zero TB-paw group $D_{n}(T,f)$ and composed graphic group $F_n(\odot \langle H_i,L_i\rangle, f_i\odot h_i)$ made by two every-zero graphic groups $F_n(H,f)=\{H_i\}^n_1$ and $F_n(L,h)=\{L_i\}^n_1$. Moreover, Topsnut-gpw sequences $\{G_{(k_i,d_i)}\}^m_1$ can encrypt graphs/networks.

Topsnut-gpws are based on the \emph{open structural cryptographic platform}, that is, this platform allows people make themselves pan-Topsnut-gpws by their remembered and favorite knowledge. We believe: ``\emph{If a project has its practical and effective application, and has mathematics as its support, it can go far. The practical application gives it long life, and mathematics makes it growing stronger and going faster. This project feedbacks material comforts to people, and returns new objects and new problems to mathematics}.''

\section*{Acknowledgment}

The author, \emph{Bing Yao}, is delight for supported by the National Natural Science Foundation of China under grants 61163054, 61363060 and 61662066; Scientific research project of Gansu University under grants 2016A-067, 2017A-047 and 2017A-254. Bing Yao, also, thanks every member of \emph{Topological Graphic Passwords Symposium} in the first semester of 2018-2019 academic year for their constructive suggestions and hard works.




\begin{thebibliography}{1}
\bibitem{A-L-Barabasi-R-Albert1999}A.-L. Barab\'{a}si, and R. Albert. Emergence of
scaling in random networks. \emph{Science}, \textbf{286} (1999), 509-512.
\bibitem{Barabasi-Bonabeau2003}A.-L. Barab\'{a}si, and E. Bonabeau. Scale-Free Networks.
Scientific American, \textbf{288} (2003), 60-69.
\bibitem{Bondy-2008} J. A. Bondy, U. S. R. Murty. Graph Theory. Springer London, 2008.
\bibitem{Gallian2016} Joseph A. Gallian. A Dynamic Survey of Graph Labeling. \emph{The electronic journal of
combinatorics}, \textbf{17} (2016), \# DS6.
\bibitem{Suo-Zhu-Owen-2005} Xiaoyuan Suo, Ying Zhu, G. Scott. Owen. Graphical Password: A Survey. In: Proceedings of Annual
Computer Security Applications Conference (ACSAC), Tucson, Arizona. IEEE (2005) 463-472.

\bibitem{Biddle-Chiasson-van-Oorschot-2009}R. Biddle, S. Chiasson, and P.C. van Oorschot. Graphical passwords: Learning from the First Twelve Years. ACM Computing Surveys, \textbf{44} (4), Article 19:1-41. Technical Report TR-09-09, School of Computer Science, Carleton University, Ottawa, Canada. 2009.

\bibitem{Gao-Jia-Ye-Ma-2013}Haichang Gao, Wei Jia, Fei Ye and Licheng Ma. A Survey on the Use of Graphical Passwords in Security. Journal Of Software, Vol. \textbf{8} (7), July 2013, 1678-1698.

\bibitem{Ma-Bing-Yao-Physica-A-2017}Fei Ma, Bing Yao. The relations between network-operation and topological-property in a scale-free and small-world network with community structure. Physica A (2017). https://dx.doi.org/10.1016/j.physa.2017.04.135
\bibitem{Ma-Bing-Yao-Computer-2018}Fei Ma, Bing Yao. An iteration method for computing the total number of spanning trees and its applications in graph theory. Theoretical Computer Science (2018). https://doi.org/10.1016/j.tcs.2017.10.030
\bibitem{Ma-Bing-Yao-Physica-A-2018}Fei Ma, Bing Yao. A Family of Deterministic Small-world Network Models Built by Complete Graph and Iteration-function. Physica A (2018). https://doi.org/10.1016/j.physa.2017.11.136
\bibitem{Ma-Wang-Wang-Yao-Theoretical-Computer-Science-2018}Fei Ma, Ding Wang, Ping Wang, Bing Yao. The First Handshake Between Fibonacci Series And ``Pure'' Preferential Attachment Mechanism On A Graph Model. submitted to Theoretical Computer Science, 2018.
\bibitem{Ma-Su-Hao-Yao-Physica-A-2018}Fei Ma, Jing Su, Yongxing Hao, Bing Yao. A class of vertex-edge-growth small-world network models having scale-free, self-similar and hierarchical characters. Physica A (2018). https://doi.org/10.1016/j.physa.2017.11.047

\bibitem{Sun-Zhang-Yao-IAEAC-2017}Hui Sun, Xaohui Zhang, Bing Yao. On Operation Phenomena In Networks With Hub-Rings. 2017 IEEE 2nd Advanced Information Technology, Electronic and Automation Control Conference (IEEE IAEAC 2017), 82-85.
\bibitem{Hui-Sun-Bing-Yao2018}Hui Sun, Bing Yao. New Graph Labellings Of Euler's Graphs For Designing Topological Graphic Passwords. submitted (2018).
\bibitem{Sun-Zhang-Zhao-Yao-2017}Hui Sun, Xiaohui Zhang, Meimei Zhao and Bing Yao. New Algebraic Groups Produced By Graphical Passwords Based On Colorings And Labellings. ICMITE 2017, MATEC Web of Conferences \textbf{139}, 00152 (2017), DOI: 10.1051/matecconf/201713900152
\bibitem{SH-ZXH-YB-2017-3}Sun H, Zhang X, Yao B. Strongly $(k,d)$-graphical Labellings For Designing Graphical Passwords In Communication. ICMITE 2017, MATEC Web of Conferences 139, 00204 (2017) DOI: 10.1051/matecconf/201713900204
\bibitem{SH-ZXH-YB-2017-4}Hui Sun, Xiaohui Zhang, Bing Yao. Construction Of New Graphical Passwords With Graceful-type Labellings On Trees. 2018 2nd IEEE Advanced Information Management, Communicates,Electronic and Automation Control Conference(IMCEC 2018), 1491-1494.
\bibitem{SH-ZXH-YB-2017-2}Hui Sun, Xiaohui Zhang, Bing Yao. New Graphical Passwords On Trees Having Perfect Matchings. submitted (2017).
\bibitem{SH-ZXH-YB-2018-5}Hui Sun, Jing Su, Xiaohui Zhang, Bing Yao. New Graphic Cryptography Made By Strongly Graceful-Type Labellings Towards Communication Security. submitted (2018).

\bibitem{Wang-Xu-Yao-2016} Hongyu Wang, Jin Xu, Bing Yao. Exploring New Cryptographical Construction Of Complex Network Data. IEEE First International Conference on Data Science in Cyberspace. IEEE Computer Society, (2016) 155-160.
\bibitem{Wang-Xu-Yao-Key-models-Lock-models-2016}Hongyu Wang, Jin Xu, Bing Yao. The Key-models And Their Lock-models For Designing New Labellings Of Networks.Proceedings of 2016 IEEE Advanced Information Management, Communicates, Electronic and Automation Control Conference (IMCEC 2016) 565-5568.
\bibitem{Wang-Xu-Yao-Ars-2018}Hongyu Wang, Jin Xu, Bing Yao. On Generalized Total Graceful labellings of Graphs. Ars Combinatoria, Vol. \textbf{139} July 2018.
\bibitem{Wang-Xu-Yao-2017-Twin}Hongyu Wang, Jin Xu, Bing Yao. Twin Odd-Graceful Trees Towards Information Security. Procedia Computer Science \textbf{107} (2017)15-20, DOI: 10.1016/j.procs.2017.03.050
\bibitem{Wang-Xu-Yao-2017}Hongyu Wang, Jin Xu, Bing Yao. Odd-elegant Matching Trees In Planning New-type Graphical Passwords. submitted (2017).
\bibitem{Hongyu-Wang-2018-Doctor-thesis}Hongyu Wang. The Structure And Theoretical Analysis On Topological Graphic Passwords. Doctor's thesis. School of Electronics Engineering and Computer Science, Peking University, 2018.

\bibitem{Jin-Xu-(1)-2016}Jin Xu. Theory on Structure and Coloring of Maximal Planar Graphs: (1) Recursion Formulae of Chromatic Polynomial and Four-Color Conjecture. Journal of Electronics and Information Technology. Vol.38 No.4, Jul. 2016, 763-770.
\bibitem{Jin-Xu-(2)-2016}Jin Xu. Theory on Structure and Coloring of Maximal Planar Graphs: (2) Domino Configurations and Extending-Contracting Operations. Journal of Electronics and Information Technology. Vol.38 No.6, Jul. 2016, 1271-1327.
\bibitem{Jin-Xu-(3)-2016}Jin Xu. Theory on Structure and Coloring of Maximal Planar Graphs: (3) Purely Tree-colorable and Uniquely 4-colorable Maximal Planar Graph Conjectures. Journal of Electronics and Information Technology. Vol.38 No.6, Jul. 2016, 1329-1353.
\bibitem{Jin-Xu-(4)-2016}Jin Xu. Theory on Structure and Coloring of Maximal Planar Graphs: (4) $\sigma$-Operations and Kempe Equivalent Classes. Journal of Electronics and Information Technology. Vol.38 No.7, Jul. 2016, 1558-1585.







\bibitem{Yao-Sun-Zhao-Li-Yan-2017}Bing Yao, Hui Sun, Meimei Zhao, Jingwen Li, Guanghui Yan. On Coloring/Labelling Graphical Groups For Creating New Graphical Passwords. (ITNEC 2017) 2017 IEEE 2nd Information Technology, Networking, Electronic and Automation Control Conference.(2017) 1371-1375.
\bibitem{Yao-Mu-Sun-Zhang-Wang-Su-2018} Bing Yao, Yarong Mu, Hui Sun, Xiaohui Zhang, Hongyu Wang, Jing Su. Connection Between Text-based Passwords and Topological Graphic Passwords. 2018, submitted
\bibitem{Yao-Liu-Yao-2017}Bing Yao, Xia Liu and Ming Yao. Connections between labellings of trees. Bulletin of the Iranian Mathematical Society, ISSN: 1017-060X (Print) ISSN: 1735-8515 (Online), Vol. \textbf{43} (2017), 2, pp. 275-283.
\bibitem{Yao-Sun-Zhang-Mu-Sun-Wang-Su-Zhang-Yang-Yang-2018arXiv}Bing Yao, Hui Sun, Xiaohui Zhang, Yarong Mu, Yirong Sun, Hongyu Wang, Jing Su, Mingjun Zhang, Sihua Yang, Chao Yang. Topological Graphic Passwords And Their Matchings Towards Cryptography. arXiv:1808.03324v1 [cs.CR] 26 Jul 2018.
\bibitem{Bing-Yao-Cheng-Yao-Zhao2009}Bing Yao, Hui Cheng, Ming Yao and Meimei Zhao. A Note on Strongly
Graceful Trees. Ars Combinatoria \textbf{92} (2009), 155-169.
\bibitem{Yao-Mu-Sun-Zhang-Su-Ma-Wang2018} Bing Yao, Yarong Mu, Hui Sun, Xiaohui Zhang, Jing Su, Fei Ma, Hongyu Wang. Algebraic Groups For Construction Of Topological Graphic Passwords In Cryptography. submitted (2018)
\bibitem{Yao-Zhang-Sun-Mu-Wang-Zhang2018}Bing Yao, Xiaohui Zhang, Hui Sun, Yarong Mu, Hongyu Wang, Mingjun Zhang. On Space and Design of Topological Graphic Passwords. 2018 submitted
\bibitem{YAO-SUN-WANG-SU-XU2018arXiv}Bing Yao, Hui Sun, Hongyu Wang, Jing Su, Jin Xu. Graph Theory Towards New Graphical Passwords In Information Networks. arXiv:1806.02929v1 [cs.CR] 8 Jun 2018
\bibitem{Yao-Chen-Yang-Wang-Zhang-Zhang2012}B. Yao, X.-E Chen, C. Yang, H.-Y Wang, J.-J Zhang, X.-M Zhang. Spanning Trees And Dominating Sets In Scale-Free Networks. Proceeding of 2012 IET International Conference on Information Science and Control Engineering (ICISCE 2012), December 2012, Shenzhen, China. 111-115.
\bibitem{Yao-Zhang-Wang-2010}Bing Yao, Zhongfu Zhang and Jianfang Wang. Some results on spanning trees. Acta Mathematicae Applicatae Sinica, English Series, 2010, 26(4).607-616. DOI:10.1007/s10255-010-0011-4

\bibitem{Wang-Zhang-Mei-Yao2018-Split}Xiaomin Wang,  Wei Zhang,  Hong Mei  and Bing Yao. On Split-type Connectivity of Graphs. submitted 2018.

\bibitem{Zhou-Yao-Chen-Tao2012}Xiangqian Zhou, Bing Yao, Xiang'en Chen and Haixia Tao. A proof to the
odd-gracefulness of all lobsters. \emph{Ars Combinatoria} \textbf{103} (2012), 13-18.
\bibitem{Zhou-Yao-Chen-2013}Xiangqian Zhou, Bing Yao, Xiang'en Chen. Every lobster is odd-elegant. Information Processing Letters, \textbf{113},1-2(2013) 30-33.
\bibitem{Jamison-McMorris-Mulder-2003}Robert E. Jamison, F.R. McMorris, Henry Martyn Mulder. Graphs with only caterpillars as spanning trees. Discrete Mathematics \textbf{272} (2003) 81-95.
\bibitem{Zhang-Zhou-Fang-Guan-Zhang-2007}Zhongzhi Zhang, Shuigeng Zhou, Lujun Fang, Jihong Guan and Yichao Zhang. EPL, 2007,
79: 38007.
\bibitem{Fernau-Kneis-Kratsch-Langer-Liedloff-Raible2011}H. Fernau, J. Kneis, D. Kratsch, A. Langer, M. Liedloff, D. Raible, and P. Rossmanith. An exact algorithm for the Maximum Leaf Spanning Tree problem. Theoretical Computer Science, \textbf{412}(45) (2011), 6290-6302.

\bibitem{Garey-Johnson1979}M. R. Garey, and D. S. Johnson. Computers and Intractability. A Guide to the Theory of NP-Completeness. W. H. Freeman and Company, New York, 1979.

\bibitem{Douglas-Robert-J1992}Douglas, and Robert J. NP-completeness and degree restricted spanning trees. Discrete Mathematics \textbf{105} (1-3)(1992), 41-47.

\bibitem{Blum-Ding-Thaeler-Cheng2004}J. Blum, M. Ding, A. Thaeler, and X. Cheng. Connected dominating set in sensor networks and manets. Handbook of Combinatorial Optimization (2004) 329-369.


\bibitem{Unterschutz-Turau2012}S. Untersch\"{u}tz, and V. Turau. Construction of Connected Dominating Sets in Large-Scale MANETs Exploiting Self-Stabilization. http://www.ti5.tu-harburg.de/research/projects/heliomesh/
\bibitem{Zhang-Sun-Yao-ICMITE2017}Xiaohui Zhang, Hui Sun and Bing Yao. Graph Theory Towards Module-K Odd-Elegant Labelling Of Graphical Passwords. ICMITE 2017, MATEC Web of Conferences 139, pp1640-1644. 00206 (2017) DOI: 10.1051/matecconf/201713900206
\bibitem{ZHANG-MU-SUN-YAO-IAEAC2018}Xiaohui ZHANG,Yarong MU,Hui SUN, Bing YAO. Graph Module-k Odd-elegant Labelling Towards Topological Graphical Passwords.2018 IEEE IAEAC 2018, submitted.
\end{thebibliography}
%

\vskip 1cm

\begin{flushleft}
\textbf{Appendix A. } \textbf{LARGEDEGREE-NEIGHBOR-FIRST Algorithm}
\end{flushleft}

Let $N(X)$ and $N(u)$ be the sets of neighbors of a vertex $u$ and a set $X$. A vertex set $S$ of a $(p,q)$-graph $G$ is called a \emph{dominating set} of $G$ if each vertex $x\in V(G)\setminus S$ is adjacent with some vertex $y\in S$, and moreover the dominating set $S$ is \emph{connected} if the induced graph over $S$ is a connected subgraph of $G$.

\vskip 0.4cm

\textbf{Input.} A connected and simple graphs $G=(V, E)$.

\textbf{Output. } A spanning tree and a connected dominating set of $G$.

\textbf{Step 1.} Let $S_1:=N(u_1)\cup \{u_1\}$, $\ud_G(u_1)=\Delta(G)$, and $T_1$ is a tree with vertex set $S_1$, $k:=1$.

\textbf{Step 2.} If $Y_k=V\setminus (S_k\cup N(S_k))\neq \emptyset$, goto Step 3, otherwise Step 4.

\textbf{Step 3.} Select vertex $u_{k+1}\in L(T_k)$ holds $\ud_G(u_{k+1})\geq \ud_G(x)$ ($x\in L(T_k)$), and let $S_{k+1}:=S_k\cup N^*(u_{k+1})$, where $N^*(u_{k+1})=N(u_{k+1})\setminus (N(u_{k+1})\cap S_k)$, $T_{k+1}:=T_k+\{u_{k+1}u':u '\in N^*(u_{k+1})\}$, $k:=k+1$, goto Step 2.

\textbf{Step 4.} For $Y_k=\emptyset$, $y\in V\setminus V(T_{k+1})$, do: $y$ is adjacent wit $v\in V(T_{k+1})$ when $\ud_{T_{k+1}}(v)\geq \ud_{T_{k+1}}(x)$, $xy\in E$. The resulting tree is denoted as $T^*$.

\textbf{Step 5.} Return a connected dominating set $S_k=V(T_k)$ and the spanning tree $T^*$.

\begin{flushleft}
\textbf{Appendix B. } \textbf{PREDEFINED-NODES Algorithm}
\end{flushleft}

\textbf{Input.} A connected graph $G=(V, E)$, and indicate a subset $S=\{u_1, u_2, \dots , u_m\}$ of $G$.

\textbf{Output. } A connected dominating set $X$ of $G$ such that $S\subseteq X$.

Step 1. Add new vertices $\{v_1, v_2, \dots , v_m\}$ to $G$, and join
$v_i$ with a vertex $u_i$ of $S$ by an edge, $i\in [1, m]$. The resulting graph is denoted as$G^*$, such that $V(G^*)=V(G)\cup \{v_i:i\in [1, m]\}$ and $E(G^*)=E(G)\cup \{u_iv_i:i\in [1, m]\}$.

Step 2. Find a connected dominating set $X'$ of $G^*$.

Step 3. Return a connected dominating set $X=X'$ of $G$.

\begin{flushleft}
\textbf{Appendix C. } \textbf{LARGEDEGREE-PRESERVE Algorithm}
\end{flushleft}

\textbf{Input.} A scale-free network $\mathcal {N}(t_0)=(p(u, k, t_0), G(t_0))$. $G=G(t_0)$ with $n$ vertices, and $\ud_G(v_i)\geq \ud_G(v_{i+1})$, $i=1, 2, \dots , n-1$, $\Delta(G)=\ud_G(v_1)>1$. A constant $k$ satisfies $\delta(G)<k<\Delta(G)$, $\ud_G(v_l)\geq k$, but $\ud_G(v_{l+1})<k$.

\textbf{Output. } A spanning tree $T^*$ of $G$, such that the vertices of $T^*$ hold $\ud_G(v_i)\geq \ud_G(v_{i+1})$ ($i=1, 2, \dots , l-1$), and $k> \ud_G(v)$, $v\in V(G)\setminus \{v_1, v_2, \dots , v_l\}$.

Step 1. Let $W_1=N[v_1]:=N(v_1)\cup \{v_1\}$, an induced graph $G[W_1]$ over $W_1$.

Step 2. If $N_{i, i+1}=N[v_{i+1}]\cap W_{i}=\emptyset $, let
$W_{i+1}:=W_i\cup N[v_{i+1}]$, and an induced graph $G[W_{i+1}]$; if $N_{i,
i+1}\neq \emptyset $, take a vertex $x_{i, i+1}\in N(v_{i+1})\cap
W_{i}$ with property $\ud_G(x_{i, i+1})\geq \ud_G(x)$
($x\in N(v_{i+1})\cap W_{i}$), and an induced graph $G[W_{i+1}]$, where $W_{i+1}:=W_{i}\cup W_{i, i+1}$
and $W_{i, i+1}:=(N[v_{i+1}]\setminus N_{i, i+1})\cup \{x_{i,i+1}\}$.

Step 3. If $\ud_G(v_i)\geq k$, goto Step 2, and goto Step 4, otherwise.

Step 4. Apply modified BFS-algorithm (Breadth-First Search Algorithm). Let $S:=W_i$, $R:=\{v_0\}$ for $v_0\in N(W_i)$, $\ud(v_0, v_0):=0$.

Step 5. If $R=\emptyset$, denoted the found spanning tree as $T^*$, goto Step 7. If $R\neq \emptyset$, goto Step 6.

Step 6. The vertex $v$ is the first vertex of $R$, take $y\in N(v)\setminus(R\cup S)$ holding $\ud _G(y)\geq \ud _G(x)$ ($x\in N(v)$); put $y$ into $R$, such that is the last of $R$; and take $v$ from $R$, and then put $v$ into $S$.
Lt $\ud (v_0, y):=\ud (v_0, v)+1$, goto Step 5.

Step 7. Return the spanning tree $T^*$.

\vskip 0.4cm

The spanning tree $T^*$ found by LARGEDEGREE-PRESERVE Algorithm has its own number of leaves to be approximate to the number of leaves of each spanning tree $T^{\max}$ having maximal leaves, and the connected dominating set $S^*=V(T^*)\setminus L(T^*)$ approximates to $D^+$-minimal dominating set \cite{Yao-Chen-Yang-Wang-Zhang-Zhang2012}. As we have known, no polynomial algorithm for finding: (i) $L^+$-balanced set; (ii) optimal cut set $H[S^-]$; (iii) $D^-$-minimal dominating set and $D^+$-minimal dominating set; (iv) $k$-distance dominating set. Untersch\"{u}tz and Turau \cite{Unterschutz-Turau2012} have shown the \emph{probabilistic self-stabilizing algorithm} (PSS-algorithm) for looking connected dominating set (\cite{Blum-Ding-Thaeler-Cheng2004}). PSS-algorithm is suitable large scale of networks, especially good for those networks having larger degree vertices. PSS-algorithm consists of three subprogrammes: Finding maximal independent set (MIS) first, and find weak connected set (WCDS), the last step is for finding connected dominating set (CDS).

\end{document}